\documentclass[floats,12pt,onecolumn]{article}
\usepackage{hyperref}
\usepackage{amsfonts}
\usepackage{amsmath}
\usepackage{epsfig}
\usepackage{pstricks}
\usepackage{calc}
\usepackage{color}
\usepackage{subfigure} 
\usepackage{graphics}
\usepackage{enumerate}
\usepackage{amssymb}
\usepackage{amsthm}
\usepackage{psfrag}
\usepackage[noblocks]{authblk}
\usepackage{caption}
\usepackage{dblfloatfix}

\usepackage{csquotes}

\usepackage{mathtools}
\DeclarePairedDelimiter{\set}{\{}{\}}
\DeclarePairedDelimiterX{\setcomprehension}[2]{\{}{\}}{#1\mathrel{\delimsize\vert} #2}

\usepackage{float}
\restylefloat{table}

\newcommand*{\bbA}{\mathbb{A}}
\newcommand*{\bbZ}{\mathbb{Z}}

\newcommand*{\cA}{\mathcal{A}} 
\newcommand*{\cB}{\mathcal{B}}
\newcommand*{\cC}{\mathcal{C}}

\newcommand*{\cG}{\mathcal{G}}
\newcommand*{\cH}{\mathcal{H}}

\newcommand*{\cD}{\mathcal{D}}

\newcommand*{\cR}{\mathcal{R}}

\newcommand*{\cP}{\mathcal{P}}

\newcommand*{\cT}{\mathcal{T}}

\newcommand*{\bPi}{{\mathbf{\Pi}}}
\newcommand*{\bD}{{\mathbf{D}}}

\newcommand*{\bU}{{\mathbf{U}}}

\newcommand*{\bV}{{\mathbf{V}}}
\newcommand*{\bI}{{\mathbf{I}}}
\newcommand*{\bW}{{\mathbf{W}}}

\newcommand*{\mcg}{{\mathsf{MCG}}}

\newcommand*{\bUt}{{\bf \tilde{U}}}
\newcommand*{\Ut}{ \tilde{U}}

\newtheorem{theorem}{Theorem}[section]
\newtheorem{lemma}[theorem]{Lemma}
\newtheorem{proposition}{Proposition}[section]
\newtheorem{postulate}[theorem]{Postulate}

\newtheorem{definition}[theorem]{Definition}

\newtheorem{example}[theorem]{Example}
\newtheorem{corollary}[theorem]{Corollary}

\usepackage{amsfonts}

\newcommand{\ket}[1]{|#1\rangle}

\newcommand{\proj}[1]{|#1\rangle\langle#1|}

\newcommand{\comment}[1]{}
\newcommand{\diag}{\mbox{diag}}

\newcommand*{\phiperm}{\Delta}

\newcommand*{\Smatrix}[2]{\raisebox{-1.7ex}{
\begin{picture}(39,28)(-2,-4)
\put(0,-6){\includegraphics[scale=.5]{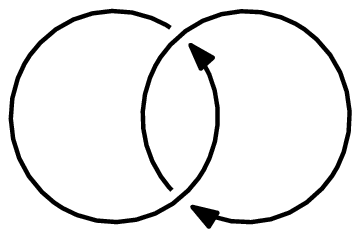}}
\put(-7,5){\small $#1$}
\put(52,5){\small $#2$}
\end{picture}}}

\newcommand*{\connected}[1]{\Leftrightarrow_{#1}}

\newcommand*{\phaseExp}[2]{\mathrm{e}^{\mathrm{i}\varphi{#1}(#2)}}
\newcommand*{\phaseArg}[2]{\varphi{#1}(#2)}
\newcommand*{\phase}[2]{\lambda_{#2}#1}
\newcommand*{\phaseGlobal}{\lambda}

\newcommand*{\perms}{\vec{\pi}}
\newcommand*{\permLoop}[1]{\pi^{#1}}

\newcommand*{\unitaryGroup}[1]{\mathrm{U}(#1)}

\newcommand*{\labeling}[1]{#1}
\newcommand*{\labelSequenceSphere}[1]{#1}

\newcommand*{\dual}[1]{\overline{#1}}

\newcommand*{\conjugate}[1]{\overline{#1}}

\begin{document}

\newcommand*{\id}{\mathsf{id}}

\newcommand*{\f}{\mathbf{f}}
\newcommand*{\p}{\mathbf{p}}
\newcommand*{\q}{\mathbf{q}}
\newcommand*{\Ver}{\mathsf{Ver}}

\newcommand*{\four}[5]{\begin{pspicture}[shift=-0.350](0.000,0.000)(3.350,0.700)
\put(0.100,0.000){\scalebox{1.000}{\epsfig{file=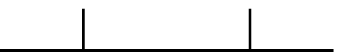}}}
\put(0.800,0.600){$#1$}
\put(0.100,0.200){$#2$}
\put(1.200,0.200){$#5$}
\put(2.850,0.200){$#3$}
\put(2.500,0.600){$#4$}
\end{pspicture}}

\newcommand*{\fourdot}[4]{\four{#1}{#2}{#3}{#4}{\quad \   \cdot}}
\newcommand{\Fset}[2]{\mathsf{Iso}\left(#1  \  \rightarrow \   #2\ \right)}
\newcommand{\fset}[3]{\mathsf{Iso}\left(#1\overset{#3}{\rightarrow} #2\right)}

\newcommand*{\logical}[1]{[#1]}
\newcommand*{\Aut}{\mathsf{Aut}}
\newcommand*{\labels}{\bbA}
\newcommand*{\pauli}{\mathsf{Pauli}}
\newcommand*{\cliff}{\mathsf{Clifford}}

\newcommand{\z}[1]{\tilde{#1}}

\title{Protected gates  for topological quantum field theories}

\author[1]{Michael E. Beverland} 
\author[2]{Oliver Buerschaper}
\author[3]{Robert Koenig}
\author[1]{Fernando Pastawski}
\author[1]{John Preskill}
\author[4]{Sumit Sijher}
\affil[1]{Institute for Quantum Information \& Matter, California Institute of Technology,  Pasadena CA 91125, USA}

\affil[2]{Dahlem Center for Complex Quantum Systems, Freie Universit\"at Berlin, 14195 Berlin, Germany }
  
 \affil[3]{Institute for Advanced Study \& Zentrum Mathematik, Technische Universit\"at M\"unchen, 85748 Garching, Germany}

 \affil[4]{Institute for Quantum Computing \& Department of Applied Mathematics, University of Waterloo, Waterloo, ON N2L 3G1, Canada}
\maketitle

\begin{abstract}
We study restrictions on locality-preserving unitary logical gates for topological quantum codes in two spatial dimensions.  A locality-preserving operation is one which maps local operators to local operators --- for example, a  constant-depth quantum circuit of geometrically local gates, or evolution for a constant time governed by a geometrically-local bounded-strength Hamiltonian. Locality-preserving logical gates of topological codes are intrinsically fault tolerant because spatially localized errors remain localized, and hence sufficiently dilute errors remain correctable. By invoking general properties of two-dimensional topological field theories, we find that the locality-preserving logical gates are severely limited for codes which admit non-abelian anyons; in particular, there are no locality-preserving logical gates on the torus or the sphere with~$M$ punctures if the braiding of anyons is computationally universal. Furthermore, for Ising anyons on the $M$-punctured sphere, locality-preserving gates must be elements of the logical Pauli group. We derive these results by relating logical gates of a topological code to automorphisms of the Verlinde algebra of the corresponding anyon model, and by requiring the logical gates to be compatible with basis changes in the logical Hilbert space arising from local $F$-moves and the mapping class group.
\end{abstract}

\tableofcontents
\section{Introduction}
In order to reliably compute, it is necessary to protect information 
against  noise. For quantum computations, this is particularly challenging because
noise in the form of decoherence threatens the very quantum nature of the process.
Adding redundancy by encoding information into a quantum error-correcting code
is a natural, conceptually appealing approach towards building noise-resilient
scalable computers based on imperfect hardware.

Among the known quantum error-correcting codes, the class of so-called topological codes stands out. Examples in 2D include 
the toric code and quantum double models~\cite{Kitaev03}, the surface codes~\cite{BK:surface},  the 2D color codes~\cite{BMD:topo}, variants of these codes~\cite{BombinTwists,Fowler08}, and the Levin-Wen model~\cite{LevinWen}. In 3D, known examples are  Bombin and Martin-Delgado's  3D~color code ~\cite{Bombin06}, as well as Haah's~\cite{Haah11} and Michnicki's~\cite{Michnicki12} models. 
These codes are attractive  for a number of reasons: their code space is topologically protected, meaning that small local deformations or locally acting noise do not affect encoded information. The degree of this protection (measured in information-theoretic notions in terms of code distance, and manifesting itself in physical properties such as gap stability) scales with the system size: in other words, robustness essentially reduces to the question of  scalability.  Finally, the code space of a topological code is the degenerate ground space of a geometrically local Hamiltonian: this means that
syndrome information can be extracted by local measurements, an important feature for actual realizations.  Furthermore, this implies that a topological code  is essentially a  phase of a many-body system and can be characterized in terms of its particle content, their statistics, and the quantum field theory emerging in the continuum limit.  In particular, the quantum field theory provides a description of such systems which captures all universal features, independently of microscopic details.

While quantum error-correcting codes can provide the necessary protection of  information against noise,  a further requirement for quantum computation is the ability to execute gates in a robust manner. Again, topological codes stand out: they usually 
provide certain intrinsic  mechanisms for executing gates in a robust way. More precisely,
there are sequences of local code deformations, under which the information stays encoded in a code with macroscopic distance, but undergoes some unitary transformation.   
In principle, this provides a robust implementation of computations by sequences of  local, and hence, potentially experimentally realizable actions. In the case of $2D$-topological codes described by topological quantum field theories, this corresponds to adiabatic movement (braiding) of quasi-particle excitations (also called anyons).

Unfortunately, as is well known, braiding (by which we mean the movement either around each other or more generally around non-trivial loops) of anyons does not always give rise to a universal gate set. Rather, the set of gates is model-dependent: braiding of $D(\mathbb{Z}_2)$-anyons generates only global phases on the sphere, and elements of the Pauli group on non-zero genus surfaces. 
Braiding of Ising anyons gives Clifford gates, whereas braiding of Fibonacci anyons generates a dense subgroup of the set of unitaries (and is therefore universal within suitable subspaces of the code space). In other words, braiding alone, without additional tricks such as magic state distillation~\cite{BravyiKitaev2005} (which has a large overhead~\cite{Fowler2012}), is not in general sufficient to provide universal fault-tolerant computation; unfortunately, the known systems with universal braiding behavior are of a rather complex nature, requiring e.g., 12-body interactions among spins~\cite{LevinWen}. Even ignoring the question of universality, the use of braiding has some potentially significant drawbacks: in general (for non-abelian anyons), it requires an amount of time which scales with the system size (or code distance) to execute a single logical gate. (Mathematically, this is reflected by the fact that string-operators cannot be implemented in constant depth for general non-abelian anyon models -- in contrast to e.g., the toric code\footnote{In the language of this paper, braiding/mapping class group elements belong to locality-preserving unitaries if the model is abelian. However, for a general non-abelian model, braiding is not locality-preserving according to our definition.}.) This implies that error-correction steps will be necessary even during the execution of such a gate (see e.g.,~\cite{PedrocchiDiVincenzo15,Hutter2015,Burton2015,Brell2014} for a recent discussion of the robustness of braiding). This may pose an additional technological challenge, for example, if the intermediate topologies are different.

Given the limitations of braiding, it is natural to look for other mechanisms for implementing robust gates in topological codes. For  stabilizer quantum codes,
the notion of transversal gates has traditionally been  used almost synonymously
with  fault-tolerant gates: their key feature is the fact that they do not propagate physical errors. More generally, for topological stabilizer codes, we can consider logical gates implementable by constant-depth quantum circuits as a proxy for robust gates: they can increase the weight of a physical error only by a constant, and are thus sufficiently robust when combined with suitable error-correction gadgets. Note that finite-depth local circuits represent a much broader class than transversal gates.

Gate restrictions on transversal, as well as constant-depth local circuits have been obtained for stabilizer and more general codes. Eastin and Knill~\cite{EastinKnill2009} argued that for any code protected against local errors, transversal gates can only generate a finite group and therefore do not provide universality.  Bravyi and K\"onig~\cite{BravyiKoenig2013} consider the group of logical gates that may be implemented by such constant-depth local circuits on geometrically local topological stabilizer codes. 
They found that such gates are contained in $\cP_D$, the $D$-th level of the Clifford hierarchy, where $D$ is the spatial dimension in which the stabilizer code is geometrically local. 

In this work, we characterize the set of gates implementable
by a locality-preserving unitary in a system described by a 2D~TQFT. By doing so, we both specialize and generalize the results of~\cite{BravyiKoenig2013}: 
we restrict our attention to dimension~$2$, but go beyond
the set of local stabilizer codes in two significant ways.

First, we obtain statements 
which are independent of the particular realization (e.g., the toric code model) but are instead phrased in terms of the TQFT (i.e., the anyon model describing the system). In this way, we obtain a characterization which holds for a gapped phase of matter, rather than just for a particular code representing that phase. On a conceptual level, this  is  similar in spirit to the work of~\cite{elseetal12},  where statements on the computational power for measurement-based quantum computation were obtained that hold throughout a  certain phase.  Here we use the term phase loosely -- we say that two systems are in the same phase if they have the same particle content. To avoid having to make  any direct reference to an underlying lattice model, we replace the notion of a constant-depth local circuit by the more general notion of a locality-preserving unitary: this is a unitary operation which maps local operators to local ones. 

Second, our results and techniques also apply to non-abelian anyon models (whereas stabilizer codes only realize certain abelian models, unless e.g., domain walls or `twists' are added~\cite{BombinTwists} that break homogeneity). In particular, we obtain statements that can be applied, e.g., to the Levin-Wen models~\cite{LevinWen}, as well as chiral phases. For such systems, restrictions on protected gates were previously not known.  Again, knowledge of the underlying microscopic model is unnecessary to apply our results, which only depend on the type of anyons present in the system. Our approach relates locality-preserving unitaries to certain symmetries of the underlying anyon model; this imposes constraints on the allowed operations.  We consider the Fibonacci and Ising models as paradigmatic examples and find that there are no non-trivial gates in the former, and only Pauli operations in the latter case. Our focus on these anyons models is for concreteness only, but our methods and conclusions apply more generally.   Some of our more general conclusions are that 
\begin{enumerate}[(i)]
\item
protected gates generically (see Section~\ref{sec:equivalenceclasses} discussing  the necessity of certain technical assumptions) form only a finite group and 
\item
when the representation of the mapping class group is computationally universal (i.e., forms a dense subgroup), then there are no non-trivial protected gates.
\end{enumerate}
Our observations are summarized in Table~\ref{tab:Braiding_VS_locality_preserving_gates}. According to our results, the 
class of locality-preserving unitaries (which is distinguished from the point of view of error correction) is too restricted and needs to be supplemented with alternative mechanisms to achieve universality.

\begin{table}[htbp]
\begin{center}
\begin{tabular}{lll}
Model & mapping class group &  locality-preserving\\
  & contained in  & unitaries contained in\\
\hline\hline
$D(\mathbb{Z}_2)$ & Pauli group & restricted Clifford group\\
\hline
abelian anyon model & generalized Pauli group &  generalized Clifford group\\
\hline
Fibonacci model & universal & global phase (trivial)\\
\hline
general anyon model & universal  & global phase (trivial)\\
\hline
Ising model & Clifford group & Pauli group\\
\hline
generic anyon model & model-dependent & finite group
\end{tabular}
\caption{
We study different anyon models (first column).
The second column describes the properties of the
unitary group generated by the (projective) representation of the mapping class group (see Section~\ref{sec:mcg}) -- this corresponds to braiding for punctured spheres. The third  column characterizes the set of protected gates. 
Our results suggest a trade-off between the computational power of the mapping class group representation and that of gates implementable by locality-preserving unitaries. 
}\label{tab:Braiding_VS_locality_preserving_gates} 
\end{center}
\end{table}

Finally, let us comment on limitations, as well as open problems arising from our work.  
The first and most obvious one is the dimensionality of the systems under consideration: our methods apply only to $2D$~TQFTs. 
The mathematics of higher-dimensional TQFTs is less developed, and currently an active research area (see e.g.,~\cite{KongWen14}). 
While the techniques of~\cite{BravyiKoenig2013}, which have recently been significantly strengthened by Pastawski and Yoshida~\cite{PastawskiYoshida14},  also apply to higher-dimensional codes (such as Haah's), they are restricted to the stabilizer formalism (but importantly,~\cite{PastawskiYoshida14} also obtain statements for subsystem codes). Obtaining non-abelian analogues of our results in higher dimensions appears  to be a challenging research problem. A full characterization of the  case $D=3$ is particularly desirable from a technological viewpoint.

Even in $2D$, there are obvious limitations of our results: the systems we consider are essentially ``homogenous'' lattices with anyonic excitations in the bulk. We are not considering defect lines, or condensation of anyons at boundaries; 
for example, our discussion excludes  the quantum double models constructed in~\cite{BeigiShorWhalen11}, which have domain walls constructed from condensation at boundaries using the folding trick. Again, we expect that obtaining statements on protected gates for these models requires additional technology in the form of more refined categorical notions, as discussed by Kitaev and Kong~\cite{KitaevKong12}. Also, although we identify possible locality preserving logical unitaries, our arguments do not show that these can necessarily be realized, either in general TQFTs or in specific models that realize TQFTs. Lastly, our work is based on the (physically motivated) assumption that a TQFT description is possible and the underlying data is given. For a concrete lattice model of interacting spins, the problem of identifying this description (or associated invariants~\cite{KitaevPreskill06,LevWenTopOrd,Haah14}), as well as constructing the relevant string-operators (as has been done for quantum double models~\cite{Kitaev03,BombinDelgado08} as well as the Levin-Wen models~\cite{LevinWen}), is a problem in its own right.

\subsubsection*{Rough statement of problem }
Our results concern families of systems defined on any $2$-dimensional orientable manifold (surface)~$\Sigma$, which we will take to be closed unless otherwise stated. Typically, such a family is defined in terms of some local physical degrees of freedom (spins) associated with sites of  a lattice embedded in~$\Sigma$.  We refer to the joint Hilbert space~$\cH_{\textrm{phys},\Sigma}$ of these spins as the `physical' Hilbert space.  The Hamiltonian~$H_\Sigma$ on~$\cH_{\textrm{phys},\Sigma}$ is local, i.e., it consists only of interactions
between ``neighbors'' within constant-diameter regions on the lattice. More generally, assuming a suitable metric on~$\Sigma$ is chosen, we may define locality in terms of the distance measure on~$\Sigma$.

We are interested in the  ground space~$\cH_\Sigma$  of $H_\Sigma$. For a topologically ordered system, this ground space is degenerate with dimension growing exponentially with the genus of~$\Sigma$, and is therefore suitable for storing and manipulating quantum information.  We will give a detailed description of
this space below (see Section~\ref{sec:TQFTbackground}); it has a preferred basis
consisting of labelings associated with some set~$\labels$. This is a finite set
characterizing all distinct types of anyonic quasiparticle excitations of~$H_\Sigma$ in the relevant low energy sector of $\cH_{\textrm{phys},\Sigma}$.

Importantly, the form of~$\cH_\Sigma$ is independent of the microscopic details (in the definition of $H_\Sigma$): it is fully determined by the associated TQFT. In mathematical terms, it can be described in terms of  the data of a modular tensor category, which also describes fusion, braiding and twists of the anyons. We will refer to~$\cH_\Sigma$ as the TQFT Hilbert space.

The significance of~$\cH_\Sigma$ is that it is protected:  local observables can not distinguish between states belonging to~$\cH_\Sigma$. This implies that $\cH_{\Sigma}$ is an error-correcting code with the property that local regions are correctable: any operator supported in a small region which preserves the code space must act trivially on it (otherwise it could be used to distinguish between ground states). 

To compute fault-tolerantly, one would like to operate on information encoded in the code space~$\cH_\Sigma$ by acting  with a unitary
$U:\cH_{\textrm{phys},\Sigma}\rightarrow\cH_{\textrm{phys},\Sigma}$ on the physical degrees of freedom\footnote{ In principle, we could consider
unitaries/isometries (or sequences thereof) of the form 
$U:\cH_{\textrm{phys},\Sigma}\rightarrow\cH'_{\textrm{phys},\Sigma'}$ which map between {\em different }systems~$\cH_{\textrm{phys},\Sigma}$ and $\cH'_{\textrm{phys},\Sigma'}$.
By a slight modification of the arguments here, we could then 
 obtain restrictions on locality-preserving isomorphisms (instead of automorphisms, cf.~Section~\ref{sec:nonabelian}). Such a scenario was discussed in~\cite{BravyiKoenig2013} in the context of stabilizer codes. Here we restrict to the case where
the systems (and associated ground spaces) are identical for simplicity,
since the main conclusions are identical.}. There are a number of features that are desirable for such a unitary to be useful  -- physical realizability being an obvious one. For fault-tolerance, two conditions are particularly natural: 
\begin{enumerate}[(i)]
\item the unitary $U$ should preserve the code space, $U\cH_\Sigma= \cH_\Sigma$ so that the information stays encoded. We call a unitary $U$ with this property an automorphism of the code and denote its restriction to $\cH_\Sigma$ by $\logical{U}:\cH_\Sigma\rightarrow\cH_\Sigma$. The action~$\logical{U}$  defines the logical operation or gate that $U$ realizes.
\item 
 typical errors should remain correctable  under the application of the unitary~$U$. 
 In the context of topological codes, which correct sufficiently local errors, and where a  local error model is usually assumed,
this condition is satisfied if $U$ does not significantly change the locality properties of an operator: if an operator $X$ has support on a region~$\cR\subset\Sigma$, then the support of $UXU^\dagger $ is contained within a constant-size neighborhood of~$\cR$.  We call such a unitary a locality-preserving unitary. 
\end{enumerate}
We call a unitary $U$ satisfying~(i) and (ii) a locality-preserving unitary automorphism of the code (or simply a topologically protected gate).  
Our goal is to characterize the set of logical operations that have the form~$\logical{U}$ for some locality-preserving\footnote{
As a side remark, we mention that our terminology is chosen with spin lattices in mind. 
However, the notion of locality-preservation can be relaxed. As will become obvious below, our results apply more generally to the set of {\em homology-preserving} automorphisms~$U$. The latter can be defined as follows: if the support of an operator $X$ is contained in a region~$\cR\subset\Sigma$ which deformation retracts to a closed curve~$C$, then the support of~$UXU^\dagger$ must be contained in a region~$\cR'\subset\Sigma$ which deformation retracts to a curve~$C'$ in the same homology class as~$C$. For example, for a translation-invariant system, translating by a possibly extensive amount  realizes such a homology-preserving (but not locality-preserving) automorphism.
}
unitary automorphism~$U$.
 For example, if $\cH_{\Sigma}$ is a topologically ordered subspace of $\cH_{\textrm{phys},\Sigma}$, the Hilbert space of a spin lattice, then (ii) is satisfied if $U$ is a constant-depth local circuit.
 Another important example is the constant-time evolution $U=\cT\exp[ -i\int dt H(t) ]$ of a system through a bounded-strength geometrically-local  Hamiltonian $H(t)$.
 Here, Lieb-Robinson bounds \cite{Lieb1972, Bravyi2006} provide quantitative statements on how the resulting unitary may be exponentially well approximated by a locality-preserving unitary.
 This is relevant since it describes the time evolution of a physical system and can also be used to model adiabatic transformations of the Hamiltonian \cite{Chen2010}.

From a computational point of view, the group
\begin{align*}
\langle \left\{\logical{U}\ |\ U\textrm{ locality-preserving unitary automorphism} \right\}\rangle
\end{align*}
generated by such gates is of particular interest: it determines the computational power of gates that are implementable fault-tolerantly with locality preserving automorphisms.

\subsubsection*{Outline}
In Section~\ref{sec:TQFTbackground}, we provide a brief introduction to the relevant concepts of TQFTs. We then derive our main results on the characterization of protected gates in Section~\ref{sec:nonabelian}. Further restrictions on the allowed protected gates are provided in Sections~\ref{sec:glblmcg} and \ref{sec:globalconstraintsfmove}. In Section~\ref{sec:examplesnonabelian}, we apply our results to particular models, deriving in particular our characterizations for  Ising and Fibonacci anyons. Finally, in Section~\ref{sec:abelianmodels} we use additional properties of abelian models to show that their protected gates must be contained within a proper subgroup of the generalized Clifford group, which is similar to the result of~\cite{BravyiKoenig2013}, but goes further. 

\section{TQFTs: background\label{sec:TQFTbackground}}
In this section, we provide the necessary  background on topological quantum field theories (TQFTs). Our discussion will be rather brief; for a more detailed discussion of topological quantum computation and anyons, we refer to~\cite{preskillnotes}. Following Witten's  work~\cite{Witten89}, TQFTs have been axiomatized by  Atiyah~\cite{Atiyah89} based on Segal's work~\cite{Segal} on conformal field theories. Moore and Seiberg~\cite{MooreSeiberg98} derived the relations satisfied by the basic algebraic data of such theories (or more precisely, a modular functor). Here we borrow some of the terminology developed in full generality by Walker~\cite{Walker91}  (see  also~\cite{FreedmanKitaevWang02}). For a thorough treatment of the category-theoretic concepts, we recommend the appendix of~\cite{Kitaev05}.

Our focus is on the Hilbert space~$\cH_{\Sigma}$ spanned by the vacuum states of a TQFT defined on the orientable surface~$\Sigma$. 
Recall that this is generally a subspace~$\cH_{\Sigma}\subset\cH_{\textrm{phys},\Sigma}$ of a Hilbert space of physical degrees of freedom. The TQFT is specified by a finite set of anyon labels $\labels = \{1,a,b,c \dots \}$, their {\em fusion rules} (described using a non-negative integer $N^c_{ab}$ for each triple of anyons $a,b,c$, called fusion multiplicities), along with $S$, $F$, $R$ and $T$ matrices (complex valued matrices with columns and rows indexed by anyon labels). If the TQFT arises from taking continuous limits of a local Hamiltonian model such as the toric code, the anyons are simply the elementary excitations of the model, and the fusion rules and matrices can be understood in terms of creating, combining, moving and annihilating anyons in the surface. The anyon set must contain a trivial particle $1 \in \labels$ such that when combined with any particle, the latter remains unchanged $N^c_{a1}=N^c_{1a} = \delta^c_a$, and each particle $a \in \labels$ must have an antiparticle $\dual{a} \in \labels$ such that $N^1_{a\dual{a}} \neq0$.  We will restrict our attention to models where $N^c_{ab}\in\{0,1\}$ for all $a,b,c\in \labels$ for simplicity (our results generalize with only minor modifications).

\subsection{String-like operators and relations\label{sec:stringlikeopsandrelations}}

We are interested in the algebra $\cA_\Sigma$ of operators $X:\cH_{\textrm{phys},\Sigma}\rightarrow\cH_{\textrm{phys},\Sigma}$ which preserve the subspace $\cH_{\Sigma}$. 
We call such an element $X\in\cA_\Sigma$ an automorphism and denote by $\logical{X}:\cH_{\Sigma}\rightarrow\cH_{\Sigma}$ the restriction to $\cH_{\Sigma}$. We call $X$ a representative (or realization)  of~$\logical{X}$. 
Operators of the form $\logical{X}$,  where $X\in\cA_\Sigma$, 
define an associative $*$-algebra~$\logical{\cA_\Sigma}$ with unit and multiplication~$\logical{X}\logical{Y}=\logical{XY}$. 
The unit element in~$\logical{\cA_\Sigma}$ is represented by the identity operator~$\id$ on the whole space~$\cH_{\textrm{phys},\Sigma}$.

Our constraints on protected gates are derived by studying how they transform certain operators acting on~$\cH_{phys,\Sigma}$ (see Fig.~\ref{fig:TopologicalLoops}). To define the latter, fix a simple closed curve~$C:[0,1]\rightarrow\Sigma$ on the surface and an ``anyon label'' $a\in\labels$. (The set of labels~$\labels$ is determined by the underlying model.)
Then there is a ``string-operator'' $F_a(C)$ acting on $\cH_{\textrm{phys},\Sigma}$, supported in a constant-diameter neighborhood of~$C$.  It corresponds to the process of creating a particle-anti\-particle-pair $(a,\dual{a})$, moving $a$ along~$C$, and subsequently fusing to the vacuum. The last step in this process involves projection onto the ground space, which is not trivial in general: the operator~$F_a(C)$ can involve post-selection, in which case it is a non-unitary element of~$\cA_\Sigma$.

The operators $\{F_a(C)\}_{a\in\labels}$ form a closed subalgebra~$\cA(C)\subset\cA_\Sigma$: they preserve the ground space and
satisfy
\begin{align}
F_a(C)F_b(C)&=\sum_n N_{ab}^n F_n(C)\ ,\qquad\qquad F_a(C)^\dagger=F_{\dual{a}}(C)\qquad\textrm{ and }\qquad F_1(C)=\id_{\cH_{\textrm{phys}}}\  \label{eq:fproductrule}
\end{align}
for the fusion multiplicities~$N_{ab}^n$ (see Section~\ref{sec:verlinde}). In addition, reversing the direction of $C$, i.e., considering $C^{-1}(t)\equiv C(1-t)$, is equivalent to exchanging the particle with its antiparticle, i.e.,
\begin{align}
F_a(C^{-1})=F_{\dual{a}}(C)\ .\label{eq:fcinversedef}
\end{align} 
Here $a\mapsto\dual{a}$ is an involution on the set of particle labels~$\labels$, again defined by the underlying model. Properties~\eqref{eq:fproductrule} and~\eqref{eq:fcinversedef} of the string-operators can be shown in the diagrammatic formalism mentioned below (but this is not needed here; we will use them as axioms).

 We denote the restriction of $F_a(C)$ to the code space~$\cH_\Sigma$ by $\logical{F_a(C)}$.
Note that, while $\logical{F_a(C)}$ is unitary in abelian anyon models, this is not the case in general.

\begin{figure}
\begin{center}
\includegraphics[width=0.6\textwidth]{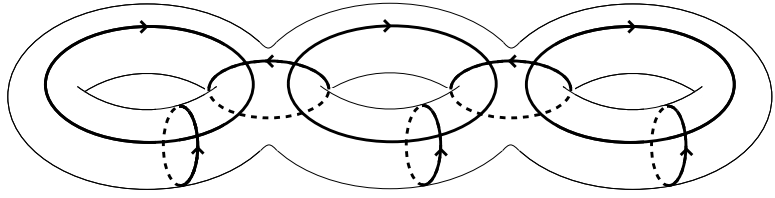}
\end{center}
\caption{Closed $2$-manifolds are characterized by their genus~$g$. The figure illustrates the 3-handled torus~$\Sigma_g$ corresponding to~$g=3$. 
A canonical set of $3g-1$ generators of the mapping class group of the surface~$\Sigma_g$ can be specified in terms of a set~$\cG=\{C_j\}_{j=1}^{3g-1}$ of loops (each associated with a Dehn twist). 
Dragging an anyon~$a$  around such  loop~$C:[0,1]\rightarrow\Sigma_g$ and fusing to the vacuum implements an undetectable operator~$F_a(C)$;   homologically non-trivial loops realize logical operations. The full algebra of logical operators is generated by the set of operators~$\{F_{a}(C)\}_{a\in\labels,C\in\cG}$. However, these operators are generally not independent.}
\label{fig:TopologicalLoops} 
\end{figure}

\begin{example}[$D(G)$ and Kitaev's toric code]\label{ex:kitaevtoric}
As an example, consider a model described by the quantum double $D(G)$ of a finite group $G$, for which Kitaev has constructed a lattice model~\cite{Kitaev03}. In the case where $G$ is abelian, we have $D(G)\cong G\times G$, i.e., the particles and fusion rules are simply given by the product group~$\labels=G\times G$. 

Specializing to $G=\mathbb{Z}_2$ gives the particles commonly denoted by $1=(0,0)$ (vacuum), $m=(1,0)$, $e=(0,1)$ and $\epsilon=m\times e=(1,1)$.  For the toric code model, the associated ribbon operators are 
\begin{align*}
F_1(C)=\id\quad
F_e(C)=\bar{X}(C)\quad
F_m(C)=\bar{Z}(C)\quad
F_\epsilon(C)=\bar{X}(C)\bar{Z}(C)\ ,
\end{align*}
where $\bar{X}(C)=\otimes_{j\in \partial_+ C} X_j$ and $\bar{Z}(C)=\otimes_{j\in \partial_- C}Z_j$ are appropriate tensor products of Pauli-$X$ and Pauli-$Z$-operators along~$C$ (as specified in~\cite{Kitaev03}). 

Specializing to $G=\mathbb{Z}_N$, with $\omega_N=\exp(2\pi i/N)$ and generalized $N$-dit Pauli operators $X$ and $Z$ (and their inverses), 
defined by their action
\begin{align*}
X\ket{j}=\ket{j+1\mod N}\qquad Z\ket{j}=\omega_N^j \ket{j}
\end{align*}
on computational basis states~$\{\ket{j}\}_{j=0,\ldots,N-1}$, we can consider such a model (the $\mathbb{Z}_N$-toric code) with generalized ribbon operators. Here
\begin{align*}
F_{(a,a')}(C)&=\bar{X}(C)^a\bar{Z}(C){^{a'}},\
\end{align*}
where $\bar{X}(C)$ is a  tensor product of Pauli-$X$ and its inverse depending on the orientation of the underlying lattice,  and  similarly for $\bar{Z}(C)$.

It is easy to check that operators associated with the same loop commute, i.e.,
\begin{align}
[F_{(a,a')}(C), F_{(b,b')}(C)]=0\ ,\label{eq:commutatrelxv}
\end{align}
and since $Z^aX^b=\omega_N^{ab}X^bZ^a$, we get the commutation relation 
\begin{align}
F_{(a,a')}(C_1)F_{(b,b')}(C_2)&=\omega_N^{a b'-a' b} F_{(b,b')}(C_2)F_{(a,a')}(C_1)\ \label{eq:commutationrelationszn}
\end{align}
for any two strings~$C_1,C_2$ intersecting once.
\end{example} 

Returning to the general case, the 
algebra of string operators does not necessarily satisfy relations as simple as~\eqref{eq:commutatrelxv} and~\eqref{eq:commutationrelationszn}. Nevertheless, some essential features hold under very general assumptions.
We express  these as postulates; they can be seen as a subset of the  isotopy-invariant calculus of labeled ribbon graphs associated with the underlying category (see e.g.,~\cite{Freedmanetal08} for a discussion of the latter). That is, the properties expressed by our postulates are a subset of the 
axioms formalizing TQFTs, and serve to capture the essential features in an algebraic manner. For particular systems (such as the toric code or the quantum double models), these postulates can be rigorously established (see~\cite{Kitaev03,BombinDelgado08}), whereas in other cases, only partial results are known (see e.g., the discussion in~\cite[p. 107]{wangbook}) but they are conjectured to hold. We sidestep the independent important and challenging problem of rigorously establishing these postulates, and instead  derive some consequences. Throughout our work, we hence assume that the models under consideration satisfy our postulates.

\begin{postulate}[Completeness of string-operators]\label{pos:completeness_of_strings}
Consider an operator $U$ with support in some region~$\cR$ which preserves the code space $\cH_\Sigma$. 
Then its action on the code space is equivalent to that of a linear combination of products of operators of the form $F_a(C)$, for a closed loop $C:[0,1]\rightarrow\cR$ which is supported in~$\cR$. 
That is, we have
\begin{align*}
\logical{U}=\sum_j \beta_j \prod_k \logical{F_{a_{j,k}}(C_{j,k})}\ .
\end{align*}
\end{postulate}
This postulate essentially means that, as far as the logical action is concerned, we may think of $\logical{U}$ as a linear combination of products of closed-loop string operators. 
Such products~$F_{a_m}(C_m)\cdots F_{a_1}(C_1)$  can conveniently be thought of as `labeled' loop gases embedded in the three-manifold~$\Sigma\times [0,1]$, where, for some $0<t_1<\cdots <t_m<1$, the operator $F_{a_j}(C_j)$ is applied at `time' $t_j$ (and hence a labeled loop is embedded in the slice $\Sigma\times \{t_j\}$). Diagrammatically, one represents such a product by the projection onto~$\Sigma$ with crossings representing temporal order, as in
\begin{align}
F_{a_2}(C_2)F_{a_1}(C_1) \ \ &= \quad \raisebox{-6ex}{\includegraphics[scale=0.5]{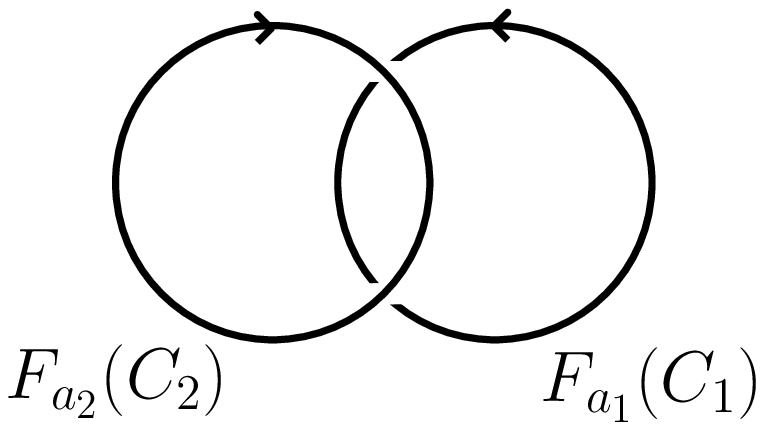}}
\label{eq:Fstringloopcrossing}
\end{align}
One may manipulate every term in a linear combination representing $U$ without changing the logical action according to certain local `moves'; in particular, the order of application of these moves is irrelevant (a fact formalized by MacLane's theorem~\cite{MacLane}).

\begin{figure}
\begin{center}
\includegraphics[width=0.4\textwidth]{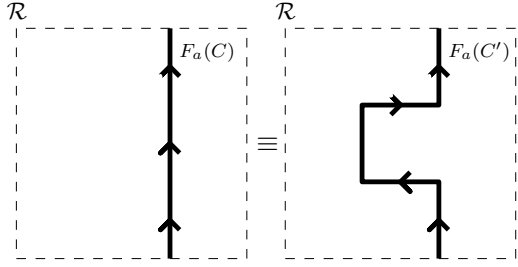}
\end{center}
\caption{The content of Postulate~\ref{pos:stringdeformationpostulate}:  We can deform a line without changing the logical action of the
string-operator\label{fig:deformationpostulate}.}

\end{figure}
For our purposes, we only require the following `local' moves,
which relate two  products $U$ and $U'$ of string-operators given 
by diagrams such as~\eqref{eq:Fstringloopcrossing}. 
More generally, they may be applied term-by-term to any linear combination if each term contains the same local sub-diagram. 
\begin{postulate}[String deformation (see Fig.~\ref{fig:deformationpostulate})]\label{pos:stringdeformationpostulate}
Suppose operators $U,U' \in \cA_\Sigma$ are identical on the complement of some region $\cR$.
Assume further that inside~$\cR$, both $U$ and $U'$ contain a single string describing the dragging of the same anyon type along a path $C$ and $C'$, respectively,  where $C'$ can be locally deformed into~$C$. Then the logical action of  $U$ and $U'$ must be equivalent: $\logical{U}=e^{i \theta}\logical{U'}$ for some unimportant phase $e^{i \theta}$.
\end{postulate}
In particular, this postulate  implies that if $C$ and $C'$ are two closed homologically equivalent loops and $a$ is an arbitrary anyon label, then the operators $F_a(C)$ and $F_a(C')$ realized by ``dragging'' the specified anyon along $C$ and $C'$ respectively have  equivalent logical action on the code space,  $\logical{F_a(C)}=e^{i\theta}\logical{F_{a}(C')}$. 

The next postulate involves local operators, and essentially states that the space~$\cH_\Sigma$ is a quantum error-correcting code protecting against local errors. While we may state it in a form only referring to local operators, we will find it more intuitive to combine it with the deformation postulate: this extends correctability from small regions to contractible loops (i.e., loops that are homotopic to a point). 
\begin{postulate}[Error correction postulate]\label{pos:errorcorrection}
If $C$ is a contractible loop, then for each $a\in\labels$, the operator~$F_a(C)$ has trivial action on the space~$\cH_\Sigma$ up to a global constant~$d_a$, that is, 
\begin{align}
\logical{F_a(C)}=d_a \id_{\cH_\Sigma}\ .
\end{align}
\end{postulate}
\noindent This postulate essentially means that we may remove certain closed loops from diagrams such as~\eqref{eq:Fstringloopcrossing}.  

An immediate consequence of these postulates is the following statement. 
\begin{proposition}[Local completeness of string operators]\label{prop:CompletenessFatStrings}
Consider an operator~$O\in\cA_\Sigma$ whose support is contained within a constant-diameter neighborhood of a simple loop $C$. Then $\logical{O}=\logical{\tilde{X}}$ for some  $\tilde{X}\in \cA(C)$.  In other words, the logical action of $O$ is identical to that of a linear combination of string-operators~$F_a(C)$.
\end{proposition}
This proposition can be seen as a consequence of the completeness condition for strings (Postulate~\ref{pos:completeness_of_strings}), the string deformation Postulate~\ref{pos:stringdeformationpostulate} and~(\ref{eq:fproductrule}).
A similar argument leads us to the following conclusion.

\begin{proposition}[Global completeness of few homology classes]\label{prop:CompletenessOfHomology}
The full logical algebra $\logical{\cA_\Sigma}$ is generated by the logical algebras ~$\logical{\cA(C)}$ associated with a finite number of inequivalent non-contractible simple loops $C$.
\end{proposition}
\begin{proof}
That the algebra $\logical{\cA_\Sigma}$  is finite-dimensional can be seen from the finite dimensionality of the code space $\cH_\Sigma$.
By Postulate \ref{pos:completeness_of_strings}, the algebra $\logical{\cA_\Sigma}$ is generated by $\{\cA(C)\}_C$. 
Let us start from a trivial algebra and build up $\logical{\cA_\Sigma}$ from a finite number of loops.
As long as the algebra is not complete, we may include additional loops $C$ such that $\logical{\cA(C)}$ is not included in the partially generated algebra.
Such a loop $C$ must be inequivalent to the previously included loops due to Postulate \ref{pos:stringdeformationpostulate}.
After a number of steps no greater than the square of the ground space dimension, we will have constructed the complete algebra.
\end{proof}

Therefore there exists a finite, \emph{minimal} set of loops which is sufficient to span $\logical{\cA_\Sigma}$.

\subsection{The Verlinde  algebra\label{sec:verlinde}}
It is convenient to formally introduce some algebraic data defined by the  underlying anyon model. We will return to the discussion of string-operators in Section~\ref{sec:baseshilberspacehsigma} and relate them to this algebraic language.

As before, let $\labels$ be the set of particle labels (generally a finite set), and let $a\mapsto\dual{a}$ be the involution giving the antiparticle associated with particle~$a$. The {\em fusion rules} of the model are encoded in integers $N^c_{ab}$, which are called fusion multiplicites. We will restrict our attention to models where $N^c_{ab}\in\{0,1\}$ for all $a,b,c\in \labels$ for simplicity (our results generalize with only minor modifications).

The {\em Verlinde algebra}~$\Ver$ is the commutative associative $*$-algebra spanned by elements~$\{\f_a\}_{a\in \labels}$ satisfying the relations
\begin{align}
\f_a \f_b=\sum_c N^c_{ab}\f_c\qquad\textrm{ and }\qquad \f_a^\dagger=\f_{\dual{a}}\ .\label{eq:verlindealgebrarelation}
\end{align}
Note that $\f_1=\id$ is the  identity element because the numbers $\{N^c_{ab}\}$ satisfy~$N^c_{a1}=N^c_{1a}=\delta_{ac}$. 

Since every anyon model is braided by definition, one indeed has $N^c_{ab}=N^c_{ba}$ and the algebra~$\Ver$ is a finite-dimensional commutative~$C^*$-algebra. Therefore $\Ver\cong\mathbb{C}^{\oplus (\dim\Ver)}$ is a direct sum of  copies of~$\mathbb{C}$. 
The fusion multiplicity $N_{ab}^c$ may also be written in terms of the modular $S$-matrix, whose matrix elements are, in the diagrammatic calculus, given by the Hopf link and the total quantum dimension $\cD$ by
\begin{align*}
S_{ab} = \frac{1}{\cD}\ \Smatrix{a}{b}\qquad\  .
\end{align*}
We consider (and restrict our attention to) the case where the $S$-matrix is unitary: here the isomorphism~$\Ver\cong\mathbb{C}^{\oplus (\dim\Ver)}$ can be made explicit thanks to the 
{\em Verlinde formula}~\cite{Verlinde88}
\begin{align}
N^c_{ab}&=\sum_x \frac{S_{ax}S_{bx}S_{\dual{c}x}}{S_{1x}}\ ,\label{eq:verlinde}
\end{align}
as the proof of the following Proposition~\ref{prop:primitiveIdempotents} shows. 
(Note that $S_{1x}=d_x/\cD$ where $\cD=\sqrt{\sum_a d_a^2}$.)
  For this purpose, we define 
the  elements
\begin{align}
\p_a &=S_{1a}\sum_b \conjugate{S_{ba}}\f_b\qquad \textrm{ for all } a\in \labels\ .\label{eq:idempotents}
\end{align}
This relation can be inverted by making use of unitarity of the $S$-matrix
\begin{align}
\f_b &=\sum_a \frac{S_{ba}}{S_{1a}}\p_a\qquad \textrm{ for all } a\in \labels\ .\label{eq:idempotentsinverted}
\end{align}
The main statement we use is the following:
\begin{proposition}[Primitive idempotents]\label{prop:primitiveIdempotents}
The elements $\{\p_a\}_{a\in \mathbb{A}}$ 
are  the unique complete set of  orthogonal minimal idempotents 
spanning the Verlinde algebra,
\begin{align}
\Ver=\bigoplus_a \mathbb{C}\p_a\ .\label{eq:verlindealgebradecomposition}
\end{align}
Furthermore,  they satisfy
\begin{align}
\sum_a \p_a=\f_1=\id\ .\label{eq:completenessrelation}
\end{align}
\end{proposition}
\begin{proof}
That $\{\p_a\}_{a\in \mathbb{A}}$ span the algebra $\Ver$ is evident from the fact that $\{\f_a\}_{a\in \mathbb{A}}$ span the algebra, and each $\f_a$ can be written in terms of $\{\p_a\}_{a\in \mathbb{A}}$ via Eq.~(\ref{eq:idempotentsinverted}). To show they are orthogonal idempotents
$\p_a\p_b=\delta_{a,b}\p_a$, first note that
\begin{align*}
\p_a\p_b&=S_{1a}S_{1b}\sum_{g,h}\conjugate{S_{ga}}\conjugate{S_{hb}}\f_g\f_h\\
&=S_{1a}S_{1b}\sum_{g,h,j}\conjugate{S_{ga}}\conjugate{S_{hb}}N^j_{gh}\f_j\\
&=S_{1a}S_{1b}\sum_{g,h,j,x}\conjugate{S_{ga}}\conjugate{S_{hb}}\frac{S_{gx}S_{hx}S_{\dual{\jmath}x}}{S_{1x}}\f_j
\end{align*}
where we used the Verlinde formula~\eqref{eq:verlinde} in the second step.
With the unitarity of the $S$-matrix, we then obtain
\begin{align*}
\p_a\p_b&=S_{1a}S_{1b}\sum_{j,x}\delta_{a,x}\delta_{b,x}\frac{S_{\dual{\jmath}x}}{S_{1x}}\f_j\\
&=\delta_{a,b}S_{1a}^2\sum_{j} \frac{S_{\dual{\jmath}a}}{S_{1a}}\f_j\\
&=\delta_{a,b}S_{1a}\sum_j S_{\dual{\jmath}a} \f_j\ .
\end{align*}
It  follows that $\p_a\p_b=\delta_{a,b}\p_a$ from the symmetry property $S_{\dual{\jmath}a}=\conjugate{S_{ja}}$, see e.g.,~\cite[Eq.~(224)]{Kitaev05}. It remains to verify that the set of projectors is unique. Consider $\q_b=\sum_a \alpha_{ba} \p_a$ for some constants $\alpha_{ba} \in  \mathbb{C}$, such that $\q_a\q_b=\delta_{a,b}\q_a$. This implies \begin{align*}
\q_a\q_b&=\sum_{dc} \alpha_{ac} \alpha_{bd} \p_c \p_d \\
&=\sum_{c} \alpha_{ac} \alpha_{bc} \p_c =  \delta_{a,b} \sum_{c} \alpha_{ac} \p_c,
\end{align*}
which implies $\alpha_{ac} \alpha_{bc} = \delta_{a,b} \alpha_{ac}$ for all $a,c \in \labels$ by linear independence of the $\p_a$'s. This implies $\alpha_{ac} =0,1$, and can only form a complete basis for the algebra $\Ver$ if $\alpha_{ac}$ is a permutation matrix, implying $\{\q_a\}_{a\in \mathbb{A}} \equiv \{\p_a\}_{a\in \mathbb{A}}$.

\end{proof}
As explained in Section~\ref{sec:baseshilberspacehsigma}, the string operators of anyons around a loop $C$ give rise to a representation of the Verlinde algebra. While the 
projections (introduced in Eq.~\eqref{eq:measurementoperators} below)  associated with the idempotents  
are not a basis for the logical algebra~$[\mathcal{A}_\Sigma]$, they are a basis of 
a subalgebra~$[\mathcal{A}_\Sigma(C)]$ isomorphic to the Verlinde algebra. This algebra must be respected by the locality-preserving unitaries, and this is best understood in terms of the idempotents.
 This is the origin of the non-trivial constraints we obtain on the realizable logical
operators.

\subsection{Bases of the Hilbert space~$\cH_\Sigma$\label{sec:baseshilberspacehsigma}}
Eq.~\eqref{eq:fproductrule} shows that the collection of operators $\{ \logical{F_a(C)}\}_{a\in \labels}$ form a representation of the Verlinde (fusion) algebra~$\Ver$. 
By linear independence of operators $ \{ \logical{F_a(C)} \}_{a\in \labels}$, we see that the representation is faithful, such that the logical loop algebra is isomorphic to the Verlinde algebra
\begin{align}
\logical{\cA(C)}\cong \Ver .\label{eq:loopisomorphictover}
\end{align}
This will be central in the following development.
Considering the primitive idempotents~\eqref{eq:idempotents}, it is natural to consider the corresponding operators in this representation, that is, we set
\begin{align}
\logical{P_a(C)}=S_{1a}\sum_b \conjugate{S_{ba}} \logical{F_b(C)}\ .\label{eq:measurementoperators}
\end{align}
Since the set $\{ \logical{F_a(C)} \}_{a\in \labels}$ forms a representation of the Verlinde algebra, the $\{ \logical{P_a(C)}\}_{a\in \labels}$ are orthogonal projectors as a consequence of Proposition \ref{prop:primitiveIdempotents}.
The inverse relation to~\eqref{eq:measurementoperators} is given by
\begin{align}
\logical{F_b(C)}&=\sum_a\frac{S_{ba}}{S_{1a}} \logical{P_a(C)}\ .\label{eq:inverserelationm}
\end{align}
While the projectors $\logical{P_a(C)}$ associated with a loop do not span the full logical algebra, they do span the local logical algebra of operators supported along $C$ which must be respected by locality preserving unitaries.
Intuitively, $\{P_a(C)\}_{a\in\labels}$ are projectors onto the smallest possible sectors of the Hilbert space which can be distinguished by a measurement supported on $C$.

A state in the image of $P_a(C)$ has the interpretation of carrying flux~$a$ through the loop~$C$. 
In particular, since the code space $\cH_{\Sigma}$ corresponds to the vacua of a TQFT, there are no anyons present on $\Sigma$, however, there can be flux associated to non-contractible loops.
We can use the operators $\{P_a(C)\}_{a,C}$ to define bases of the Hilbert space~$\cH_\Sigma$.

Let us first define the Hilbert space~$\cH_\Sigma$ in more detail.

\begin{definition}[DAP-decomposition]
 Consider a minimal collection $\cC=\{C_j\ |\ C_j:[0,1]\rightarrow\Sigma\}_j$ of pairwise non-intersecting non-contractible loops, which cut the surface~$\Sigma$ into a collection of surfaces homeomorphic to discs, annuli and pants. We call $\cC$ a DAP-decomposition (see Fig.~\ref{fig:dapdecompositiontor}). 
\end{definition}

  \begin{figure}
\begin{center}
\includegraphics[width=0.3\textwidth]{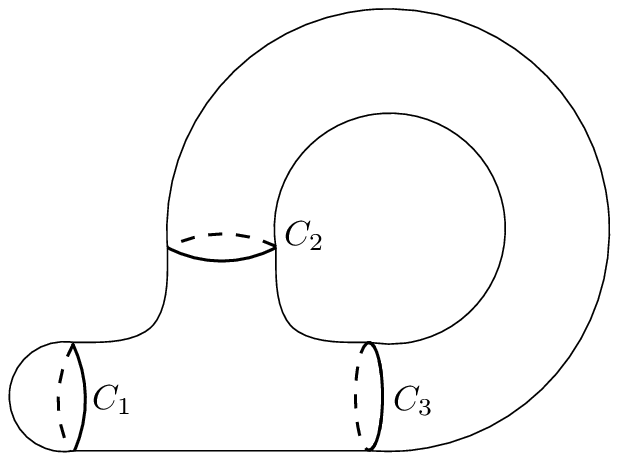}
\end{center}
\caption{A simple DAP decomposition of a torus utilizing a disc enclosed by $C_1$, an annulus enclosed by $\{C_2, C_3\}$ and a pair of pants enclosed by $\{C_1, C_2, C_3\}$.
This decomposition is not minimal in that the same manifold could have been decomposed using a single loop.\label{fig:dapdecompositiontor}}
\end{figure} 

A {\em labeling}~$\labeling{\ell}\colon\cC\mapsto \labels$  is an assignment of an anyon label $\labeling{\ell}(C)$ to every  loop $C\in\cC$ of a DAP decomposition. 
We call $\labeling{\ell}$ fusion-consistent if it satisfies the following conditions:
\begin{itemize}
\item[(i)]  for every loop $C \in \cC$ enclosing a disc on $\Sigma$, $\labeling{\ell}(C)=1$, the vacuum label of the anyon model.
\item[(ii)] for every pair of loops $\{ C_2, C_3\}\subset \cC$ defining an annulus in $\Sigma$, $\labeling{\ell}(C_2)=\dual{\labeling{\ell}(C_3)}$ assuming the loops are oriented such that the annulus is found to the left.
\item[(iii)] for every triple $\{C_1,C_2,C_3\}\subset\cC$ defining a pair of pants in $\Sigma$, the labeling~$\labeling{\ell}$ satisfies the fusion rule
\begin{equation*}
    N_{\labeling{\ell}(C_1),
       \labeling{\ell}(C_2)}^{\dual{\labeling{\ell}(C_3)}}
    \neq
     0,
\end{equation*}
where the loops are oriented such that the pair of pants is found to the left.
\end{itemize}
Here we may assume $\labeling{\ell}(C^{-1})=\dual{\labeling{\ell}(C)}$, where $C^{-1}$ denotes the loop coinciding with $C$ but with opposite orientation.  

Now fix any DAP-decomposition $\cC$ of $\Sigma$
and let $\mathsf{L}(\cC)\subset \labels^{|\cC|}$ be the set of fusion-consistent labelings. The Hilbert space~$\cH_{\Sigma}$ is the formal span of elements of $\mathsf{L}(\cC)$
\begin{align*}
\cH_\Sigma\coloneqq\sum_{\labeling{\ell}\in\mathsf{L}(\cC)} \mathbb{C}\labeling{\ell}=\sum_{\labeling{\ell}\in\mathsf{L}(\cC)} \mathbb{C}\mkern2mu\ket{\labeling{\ell}} .
\end{align*}
Any fusion-consistent labeling $\labeling{\ell}\in\mathsf{L}(\cC)$ defines an element~$\ket{\labeling{\ell}} \in \cH_\Sigma$ such that the vectors $\{\ket{\ell} \}_{\ell\in\mathsf{L}(\cC)}$ are an orthonormal basis (which we call $\cB_{\cC}$) of $\cH_\Sigma$, and  this defines the inner product.

It is important to remark that this construction of~$\cH_\Sigma$ is independent of the DAP-decom\-position $\cC$ of $\Sigma$ in the following sense: if $\cC$ and $\cC'$ are two distinct DAP-decompositions, then there is a unitary basis change between the bases $\cB_{\cC}$ and $\cB_{\cC'}$. In most cases under consideration, this basis change can be obtained as
a product of unitaries  associated with local ``moves'' connecting two DAP decompositions~$\cC$ and~$\cC'$. One such basis change is associated with a four-punctured sphere (the $F$-move), and specified by  the unitary $F$-matrix in Fig.~\ref{fig:Fmove}. Another matrix of this kind, the $S$-matrix (which also arose in our discussion of the Verlinde algebra), connects the two bases $\cB_\cC$ and $\cB_{\cC'}$ of $\cH_{\textrm{torus}}$ associated with the first and second non-trivial cycles on the torus (Fig.~\ref{fig:Fmove}).
  In this case, writing $\cB_\cC=\set{\ket{a}_\cC}_a$ and $\cB_{\cC'}=\set{\ket{a}_{\cC'}}_a$ since each basis element~$\ket{\labeling{\ell}}$ is specified by a single label $\labeling{\ell}(C),\labeling{\ell}(C')\in\labels$, we have the relation
  \begin{align}
  \ket{a}_{\cC'} &= \sum_b S_{ba}\ket{b}_\cC\ .\label{eq:smatrixbasischangetorus}
  \end{align}
Other unitary basis changes arise from the representation of the mapping class group, as discussed in Section~\ref{sec:mcg}.
All these basis changes constitute the second ingredient for the non-trivial constraints we obtain on the realizable logical operators.

A basis element $\ket{\labeling{\ell}}\in\cB_{\cC}$ associates the anyon label~$\labeling{\ell}(C)$ with each curve $C\in\cC$. The vector~$\ket{\ell}$ is the (up to a phase) unique simultaneous $+1$-eigenvector of all the projections $\{P_{\ell} (C)\}_{C\in\cC}$. It is also a simultaneous eigenvector with respect to Dehn-twists along each curve~$C\in\cC$ with eigenvalue $e^{i\theta_{\labeling{\ell}(C)}}$.   The action of Dehn-twists  along curves~$C'$ not belonging to~$\cC$ can be obtained by applying the local moves to change into a basis~$\cB_{\cC'}$ associated with a DAP-decomposition $\cC'$ containing~$C'$.
   
  \begin{figure}
\begin{center}
\includegraphics[width=0.4\textwidth]{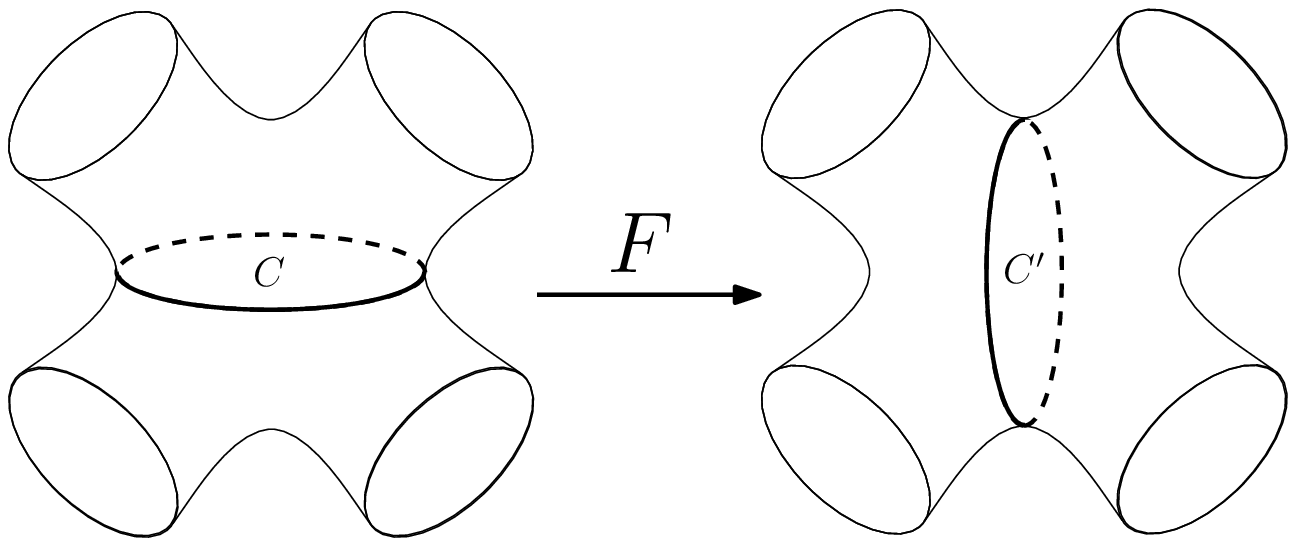}\qquad\raisebox{2ex}{\includegraphics[width=0.45\textwidth]{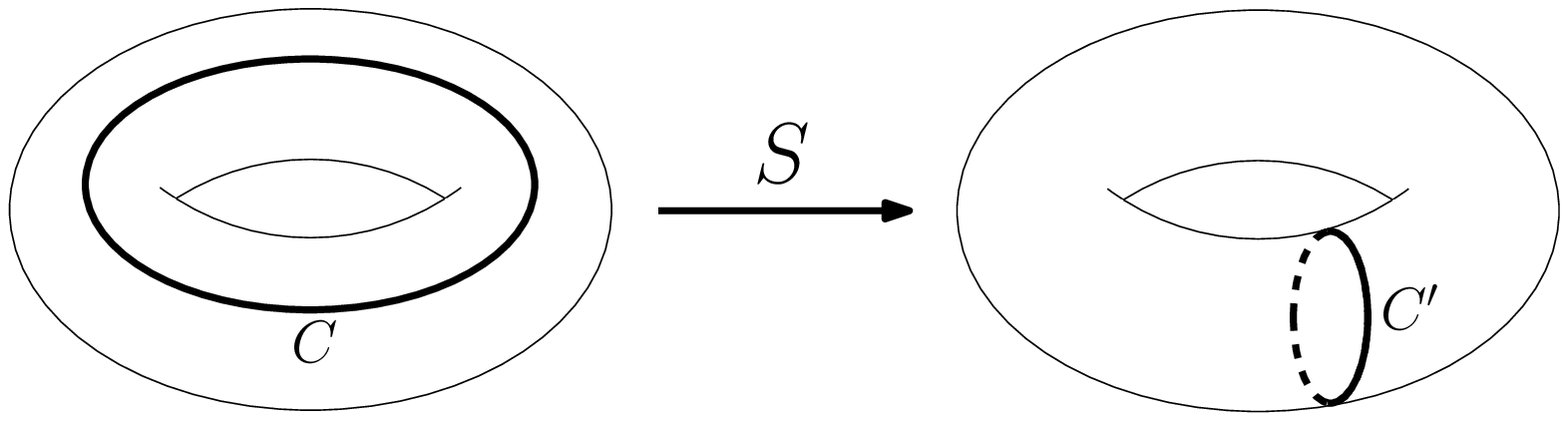}}
\end{center}
\caption{Two DAP-decompositions $\cC=\{C\}$ and $\cC'=\{C'\}$ of either the $4$-punctured sphere (left), or the torus (right), are related by an $F$-move or an $S$-move, respectively.}
\label{fig:Fmove} 
\end{figure}

  \subsection{Open surfaces: labeled boundaries\label{sec:opensurfaces}}
\newcommand*{\hC}{\hat{C}}
So far, we have been discussing the Hilbert space~$\cH_{\Sigma}$ associated with closed surfaces; this does not cover the physically important case of pinned localized excitations (which correspond to punctures/holes in the surface). Here we describe the modifications necessary to deal with surfaces with boundaries. 
We assume that the boundary~$\partial \Sigma=\bigcup_{\alpha=1}^{M} \hC_\alpha$ is the disjoint union of $M$ simple closed curves, and assume that an orientation~$\hC_\alpha:[0,1]\rightarrow\partial \Sigma$ has been chosen for each boundary component~$\hC_\alpha$ such that $\Sigma$ is found to the left. In addition, we fix a label $a_\alpha\in\labels$ for every boundary component~$\hC_\alpha$. We call this a labeling of the boundary. Let us write $\Sigma(a_1,\ldots,a_M)$ for the resulting object (i.e., the surfaces, its oriented boundary components, and the associated labels). We call $\Sigma(a_1,\ldots,a_M)$ a surface with labeled boundary components; slightly abusing notation,  we sometimes write $\Sigma=\Sigma(a_1,\ldots,a_M)$ when the presence of boundaries is understood/immaterial.

A TQFT associates to every surface~$\Sigma(a_1,\ldots,a_M)$ with labeled boundary components a Hilbert space~$\cH_{\Sigma(a_1,\ldots,a_M)}$. 
The construction is analogous to the case of closed surfaces and based on DAP-decompositions. 
The only modification compared to the case of closed surfaces is that only DAP-decompositions
including the curves $\{\hC_\alpha\}^{M}_{\alpha=1}$ are allowed; furthermore,
the labeling on these boundary components is fixed by $\{a_\alpha\}_{\alpha=1}^M$. That is,
``valid'' DAP-decompositions are of the form~$\cC=\{C_1,\ldots,C_N,\hC_1,\ldots,\hC_M\}$ with
curves $\{C_j\}_{j=1}^N$ ``complementing'' the boundary components, and valid labelings are fusion-consistent, i.e.,~$\labeling{\ell}\in\mathsf{L}(C)$ with the additional condition that they agree with the boundary labels,~$\labeling{\ell}(\hC_\alpha)=a_\alpha$ for $\alpha=1,\ldots,M$.   
To simplify the discussion, we will often omit the boundary components~$\{\hC_\alpha\}_{\alpha}$ and focus on the remaining degrees of freedom associated with the curves~$\{C_j\}_{j}$. 
It is understood that boundary labelings have to be  fusion-consistent with the labeling~$\{a_\alpha\}_{\alpha}$ of the boundary under consideration.

As a final remark, note that boundary components labeled with the trivial particle~$1\in\labels$ correspond to contractible loops in a surface without this boundary (i.e., obtained by ``gluing in a disc''). This means that they can be omitted: we have  the isomorphism
\begin{align*}
\cH_{\Sigma(1)}\cong \cH_{\Sigma'}\ ,
\end{align*}
where $\Sigma'$ is the surface with one boundary component less that of $\Sigma$.

\subsubsection*{Example: the $M$-anyon Hilbert space}
A typical example we are interested in
is the labeled surface
\begin{align*}
S^2(z^M)=S^2(\underbrace{z,\ldots,z}_{M\textrm{ times}})\ ,
\end{align*}
where $S^2( ~, ~, ~...~,~,~)$ is the punctured sphere, and $z\in\labels$ is some fixed anyon type (we assume that
each boundary component has the same orientation). 
The Hilbert space $\cH_{S^2(z^M)}$ is the space of $M$ anyons of type~$z$. 
When $M=N+3$ for some $N\in\mathbb{N}$, we can choose a `standard' DAP-decomposition~$\cC=\{C_j\}_{j=1}^N$ as shown in Fig.~\ref{fig:standardpants}. 
A fusion-consistent labeling~$\labeling{\ell}$ of the standard DAP-decomposition~$\cC$
corresponds to a sequence
$(x_1,\ldots,x_N)=(\labeling{\ell}(C_1),\dots,\labeling{\ell}(C_N))$ 
such that
\begin{align}
N_{zz}^{x_1}=N_{x_Nz}^{\dual{z}}=1\qquad\textrm{ and }\qquad
N_{x_jz}^{x_{j+1}}=1\qquad\textrm{ for all }j=1,\ldots,N-1,\label{eq:fusiontreecond}
\end{align}
as illustrated by Fig.~\ref{fig:standardpants}.

  \begin{figure}
\begin{center}
\includegraphics[width=0.5\textwidth]{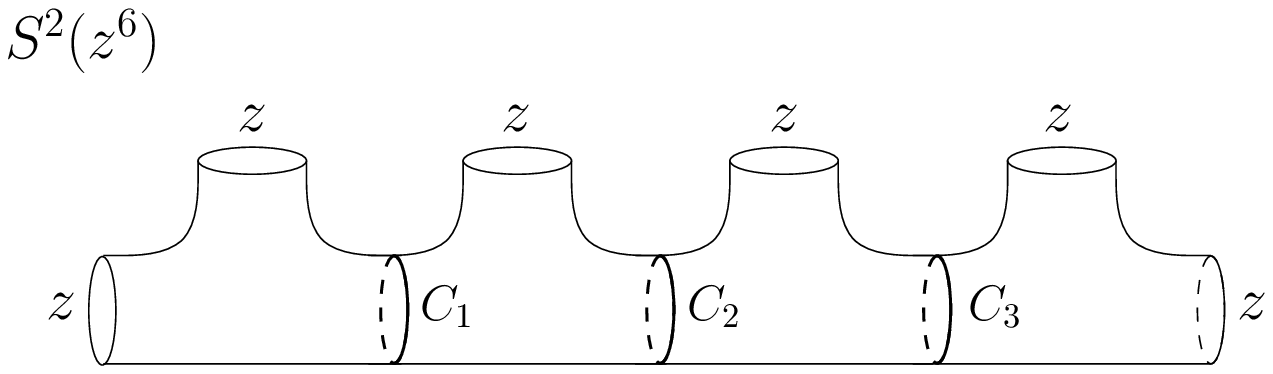}
\vspace{1.5cm}
\hspace{0.2cm}\includegraphics[width=0.5\textwidth]{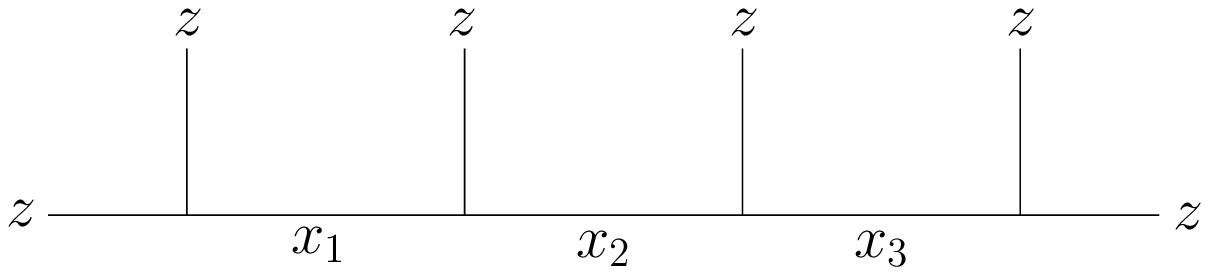}
\end{center}
\caption{\label{fig:standardpants} The `standard' DAP-decomposition of the $6$-punctured sphere, and the corresponding fusion-tree notation representing the labeling which assigns $\labeling{\ell}(C_i)=x_i$.}
\end{figure} 

 \subsection{The gluing postulate\label{sec:gluing}}
Throughout our work, we restrict our attention to models satisfying an additional property we refer to as the gluing postulate (which is often called the ``gluing axiom'' in the literature). Consider a closed curve~$C$ embedded in~$\Sigma$. We will assume that $C$ is an element of a DAP-decomposition~$\cC$; although this is not strictly necessary, it will simplify our discussion. Now consider the surface~$\Sigma'$ obtained by cutting~$\Sigma$ along~$C$. Compared to $\Sigma$, this is a surface with two boundary components~$C'_1,C'_2$ (both isotopic to~$C$) added. We will assume that these have opposite orientation. A familiar example is the case where cutting~$\Sigma$ along~$C$ results
in two disconnected surfaces $\Sigma'=\Sigma_1\cup\Sigma_2$, as depicted in Fig.~\ref{fig:4puncturecut} in the case where $\Sigma$ is the $4$-punctured sphere.

  \begin{figure}
\begin{center}
\includegraphics[width=0.2\textwidth]{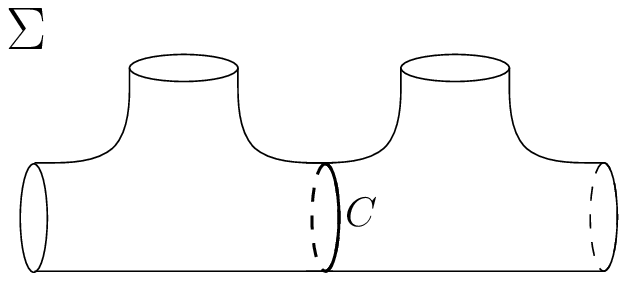}
\vspace{1cm}

\includegraphics[width=0.3\textwidth]{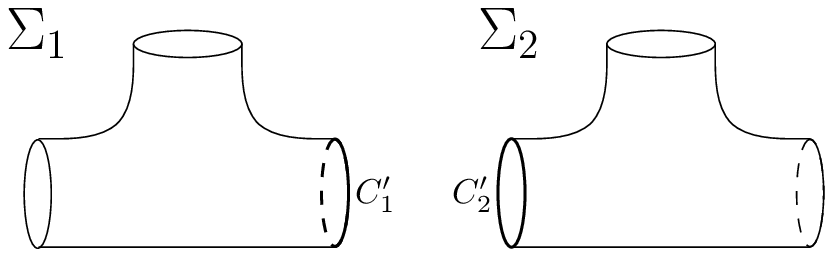}
\end{center}
\caption{Cutting a surface $\Sigma$ along some closed curve $C$ of a DAP-decomposition yields a disconnected surface $\Sigma'=\Sigma_1\cup\Sigma_2$ having additional boundary components $C'_1$ and $C'_2$.}
\label{fig:4puncturecut} 
\end{figure}

Let $a$ be a particle label. We will denote by $\cH_{\Sigma'(a,\dual{a})}$ the Hilbert space associated with the open surface~$\Sigma'$, where boundary~$C_1'$ is labeled by $a$ and boundary $C'_2$ by $\dual{a}$. The gluing postulate states that
the Hilbert space of the surface~$\Sigma$ has the form
\begin{align}
\cH_\Sigma \cong \bigoplus_a \cH_{\Sigma'(a,\dual{a})}\ \label{eq:gluingaxiom}
\end{align}
where the direct sum is over all particle labels~$a$ that occur in different fusion-consistent labelings of~$\cC$. In the  special case where cutting along~$C$ gives two components~$\Sigma_1,\Sigma_2$, we have~$\cH_\Sigma\cong\bigoplus_a \cH_{\Sigma_1(a)}\otimes \cH_{\Sigma_2(\dual{a})}$.

The isomorphism~\eqref{eq:gluingaxiom} can easily be made explicit. A first observation is that~$\cH_\Sigma$ decomposes as~$\cH_\Sigma=\bigoplus_a \cH_{a,\Sigma}(C)$,
where
\begin{align}
\cH_{a,\Sigma}(C)\coloneqq\mathsf{span}\{\ket{\labeling{\ell}} \ |\ \labeling{\ell}\in\mathsf{L}(\cC), \labeling{\ell}(C)=a\}\label{eq:fusealabelH}
\end{align}
is the space spanned by all labelings which assign the  label~$a$ to~$C$. 
It therefore suffices to argue that
\begin{align}
\cH_{a,\Sigma}(C) \cong\cH_{\Sigma'(a,\dual{a})}\ .\label{eq:restrictedisomorphismalabels}
\end{align}
To do so, observe that the DAP-decomposition~$\cC$ of~$\Sigma$ gives rise to a DAP-decomposition~$\cC'=\cC\backslash\{C\}$ of $\Sigma'$. Any labeling~$\labeling{\ell}\in\mathsf{L}(\cC)$ 
with $\labeling{\ell}(C)=a$
restricts to a labeling~$\labeling{\ell}'\in\mathsf{L}(\cC')$ of the 
labeled surface~$\Sigma'(a,\dual{a})$.  Conversely, any 
labeling~$\labeling{\ell}'\in\mathsf{L}(\cC')$ of the
surface~$\Sigma'(a,\dual{a})$ provides a labeling~$\labeling{\ell}\in\mathsf{L}(\cC)$
(by setting~$\labeling{\ell}(C)=a$). This 
defines the isomorphism~\eqref{eq:restrictedisomorphismalabels} in terms of  basis states~$\{\ket{\labeling{\ell}}\}_{\labeling{\ell}\in\mathsf{L}(\cC)}$ and~$\{\ket{\labeling{\ell}'}\}_{\labeling{\ell}'\in\mathsf{L}(\cC')}$.

\subsubsection*{Example: decomposing the $M$-anyon Hilbert space}
Consider the $M$-punctured sphere~$\Sigma=S^2(z^{M})$ with the standard 
DAP decomposition of Fig.~\ref{fig:standardpants} 
and boundary labels~$z$ (corresponding to $M$~anyons of type~$z$).  Cutting~$S^2(z^{M})$ along $C_j$
gives a surface~$\Sigma'_j$ which is the disjoint union of
two punctured spheres, with $j+2$ and $M-j$ punctures, respectively.  The resulting surface labelings
are $S^2(z^{j+1},a)$ and $S^2(\dual{a},z^{M-1-j})$. 
That is, if $\Sigma=S^2(z^{M})$ is the original surface and $\Sigma'_j(a,\dual{a})$ is the resulting one, then
\begin{align}
\cH_{\Sigma'_j(a,\dual{a})}&=\cH_{S^2(z^{j+1},a)}\otimes\cH_{S^2(\dual{a},z^{M-1-j})}\ .\label{eq:cutspheres}
\end{align}
This is illustrated in Fig.~\ref{fig:6puncturelabeled} for the case $M=6$ and $j=2$. 

  \begin{figure}
  \begin{center}
\includegraphics[width=0.5\textwidth]{6puncturedsphereDAP.eps}
\vspace{1.5cm}  
\includegraphics[width=0.6\textwidth]
{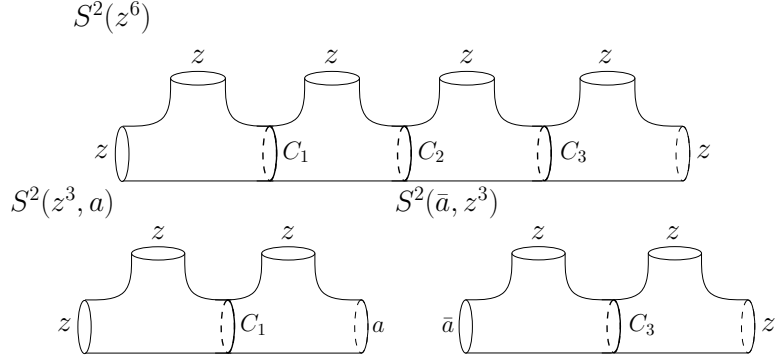}
\end{center}
\caption{The $6$-punctured sphere $S^2(z^6)$ shown with three curves $C_1,C_2,C_3\in\cC$ of a DAP-decomposition. Cutting along $C_2$ with labeling $\labeling{\ell}(C_2)=a$ results in the two surfaces $S^2(z^3,a)$ and $S^2(\dual{a},z^3)$.}
\label{fig:6puncturelabeled} 
\end{figure}

\subsection{The mapping class group\label{sec:mcg}}
In the following, we denote by~$\mcg_{\Sigma}$ the {\em mapping class group} of the surface~$\Sigma$. Physically, a mapping class group element for a surface $\Sigma$ gradually deforms the surface, but returns to the original configuration. For the $n$-punctured sphere, the mapping class group includes braiding of the punctures. For the torus, a Dehn twist is an element of the mapping class group.
More formally, elements of $\mcg_{\Sigma}$ are isotopy classes of orientation-preserving diffeomorphisms of~$\Sigma$ preserving labels and commuting with boundary parametrization (see e.g.,~\cite{FreedmanKitaevWang02}). Slightly abusing notation, we will often simply write~$\vartheta\in\mcg_{\Sigma}$ for an equivalence class represented by a map    $\vartheta:\Sigma\rightarrow\Sigma$. If $\Sigma$ is the torus, then the mapping class group is generated by two elements, $\mcg_\Sigma=\langle s,t\rangle$ where $s$ and $t$ are the standard generators of the modular group. For the $M$-punctured sphere $S^2(z^M)$, we will also need the $M-1$  elements $\{\sigma_j\}_{j=1}^{M-1}$, where $\sigma_j$ braids holes~$j$ and~$j+1$.


The Hilbert space~$\cH_{\Sigma}$ is equipped with a projective unitary representation 
\begin{align}
\begin{matrix}
\mcg_\Sigma & \rightarrow &\mathsf{U}(\cH_{\Sigma})\\
\vartheta & \mapsto &\bV(\vartheta)
\end{matrix} \label{eq:mcgrep}
\end{align}
of the mapping class group~$\mcg_\Sigma$. For example, for the torus, $\bV(s)=S$ and $\bV(t)=T$ are the usual $S$- and $T$-matrices defined by the modular tensor category. For the $M$-punctured sphere~$S^2(z^M)$ with $M=N+3$, we again use the standard DAP-decomposition with associated basis~$\{\ket{\labelSequenceSphere{x}}\}_{\labelSequenceSphere{x}}$. Here the sequences $\labelSequenceSphere{x}=(x_1,\dots,x_N)$ are subject to the fusion rules (see~\eqref{eq:fusiontreecond}) and the action on such vectors is
\begin{align*}
    \bV(\sigma_1)
    \ket{\labelSequenceSphere{x}}
    & =R_{x_1}^{z
                z}
       \ket{\labelSequenceSphere{x}}, \\
    \bV(\sigma_k)
    \ket{\labelSequenceSphere{x}}
    & =\sum_{x'}
       B(x_{k-
            1},
         x_{k+
            1})_{x'
                 x_k}
       \ket{x_1,
            \dots,
            x_{k-
               1},
            x',
            x_{k+
               1},
            \dots,
            x_N}
            \qquad
      \text{for $k=2,\dots,N+1$}, \\
    \bV(\sigma_{N+
                2})
    \ket{\labelSequenceSphere{x}}
    & =\conjugate{R_{x_1}^{z
                           z}}
       \ket{\labelSequenceSphere{x}},
\end{align*}
where $B(a,b)=\tilde{F}^{-1}\tilde{R}\tilde{F}$ is the braid matrix. Here the matrices~$\tilde{F}$ and~$\tilde{R}$ 
are given in terms of the tensors~$F$ and~$R$ associated with the TQFT\footnote{More precisely, for $B(x_{k-1},x_{k+1})$, the relevant matrices are~$\tilde{F}_{x',x}=F^{zx_{k-1}x_k}_{x_{k+1}zx'}$ and
$\tilde{R}$ is diagonal with entries $\tilde{R}_{x,x}=R^{zz}_{x}$. Here $F$ is the $F$-matrix associated with basis changes on the four-punctured sphere (see Section~\ref{sec:fourpuncturedgeneral}), whereas~$R^{xy}_z$ determines 
an isomorphisms between certain Hilbert spaces associated with the three-punctured sphere. We refer to e.g.,~\cite[p.~48]{preskillnotes} for a derivation of these expressions.}.

\section{Constraints on locality-preserving automorphisms\label{sec:nonabelian}}

In this section, we derive restrictions on topologically protected gates for general non-abelian models. Our strategy will be to consider what happens to string-operators.  We will first consider operators associated with a single loop~$C$,  and derive restrictions on the map $F_a(C)\mapsto U F_a(C)U^\dagger$, or, more precisely, its effect on logical operators, $\logical{F_a(C)}\mapsto \logical{U F_a(C)U^\dagger}$. We will argue that this map implements an isomorphism of the Verlinde algebra and exploit this fact to derive a  constraint which is `local' to a specific loop. 
We will subsequently consider more `global' constraints arising from fusion rules, as well as basis changes.

We would like to characterize locality-preserving unitary automorphisms~$U\in\cA_\Sigma$ in terms of their logical action~$\logical{U}$. For example in the toric code, where the physical qubits are imbedded in the edges of the square lattice, the locality preserving unitaries include the well-known transversal gates of single-qubit unitaries applied to each qubit. More general examples of locality preserving unitaries in the toric code are finite depth circuits composed of gates of arbitrary unitaries applied to physical qubits in geometrically-local patches of fixed diameter.

A first goal is to characterize the map
\begin{align}
\begin{matrix}
\rho_U: & \logical{\cA_\Sigma} &\rightarrow&\logical{\cA_\Sigma}\\
&\logical{X}&\mapsto  &\logical{UXU^{-1}}\ ,
\end{matrix}\label{eq:rhoudef}
\end{align}
which determines the evolution of logical observables in the Heisenberg picture. (Clearly, this does not depend on the representative, i.e., if $[X]=[X']$, then $\rho_U([X])=\rho_U([X'])$.) In fact, the map~\eqref{eq:rhoudef} fully determines $U$ up to a global phase since~$\logical{\cA_\Sigma}$ contains an operator basis for linear maps on~$\cH_\Sigma$. However, it will often be more informative to characterize the action of~$\logical{U}$ on basis elements of~$\cH_\Sigma$. This will require additional effort.

The main observation is that the map~\eqref{eq:rhoudef} defines
an automorphism of~$\logical{\cA_\Sigma}$, since
\begin{align}
\rho_U(\logical{X})\rho_U(\logical{X'})=\rho_U(\logical{X}\logical{X'})\qquad\textrm{ for all }X,X'\in\cA_\Sigma\qquad\textrm{ and }\qquad \rho_U^{-1}=\rho_{U^{-1}}\ .\label{eq:automorphism}
\end{align}
Combined with the locality of $U$,~\eqref{eq:automorphism} severly constrains~$\rho_U$. Using this fact, we obtain a number of very general constraints, which will be worked out in more detail in the following. 

\subsection{A local constraint from a simple closed loop}
Specifying the action of~$\rho_U$ on all of $\logical{\cA_\Sigma}$ completely determines~$[U]$ up to a global phase. 
However, this is not entirely straightforward; instead, we
fix some simple closed curve~$C$ and characterize the restriction to the subalgebra~$\cA(C)\subset\cA_\Sigma$, i.e., the map
\begin{align}
\begin{matrix}
\rho_U(C): & \logical{\cA(C)} &\rightarrow&\logical{\cA(C)}\\
&\logical{X}&\mapsto  &\logical{UXU^{-1}}\ ,
\end{matrix}\label{eq:rhourestricteddef}
\end{align}
Observe that this map is well-defined since $UXU^{-1}$ is supported in a neighborhood of $C$ (by the locality-preservation of $U$), and hence
 $\logical{UXU^{-1}}=\logical{X'}$ for some operator $X'\in\cA(C)$ (here we have used Proposition \ref{prop:CompletenessFatStrings}). 
 It is also easy to see that it defines an automorphism of the subalgebra~$\logical{\cA(C)}$.
 
 As we argued above, the algebra~$\cA(C)$ is isomorphic to $\Ver$.
 This carries over to $\logical{\cA(C)}\cong \Ver \cong \mathbb{C}^{\oplus |\labels|}$. As $\Ver$ has idempotents $\p_{a\in \labels}$, the logical algebra for loop $C$ has idempotents $\{\logical{P_a(C)}\}_{a\in \labels}$. Note that the idempotents $\{\logical{P_a(C)}\}_{a\in \labels}$ in the logical algebra are unique, in that there is no linear combination of these idempotents which yields a distinct, complete set of idempotents. At the physical level however, there can be huge redundency, with many different physical operators corresponding to the same logical operator, i.e. $\logical{P_a(C)} =\logical{P'_a(C)}$, for $P_a(C) \neq P'_a(C)$.
 We use the following fact:
 \begin{lemma} \label{thm:permutingprojectors}
 The set of automorphisms of the algebra $\Ver$
 is in one-to-one correspondence  with the permutations~$S_{|\labels|}$. For $\pi\in S_{|\labels|}$, the associated automorphism   $\rho_\pi:\Ver \rightarrow \Ver$ is defined by its action on the  central idempotents $\p_a$
 \begin{align}
 \rho_\pi(\p_a)=\p_{\pi(a)}\qquad\textrm{ for }a\in \labels\ \label{eq:rhopiargument}
 \end{align}
  \end{lemma}
  \begin{proof}
  It is clear that~\eqref{eq:rhopiargument} defines an automorphism for every $\pi\in S_{|\labels|}$. Also, from Eq.~(\ref{eq:automorphism}) we see that $\p_a \p_b = \delta_{ab}\p_b$ implies  $\rho(\p_a) \rho(\p_b) = \delta_{ab} \rho(\p_b)$, such that $\rho(\p_a) \in \Ver$ are a complete set of projectors (Proposition~\ref{prop:primitiveIdempotents}). As there is a unique set of complete projectors for $\Ver$, we conclude that $\rho(\p_a) = \p_{\pi(a)}$ for some permutation $\pi \in  S_{|\labels|}$.
  \end{proof}
 Applying this to $\logical{\cA(C)}$ shows that a locality-preserving unitary automorphism realizes, up to\emph{important} phases, a a permutation of labelings.
 Let us emphasize that it is the projectors (idempotents) $[P_a(C)]$ which are being permuted, and not the string operators $[F_a(C)]$.
 
 \begin{proposition}[Local constraint]\label{prop:mainsingleloop}
 Let $U$ be a locality-preserving automorphism of the code, and let $\rho_U(\logical{X})=\logical{UXU^{-1}}$. 
 \begin{enumerate}[(i)]
 \item\label{it:automorphismpermute}
 For each simple closed loop $C$ on $\Sigma$, there is a permutation $\permLoop{C}\colon\labels\rightarrow \labels$ of the particle labels such that
 \begin{align}
 \begin{matrix}
\rho_U: & \logical{\cA(C)} &\rightarrow&\logical{\cA(C)}\\
&\logical{P_a(C)}&\mapsto  &\logical{P_{\permLoop{C}(a)}(C)}&\qquad\textrm{ for all }a\in \labels\  ,
\end{matrix}\label{eq:automorphismpermutation}
\end{align}
(and linearly extended to all of $\logical{\cA(C)}$).
 \item\label{it:trivialstringop}
For some anyon model $\labels$ with an associated $S$ matrix, let $D_{a,b}=\delta_{a,b}\cdot d_a$ be the diagonal matrix with the quantum dimensions on the diagonal. Let $\permLoop{C}\colon\labels\rightarrow \labels$ be a permutation associated with a loop~$C$ as in~\eqref{it:automorphismpermute}, and let
$\Pi$ be the matrix defined by~$\Pi_{x,y}\coloneqq\delta_{x,\permLoop{C}(y)}$.
Define the matrix
\begin{align}
\Lambda\coloneqq S\Pi^{-1}D\Pi D^{-1}\Pi^{-1}S^{-1}\ .\label{eq:lambdamatrixexpr}
\end{align}
Then
\begin{align}
\rho_U(\logical{F_b(C)})&=\sum_{b'}\Lambda_{b,b'}\logical{F_{b'}(C)}\ .\label{eq:rhoufbmappingprop}
\end{align}
  \end{enumerate}
 \end{proposition}

 \begin{proof}
 We have already argued that~\eqref{it:automorphismpermute} holds.
 For the proof of~\eqref{it:trivialstringop}, we use the relationship between $\{P_a(C)\}_a$ and $\{F_a(C)\}_a$ (cf.~\eqref{eq:measurementoperators}
and~\eqref{eq:inverserelationm}) to get
   (suppressing the dependence on the loop~$C$)
  \begin{align*}
  \rho_U(\logical{F_b})=\sum_a \frac{S_{b,a}}{S_{1,a}}\logical{P_{\permLoop{C}(a)}}=\sum_{b'}\left(\sum_a\frac{S_{b,a}}{S_{1,a}} S_{1,\permLoop{C}(a)}\conjugate{S_{b',\permLoop{C}(a)}}\right) \logical{F_{b'}}\ .
  \end{align*}
  The claim~\eqref{eq:rhoufbmappingprop} follows from this using $(\Pi^{-1}S^{-1})_{a,b'}=(S^{-1})_{\permLoop{C}(a),b'}=\conjugate{S_{b',\permLoop{C}(a)}}$ by the unitarity of $S$, as well as  the fact that $S_{1,a}=d_a/\mathcal{D}$ and hence 
    $\frac{S_{b,a}}{S_{1,a}}S_{1,\permLoop{C}(a)}=(S\Pi^{-1}D\Pi D^{-1})_{b,a}$.
 \end{proof}
 
\subsection{Global constraints from DAP-decompositions, fusion rules and  the gluing postulate\label{sec:globalconstraintsgluingaxiom}}
 For higher-genus surfaces, we can obtain information by applying Proposition~\ref{prop:mainsingleloop} to all loops of a DAP-decomposition; these  must then satisfy the following consistency condition. 
 \begin{proposition}[Global constraint from fusion rules]\label{prop:globalfusionconstraint}
 Let $U$ be a locality-preserving automorphism of the code. 
 Let $\cC$ be a DAP-decomposition of $\Sigma$, and consider the family
 of permutations~$\perms=\{\permLoop{C}\}_{C\in\cC}$ defined by Proposition~\ref{prop:mainsingleloop}. Then this defines a permutation $\perms\colon\mathsf{L}(\cC)\to\mathsf{L}(\cC)$ of the set of fusion-consistent labelings via
 \begin{equation}
     \perms(\labeling{\ell})(C)
     \coloneqq
      \permLoop{C}[\labeling{\ell}(C)]
 \end{equation}
 for all $C\in\cC$.
 We have
 \begin{align}
 U\ket{\labeling{\ell}}&=\phaseExp{}{\labeling{\ell}}\ket{\perms(\labeling{\ell})}\qquad\text{for all $\labeling{\ell}\in L(\cC)$}\label{eq:upsielltransformation}
 \end{align}
 with some phase $\phaseExp{}{\labeling{\ell}}$ depending on~$\labeling{\ell}$.
 \end{proposition}
 \begin{proof}
 Let us  fix some basis element $\ket{\labeling{\ell}}\in\cB_\cC$.
The vector~$\ket{\labeling{\ell}}$ is a $+1$-eigenvector of $P_{\labeling{\ell}(C)}(C)$ for each $C\in \cC$; hence according to~\eqref{eq:automorphismpermutation}, the vector $U\ket{\labeling{\ell}}$ is a $+1$-eigenvector of $P_{\permLoop{C}[\labeling{\ell}(C)]}(C)=P_{\perms(\labeling{\ell})(C)}(C)$ for every $C\in\cC$. This implies that it is proportional to $\ket{\perms(\labeling{\ell})}$, hence we obtain~\eqref{eq:upsielltransformation}. Fusion-consistency of $\perms(\labeling{\ell})$ follows because $U\ket{\labeling{\ell}}$ must be an element of $\cH_\Sigma$.
 \end{proof}

Proposition~\eqref{prop:globalfusionconstraint}
expresses the requirement that  a locality-preserving automorphism~$U$
 maps the set of fusion-consistent labelings into itself. 

In fact, we can say more: it must be an isomorphism between the subspaces of~$\cH_\Sigma$
arising from the gluing postulate (i.e., Eq.~\eqref{eq:gluingaxiom}). This allows us to 
constrain the set of allowed permutations~$\perms=\{\permLoop{C}\}_{C\in\cC}$  arising from locality-preserving automorphisms even further:

 \begin{proposition}[Global constraint from gluing]\label{prop:gluingconstraint}
Let $C$ be an element of a DAP-decomposition of $\Sigma$. 
Recall that \begin{align}
\cH_\Sigma=\bigoplus_a \cH_{a,\Sigma}(C)\ ,\label{eq:directsumainprop}
\end{align} where the  subspaces in the direct sum are defined by labelings associating~$a$ to~$C$.  Let $U$ be a locality-preserving automorphism of the code  and let $\permLoop{C}\colon\labels\rightarrow\labels$ be the 
permutation associated with~$C$  by Proposition~\ref{prop:mainsingleloop}. 
Then for every~$a\in\labels$ occuring in~Eq.~\eqref{eq:directsumainprop},
the restriction of $U$ to $\cH_{a,\Sigma}(C)$ defines an isomorphism
\begin{align}
\cH_{a,\Sigma}(C)\cong \cH_{\permLoop{C}(a),\Sigma}(C)\ .\label{eq:isomorphismUrestriction}
\end{align}
In particular, if $\Sigma'$ is the surface obtained by cutting~$\Sigma$ along~$C$, then
\begin{align}
\cH_{\Sigma'(a,\dual{a})}\cong\cH_{\Sigma'(\permLoop{C}(a), \dual{\permLoop{C}(a)})}\label{eq:isomorphismcutspaces} 
\end{align}
for every $a\in\labels$ occuring in the sum~\eqref{eq:directsumainprop}.
 \end{proposition}
The reason we refer to
 Proposition~\eqref{prop:gluingconstraint} as a global constraint (even though it 
superficially only concerns a single curve~$C$) is that the surface~$\Sigma'$ and hence the spaces~\eqref{eq:isomorphismcutspaces}
depend on the global form of the surface~$\Sigma$ outside the support of~$C$.
 \begin{proof}
Proposition~\eqref{prop:globalfusionconstraint} implies that $U\cH_{a,\Sigma}(C)\subset \cH_{\permLoop{C}(a),\Sigma}(C)$
for any $a$ in expression~\eqref{eq:directsumainprop}. 
Since~$U$  acts unitarily on the whole space~$\cH_\Sigma$,
this is compatible with~\eqref{eq:directsumainprop}
only if  $U\cH_{a,\Sigma}(C)=\cH_{\permLoop{C}(a),\Sigma}(C)$ for any such $a$. This proves~\eqref{eq:isomorphismUrestriction}. Statement~\eqref{eq:isomorphismcutspaces} then immediately follows from~\eqref{eq:restrictedisomorphismalabels}. 
 \end{proof}
A simple but useful implication of Proposition~\ref{prop:gluingconstraint}
is that
\begin{align}
\dim\left(\cH_{\Sigma'(a, \dual{a})}\right)=\dim\left(\cH_{\Sigma'(\permLoop{C}(a), \dual{\permLoop{C}(a)})}\right)\label{eq:dimisomorphismcutspaces} 
\end{align}
is a necessary condition that~$\permLoop{C}$ has to satisfy.

\subsection{Global constraints from basis changes}
 
Eq.~\eqref{eq:automorphismpermutation} essentially tells us that a locality-preserving protected gate~$U$ can only permute particle labels; it indicates that such a gate~$U$ is related to certain symmetries of the anyon model. But~\eqref{eq:automorphismpermutation} does not tell us what phases basis states may acquire. We show how to obtain constraints on these phases by considering basis changes.
 This  also further constrains the allowed permutations on the labels of the idempotents.

Consider two DAP-decompositions $\cC$ and $\cC'$.
Expressed in the first basis~$\cB_{\cC}$, we have
\begin{align}
U\ket{\labeling{\ell}}&=\phaseExp{}{\labeling{\ell}}\ket{\perms(\labeling{\ell})}\ \label{eq:uappliedtopsi}
\end{align}
for some unknown phase~$\phaseArg{}{\labeling{\ell}}$ depending only on the labeling $\labeling{\ell}\in\mathsf{L}(\cC)$. 
This means that with respect to the basis elements of $\cB_{\cC}$,
the operator~$U$ is described by a matrix~${\bf U}={\bf \Pi} {\bf D}(\{\phaseArg{}{\labeling{\ell}}\}_{\labeling{\ell}})$, where 
${\bf \Pi}$ is a permutation matrix (acting on the fusion-consistent labelings~$\mathsf{L}(\cC)$), and ${\bf D}$ is a diagonal matrix with 
entries $\{\phaseExp{}{\labeling{\ell}}\}_{\labeling{\ell}}$ on the diagonal.

Analogously, we can consider the operator $U$ expressed as a matrix~${\bf U'}$ in terms of the basis elements of $\cB_{\cC'}$. We conclude that  ${\bf U'}={\bf \Pi'} {\bf D}(\{\phaseArg{'}{\labeling{\ell}}\}_{\labeling{\ell}})$, for $\labeling{\ell}\in\mathsf{L}(\cC')$, with  a (potentially different) permutation matrix~${\bf \Pi'}$, and (potentially different) phases~$\{\phaseArg{'}{\labeling{\ell}}\}_{\labeling{\ell}}$.

Let ${\bf V}$ be the unitary change-of-basis matrix for going from $\cB_{\cC}$ to $\cB_{\cC'}$. Then we must have
\begin{equation}
    \label{eq:consistencycondition}
    {\bf V}
    {\bf U}
    ={\bf U'}
     {\bf V}.
\end{equation}
We show below that this equation strongly constrains the phases  as well as the permutations in~\eqref{eq:upsielltransformation}. More specifically, we will examine constraints arising when using basis changes~$\bV$ defined by $F$-moves in Section~\ref{sec:globalconstraintsfmove}. In Section~\ref{sec:glblmcg}, we consider  basis changes~$\bV$ defined by elements of the mapping class group.

\section{Global constraints from the mapping class group\label{sec:glblmcg}}
The following is based on the simple observation that we must have consistency conditions of the form~\eqref{eq:consistencycondition} for more general basis changes (in particular, basis changes not made up of $F$-moves only). We are particularly interested in the case where the basis change is the result of applying a mapping class group element.

\subsection{Basis changes defined by the mapping class group\label{sec:basischangesmcg}}
A key property of the representation~\eqref{eq:mcgrep} of the mapping class group~$\mcg_\Sigma$ is that it maps idempotents according to
\begin{align}
V(\vartheta) P_a(C)V(\vartheta)^\dagger &= P_a\bigl(\vartheta(C)\bigr)\ .\label{eq:mcgidempotents}
\end{align}
Let us fix a `standard' DAP-decomposition~$\cC$, and let $\cB_{\cC}=\{\ket{\labeling{\ell}}_{\cC}\}_{\labeling{\ell}}$ be the corresponding standard basis.

Let $\vartheta$ be an arbitrary element of $\mcg_\Sigma$. Consider the basis
\begin{equation*}
    \cB_{\vartheta(\cC)}\coloneqq\{V(\vartheta)\ket{\labeling{\ell}}\}_{\labeling{\ell}}.
\end{equation*}
Because
of~\eqref{eq:mcgidempotents}, this basis is a simultaneous eigenbasis
of the complete set of commuting observables associated with the DAP decomposition $\vartheta(\cC)\coloneqq\{\vartheta(C_j)\}_{j=1}^M$. The change of basis from $\cB_\cC$ to $\cB_{\vartheta(\cC)}$ is given by the image $V(\vartheta)$ of the mapping class group element~$\vartheta$.

In particular, if $\bV(\vartheta)$ is the matrix representing $V(\vartheta)$  in the standard basis, then \eqref{eq:consistencycondition} implies
\begin{align}
\bV(\vartheta)\bPi\bD&=\bPi(\vartheta)\bD(\vartheta)\bV(\vartheta)\ \label{eq:compatibilitycond}
\end{align}
for some permutation matrix $\bPi(\vartheta)$ and a diagonal matrix $\bD(\vartheta)$ consisting of phases.

Some terminology will be useful: Let
$\phiperm$ be the set of matrices
of the form $\bPi\bD$, where $\bPi$ is a permutation of fusion-consistent labelings, and $\bD$ is a diagonal matrix with phases (these are sometimes called \emph{unitary monomial matrices}).
For $\bU\in\phiperm$  and $\vartheta\in\mcg_\Sigma$, we say that $\bU$ intertwines with $\vartheta$ if
\begin{align*}
\bV(\vartheta)\bU \bV(\vartheta)^\dagger\in\phiperm\ .
\end{align*}
Let $\phiperm_\vartheta\subset \phiperm$ be the set of matrices that intertwine with~$\vartheta$, and let
\begin{align*}
\phiperm_{\mcg_\Sigma}=\bigcap_{\vartheta\in\mcg_\Sigma} \phiperm_\vartheta\ 
\end{align*}
be the matrices that are intertwiners of the whole mapping class group representation. We have shown the following:
\begin{theorem}\label{thm:mainmappinglcassgroup}
Let $\bU$ be the matrix representing a protected gate~$U$ in the standard basis. Then~$\bU\in\phiperm_{\mcg_\Sigma}$.
\end{theorem}

As an example, consider the torus: since $T=\bV(t)$ is diagonal,  it is easy to see that for any $\bPi\bD\in\Delta$, we have 
$T\bPi\bD T^{-1}=\bPi\bD'$ for some $\bD'$. This implies that $\phiperm_t=\phiperm$ is generally not interesting, i.e., $\bU\in\phiperm_t$ does not impose an additional constraint. In contrast, mapping class group elements such as~$s$ and $st$ generally give different non-trivial constraints.

\subsection{Density of the mapping class group representation and absence of protected gates}
The following statement directly links computational universality of
the mapping class group representation to the non-existence of protected gates.
\begin{corollary}\label{cor:density}
Suppose the representation of~$\mcg_\Sigma$ is dense in the projective unitary group $\mathsf{PU}(\cH_\Sigma)$. Then there is no non-trivial protected gate.
\end{corollary}
\begin{proof}
Let $U$ be an arbitrary protected gate and let $\bU\in\phiperm$ be the matrix representing it in the standard basis. Assume for the sake of contradiction that $U$ is non-trivial. Then $\bU$ is a unitary with at least two different eigenvalues~$\lambda_1$, $\lambda_2\in\unitaryGroup{1}$. In particular, there is a diagonalizing unitary ${\bV_1}$ such that
$\bV_1 \bU \bV_1^\dagger=\mathsf{diag}(\lambda_1,\lambda_2)\oplus \tilde{\bU}$ for some matrix $\tilde{\bU}$.  Setting $\bV_2=H\oplus I$, where $H$ is the Hadamard matrix
\begin{equation*}
    H
    =\frac{1}
          {\sqrt{2}}
     \left(\begin{matrix}
               1 & 1\\
               1 & -1
           \end{matrix}\right),
\end{equation*}
and $\bV=\bV_2\bV_1$, we obtain that
\begin{align}
\bV \bU\bV^\dagger &\not\in\phiperm
\label{eq:notelementphiperm}
\end{align}
because this matrix contains both diagonal and off-diagonal elements. Note that if $\lambda_2=-\lambda_1$ one may use the matrix
\begin{equation*}
    \frac{1}
         {2}
    \left(\begin{matrix}
                  1    & -\sqrt{3} \\
              \sqrt{3} & 1
          \end{matrix}\right)
\end{equation*}
instead of~$H$.

Observe also that~\eqref{eq:notelementphiperm} stays valid if we replace $\bV$ by a sufficiently close approximation (up to an irrelevant global phase) $\tilde{\bV}\approx \bV$.
In particular, by the assumed density, we may approximate $\bV$ by a product
$\tilde{\bV}=\bV(\vartheta_1)\cdots\bV(\vartheta_m)$
of images of $\vartheta_1,\ldots,\vartheta_m\in\mcg_\Sigma$.
But then  we have
\begin{align*}
\bU\not\in\Delta_{\vartheta_1\cdots \vartheta_m}\ ,
\end{align*}
which contradicts Theorem~\ref{thm:mainmappinglcassgroup}.
\end{proof}

Note that in general, the mapping class group is only dense on a subspace $\cH_0\subset \cH_\Sigma$. This is the case for example when the overall system allows for configurations where anyons can be present or absent (e.g., a boundary may or may not carry a topological charge). In such a situation,~$\cH_\Sigma$ decomposes into superselection sectors which are defined by the gluing postulate (i.e., having fixed labels associated with certain closed loops associated). Corollary~\ref{cor:density} can be adapted to this situation, e.g., as explained in Appendix~\ref{sec:appendixsuperselection} (Lemma~\ref{lem:generalizeddensity}).

\subsection{Characterizing diagonal protected gates\label{sec:diagonalgates}}
 Fix  a DAP-decomposition~$\cC$ and let $\vartheta\in\mcg_\Sigma$. Let us  call two (fusion-consistent) labelings $\labeling{\ell}_1$, $\labeling{\ell}_2$ {\em connected by~$\vartheta$} (denoted $\labeling{\ell}_1\connected{\vartheta}\labeling{\ell}_2$) if there is a labeling~$\labeling{\ell}$ such that
\begin{align*}
0 & \neq \langle \labeling{\ell}\vert V(\vartheta)\vert\labeling{\ell}_m\rangle\qquad\textrm{ for }m=1,2\ .
\end{align*}
(Here $\ket{\labeling{\ell}}$ is the associated basis element of $\cB_\cC$.)  
More generally, let us say $\labeling{\ell}_1$, $\labeling{\ell}_2$ are connected (written $\labeling{\ell}_1\Leftrightarrow\labeling{\ell}_2$) if there exists an element $\vartheta\in\mcg_\Sigma$  such that $\labeling{\ell}_1\connected{\vartheta}\labeling{\ell}_2$. Clearly, this notion is symmetric in $\labeling{\ell}_1,\labeling{\ell}_2$, and furthermore, it is reflexive, i.e., $\labeling{\ell}_1\Leftrightarrow\labeling{\ell}_1$ since $\labeling{\ell}_1\connected{\id}\labeling{\ell}_1$.  We can therefore define an equivalence relation on the set of labelings: we write $\labeling{\ell}_1\sim\labeling{\ell}_2$ if there are labelings $\labeling{k}_1,\dots,\labeling{k}_m$ such that
$\labeling{\ell}_1\Leftrightarrow\labeling{k}_1\Leftrightarrow\dots\Leftrightarrow\labeling{k}_m\Leftrightarrow\labeling{\ell}_2$.  We point out (for later use) that we can always find a finite collection~$\{\vartheta_k\}_{k=1}^M\subset\mcg_\Sigma$ that generates the relation~$\sim$ in the sense that $\labeling{\ell}_1\sim\labeling{\ell}_2$ if and only if $\labeling{\ell}_1\connected{\vartheta_k}\labeling{\ell}_2$ for some $k$ (after all, we only have a finite set of labelings~$\labeling{\ell}$).

Observe that if the representation of~$\mcg_\Sigma$ has a non-trivial invariant subspace, then  there is more than one equivalence class. We discuss an example of this below (see Section~\ref{sec:equivalenceclasses}). However, in important special cases such as the Fibonacci or Ising models, there is only one equivalence class for the relation~$\sim$, i.e., any pair of labelings are connected (see Lemma~\ref{lem:connectednessfib} and Lemma~\ref{lem:connectednessising} below). 

\begin{lemma}\label{lem:trivialphasegatecharact}
Consider a protected gate $U$ acting diagonally in the basis~$\cB_\cC$
as $U\ket{\labeling{\ell}}=\phaseExp{}{\labeling{\ell}}\ket{\labeling{\ell}}$. 
\begin{enumerate}[(i)]
\item\label{it:diagonalthetaaction}
Suppose that
$U$ also acts diagonally in the basis $\cB_{\vartheta(\cC)}$. 
Then~$\phaseArg{}{\labeling{\ell}_1}=\phaseArg{}{\labeling{\ell}_2}$ for any pair $\labeling{\ell}_1\connected{\vartheta}\labeling{\ell}_2$ connected by~$\vartheta$. 
\item\label{it:simtrivialproperty}
Suppose that 
$\{\vartheta_k\}_{k=1}^M\subset\mcg_\Sigma$ generates the relation $\sim$, and $U$ acts diagonally in each basis~$\cB_{\vartheta_k(\cC)}$. Then $\varphi$ assigns the same value to every element of
the same equivalence class under~$\sim$.
\end{enumerate}
\end{lemma}
We will refer to a protected gate~$U$ 
with property~\eqref{it:simtrivialproperty} as a $\sim$-trivial gate. One implication of Lemma~\ref{lem:trivialphasegatecharact} is that any protected gate which is close to the identity acts as a $\sim$-trivial gate (see the proof of Theorem~\ref{cor:finitegroup}). In Section~\ref{sec:finitenessprot}, we will show how to use this statement to prove that the set of protected gates is finite up to irrelevant phases.

\begin{proof}
Consider two labelings $\labeling{\ell}_1,\labeling{\ell}_2$ satisfying 
$\labeling{\ell}_1\connected{\vartheta}\labeling{\ell}_2$.
Then, writing $\bV=\bV(\vartheta)$, we know that
\begin{equation} 
    \bV_{\labeling{\ell},
         \labeling{\ell}_1}
    \neq
     0
    \qquad
    \text{and}
    \qquad
    \bV_{\labeling{\ell},
         \labeling{\ell}_2}
    \neq
     0
    \label{eq:assumptionconnected}
\end{equation}
for some labeling~$\labeling{\ell}$, where 
$\bV_{\labeling{\ell},\labeling{k}}=\langle\labeling{\ell}\vert V(\vartheta)\vert\labeling{k}\rangle$. Since~$U$ acts diagonally in both bases~$\cB_{\cC}$ and $\cB_{\vartheta(\cC)}$ by assumption, \eqref{eq:compatibilitycond} becomes simply
\begin{align}
\bV\bD\bV^\dagger &=\bD(\vartheta)\  \label{eq:commutativity}
\end{align}
when written in the standard basis.
Here the diagonal matrices are given by $\bD=\diag(\set{\phaseArg{}{\labeling{\ell}}}_{\labeling{\ell}})$ and~$\bD(\vartheta)=\diag(\set{\phaseArg{'}{\labeling{\ell}}}_{\labeling{\ell}})$.
Taking
 the diagonal entry at position $(\labeling{\ell},\labeling{\ell})$ in the matrix equation~\eqref{eq:commutativity}, we get the identity
\begin{equation}
    \sum_{\labeling{k}}
    \mathrm{e}^{\mathrm{i}
                (\phaseArg{}{\labeling{k}}-
                 \phaseArg{'}{\labeling{\ell}})}
    |\bV_{\labeling{\ell},
          \labeling{k}}|^2
    =1.
    \label{eq:angleequation}
\end{equation}
By unitarity of the mapping class group representation, we also have
\begin{equation}
    \sum_{\labeling{k}}
    |\bV_{\labeling{\ell},
          \labeling{k}}|^2
    =1.
    \label{eq:vunitarityequation}
\end{equation}
By taking the real  part of~\eqref{eq:angleequation}, it is straightforward to see that compatibility with~\eqref{eq:vunitarityequation} imposes that
$\cos\bigl(\phaseArg{}{\labeling{k}}-\phaseArg{'}{\labeling{\ell}}\bigr)=1$ whenever $|\bV_{\labeling{\ell},\labeling{k}}|\neq 0$ or 
\begin{equation*}
    \phaseArg{}{\labeling{k}}
    \equiv
     \phaseArg{'}{\labeling{\ell}}\mod
     2
     \pi
    \qquad
    \text{for all~$\labeling{k}$ with $|\bV_{\labeling{\ell},\labeling{k}}|\neq 0$}.
\end{equation*}
With~\eqref{eq:assumptionconnected}, we conclude that 
$\phaseArg{}{\labeling{\ell}_1}=\phaseArg{'}{\labeling{\ell}}=\phaseArg{}{\labeling{\ell}_2}$, which proves claim~\eqref{it:diagonalthetaaction}.

The claim~\eqref{it:simtrivialproperty} immediately follows from~\eqref{it:diagonalthetaaction}. 
 \end{proof}
 We will show how to apply this result to the Fibonacci model in Section~\ref{sec:fibnpunctured}.  Note that Lemma~\ref{lem:trivialphasegatecharact}
does not generally rule out the existence of
non-trivial diagonal protected gates in the standard basis (an example is a Pauli-$Z$  in the Ising model, see Section~\ref{sec:isingfourpuncturedcase}):  it is important that the protected gate is diagonal in 
\emph{several} different bases~$\{\cB_{\vartheta_k(\cC)}\}_k$.

A simple consequence of Lemma~\ref{lem:trivialphasegatecharact} is that any protected gate has a finite order up to certain phases:
\begin{lemma}
There is a finite $n_0$ (depending only on the dimension of $\cH_\Sigma$) such that for every protected gate $U$, there is an $n\leq n_0$ such that $U^n$ is a $\sim$-trivial phase gate.
\end{lemma}
\begin{proof}
Consider an arbitrary DAP-decomposition~$\cC$ and suppose $U$ acts as~\eqref{eq:upsielltransformation} in the basis~$\cB_{\cC}$.
Since the permutation~$\perms$ acts on the finite set~$\mathsf{L}(\cC)$ of fusion-consistent labelings, it has finite order~$n_\cC$.
This means that $U^{n_\cC}$ acts diagonally in the basis~$\cB_{\cC}$.

Assume $\{\vartheta_k\}_{k=1}^M\subset\mcg_\Sigma$ generate the relation~$\sim$. Setting $n=\mathsf{lcm}(n_{\vartheta_1(\cC)},\ldots,n_{\vartheta_M(\cC)})$, we can apply
Lemma~\ref{lem:trivialphasegatecharact} to $U^n$ to reach the conclusion that 
$U^n$ is $\sim$-trivial. Furthermore, since the number $n$ depends only on
the permutation~$\perms$, and there are only finitely many such permutations, there is 
a finite $n_0$ with the claimed property. 
\end{proof}

\subsection{Finiteness of the set of protected gates\label{sec:finitenessprot}}
In the following, we will ignore phase differences that are ``global'' to subspaces of vectors defined by the equivalence classes of~$\sim$. That is, we will call two protected gates $U_1$ and $U_2$ equivalent (written $U_1\sim U_2$) if
\begin{align*}
\begin{matrix}
\bU_1&=&\bPi \bD_1\\
\bU_2&=&\bPi \bD_2
\end{matrix}\qquad\textrm{ and }\qquad (\bD_2)_{\labeling{\ell},\labeling{\ell}}=\phaseExp{}{[\labeling{\ell}]}(\bD_1)_{\labeling{\ell},\labeling{\ell}}\ ,
\end{align*}
i.e., they encode the same permutation of fusion-consistent labels, and their phases only differ by a phase~$\phaseArg{}{[\labeling{\ell}]}$  depending on the equivalence class~$[\labeling{\ell}]$ that $\labeling{\ell}$ belongs to. This is equivalent to the statement that~$\bU_1^{-1}\bU_2=\bD_1^{-1}\bD_2$ acts as a phase  dependent only on the equivalence class, i.e., $U_1^{-1}U_2$ is a $\sim$-trivial phase gate. 

We obtain an Eastin and Knill~\cite{EastinKnill2009} type statement, which is one of our main conclusions.
 
\begin{theorem}[Finite group of protected gates]\label{cor:finitegroup}
The number of equivalence classes of protected gates is finite. 
\end{theorem}
In particular, this means that locality-preserving automorphisms on their own do not provide quantum computational universality. 

\begin{proof}
Assume that there are infinitely many equivalence classes of protected gates.
Then we can choose a sequence~$\{U_n\}_{n\in\mathbb{N}}$ of
protected gates indexed by integers and belonging to different equivalence classes each.
Since the number of permutations of fusion-consistent labels is finite,
there exists at least one permutation matrix~$\bPi$ such that
there is an infinite subsequence of protected gates $U_n$ with
$\bU_n=\bPi \bD_n$, i.e., they act with the same permutation.
Applying the Bolzano-Weierstrass theorem to this subsequence, 
we conclude that there is a convergent subsequence of protected gates $\{U_{n_j}\}_{j\in\mathbb{N}}$ such that~$\bU_{n_j}=\bPi\bD_{n_j}$ for all $j$.
 Let
$U=\lim_{j\rightarrow\infty} U_{n_j}$ be the corresponding limit, and let us define
$\Ut_j\coloneqq U^{-1}U_{n_j}$. Clearly, each $\Ut_j$ is a protected gate and
\begin{align}
\bUt_j &= \bD^{-1}\bD_{n_j}\ \label{eq:elementsdef}
\end{align}
acts non-trivially on subspaces defined by equivalence classes,
i.e., $\Ut_j$ is a $\sim$-non-trivial phase gate. This is because of the assumption that
the original sequence~$\{U_n\}_{n\in\mathbb{N}}$ has elements belonging to different equivalence classes. Furthermore, we have that
\begin{align}
\lim_{j\rightarrow\infty} \bUt_j&=\bI\ ,\label{eq:limitidentity}
\end{align}
where $\bI$ is the identity matrix.

For a mapping class group element~$\vartheta\in\mcg_\Sigma$, the matrix
expressing the action of $\Ut_j$ in the basis~$\cB_{\vartheta(\cC)}$ is given by~$\bV(\vartheta)\bUt_j \bV(\vartheta)^\dagger$. Because $\Ut_j$ is a protected gate, we get
\begin{align}
\bV(\vartheta)\bUt_j \bV(\vartheta)^\dagger&=\tilde{\bPi}_{j}\tilde{\bD}_j\label{eq:standardformdiag}
\end{align}
for some permutation matrix $\tilde{\bPi}_j$ and a diagonal matrix~$\tilde{\bD}_j$ of phases. Combining~\eqref{eq:limitidentity},~\eqref{eq:standardformdiag}, using the unitarity of $\bV(\vartheta)$ and continuity, we conclude that
there exists some $N_0=N_0(\vartheta)$ such that~$\tilde{\bPi}_{j}=\bI$ for all $j\geq N_0$, i.e.,  
$\bV(\vartheta)\bUt_j \bV(\vartheta)^\dagger$ is diagonal for sufficiently large~$j$. Equivalently, for all $j\geq N_0$, $\Ut_j$ acts diagonally in the basis $\cB_{\vartheta(\cC)}$, as well as in the basis~$\cB_{\cC}$ (by~\eqref{eq:elementsdef}).

The latter conclusion can be extended uniformly to
a finite collection~$\{\vartheta_k\}_{k=1}^M\subset\mcg_\Sigma$ of mapping class group elements: there is a constant $N=N(\vartheta_1,\ldots,\vartheta_M)$ such that
for all $j\geq N$,  the protected gate~$\Ut_j$ acts as a diagonal matrix in all bases $\cB_{\cC}$, $\cB_{\vartheta_1(\cC)}$, ..., $\cB_{\vartheta_M(\cC)}$. 
Taking a  finite collection~$\{\vartheta_k\}_{k=1}^M\subset\mcg_\Sigma$ that generates the relation $\sim$ 
and
applying Lemma~\ref{lem:trivialphasegatecharact}, we reach the conclusion that
$\Ut_j$ is a~$\sim$-trivial phase gate for all $j\geq N$. This contradicts the fact that each $\Ut_j$ is a $\sim$-non-trivial phase gate, as argued above. 
\end{proof}

\subsection{Necessity of restricting to equivalence classes\label{sec:equivalenceclasses}}
Here we briefly argue that without imposing $\sim$-equivalence on protected gates, one can end up with infinitely many protected gates (that are, however, not very interesting).

Concretely, consider a model such as the toric code, with local commuting projector Hamiltonian~$H_{top}=-\sum_{j}\Pi_j$ acting
on spins which we collectively denote by~$A$. Let $\cH_\Sigma$ be its ground space. We introduce a local spin-degree of freedom  $B_j$ associated with each term in the Hamiltonian, and let $B=\bigotimes_j B_j$ the space of these auxiliary degrees of freedom. Define an Ising-like Hamiltonian $H_{I}=-\sum_{\langle j,j'\rangle}Z_j Z_{j'}$ coupling all nearest neighbors in $B$ (according to some notion). Finally, consider the following Hamiltonian:
\begin{align}
H&= J\cdot H_I-\sum_j \Pi_j\otimes \proj{0}_{B_j}-\sum_j \Pi_j\otimes \proj{1}_{B_j}\ .\label{eq:hivbf}
\end{align}
This Hamiltonian is local, and for large $J$, has a ground space
of the form~
\begin{align}
(\cH_{\Sigma}\otimes\ket{00\cdots 0})\oplus (\cH_{\Sigma}\otimes\ket{11\cdots 1})\ .\label{eq:directsumgroundspace}
\end{align}
In other words, the ground space (and similarly the low-energy subspace) splits as~$\cH^{(0)}_{\Sigma}\oplus\cH^{(1)}_{\Sigma}$ into two isomorphic copies of the space~$\cH_\Sigma$.

Now take two arbitrary protected gates $U^{(0)},U^{(1)}$ for $H_{top}$ (these may be global phases, i.e., trivial), implementing logical operations $\overline{U}^{(0)}$, $\overline{U}^{(1)}$. Let us assume that they are implemented by circuits acting locally, i.e., they can be written (arbitrarily -- the details do not matter) in the form
\begin{align*}
U^{(m)}&=U^{(m)}_{j_1}U^{(m)}_{j_2}\cdots U^{(m)}_{j_{M_m}}
\end{align*}
with each unitary $U_{j}$ local near the support of $\Pi_j$.
Then we can define the unitary
\begin{align*}
U&=\prod_{k=1}^{M_0}\left(U^{(0)}_{j_k}\otimes\proj{0}_{B_{j_k}}+\id\otimes\proj{1}_{B_{j_k}}\right)\prod_{k=1}^{M_1}\left(\id\otimes\proj{0}_{B_{j_k}}+U^{(1)}_{j_k}\otimes\proj{1}_{B_{j_k}}\right)
\end{align*}
    on $A\otimes B$. It is easy to check that $U$ is a protected gate and its logical action is 
    \begin{align*}
    \overline{U}&=\overline{U}^{(0)}\oplus\overline{U}^{(1)}\ .
    \end{align*}
In particular, such a unitary can introduce an arbitrary relative phase between the ``superselection'' sectors $\cH_\Sigma^{(0)}$, $\cH_\Sigma^{(1)}$: we can choose $U^{(0)}=I$ and $U^{(1)}=e^{i\varphi}I$.  The construction here corresponds to the direct sum of two TQFTs; the mapping class group representation is reducible and basis elements belonging to different sectors are inequivalent. Imposing the relation~$\sim$ on the set of protected gates renders all such relative-phase gates equivalent.

A small caveat is in order here concerning
the given microscopic example. The Hamiltonian~\eqref{eq:hivbf} indeed has~\eqref{eq:directsumgroundspace}
as its ground space. However, the latter is not an error-correcting code: whether  a state belongs to~$\cH_{\Sigma}^{(0)}$ or~$\cH_{\Sigma}^{(1)}$ can be determined by a local measurement.  Thus
information should only be encoded in either one of the superselection sectors,
and this renders the introduction of (arbitrary) relative phases between two superselection sectors computationally trivial. The example given here is mainly intended to give a concrete realization of the space~\eqref{eq:directsumgroundspace} as the ground space of a local Hamiltonian, and to illustrate the fact that reducibility of the mapping class group representation has important consequences on the form of protected gates.

\section{Global constraints from $F$-moves on the $n$-punctured sphere\label{sec:globalconstraintsfmove}}
We first consider the four-punctured sphere, where there are
two inequivalent DAP-decompositions related by an $F$-move (i.e., the basis change~$\bV$ is the $F$-matrix). More generally (e.g., for the $5$-punctured sphere), we need to consider several  different $F$-moves and obtain a constraint of the form~\eqref{eq:consistencycondition} for every pair of bases related by such moves. We describe such global constraints in Section~\ref{sec:localizationoffphases}. The results obtained by considering $F$-moves are summarized in Section~\ref{sec:summaryprotected}: there we outline a general procedure for
characterizing protected gates.

The consideration of/restriction to $n$-punctured spheres  is motivated by the fact that they correspond to $n-1$~anyons situated on a disc. Realizing such a system appears to be more feasible experimentally than designing e.g., a higher-genus surface. For this reason, the $n$-punctured sphere is most commonly  considered in the context of topological quantum computation. We point out, however, that our techniques immediately generalize to other (higher-genus) surfaces with or without punctures (although  basis changes  other than those given by the $F$-matrix need to be considered).

\subsection{Determining phases for the four-punctured sphere: fixed boundary labels\label{sec:fourpuncturedgeneral}}
For a four-punctured sphere~$\Sigma$, we can fix the labels on the punctures to $i,j,k,l\in \labels$. The corresponding space~$\cH_{\Sigma(i,j,k,l)}$ associated to this open surface with labeled boundary components is the fusion space~$V^{ij}_{kl}$. (In the non-abelian case, this space can have dimension larger than $1$.) 
We have two bases $\cB_{\cC}$, $\cB_{\cC'}$ of this fusion space, corresponding to two different DAP-decompositions differing by one loop (Fig.~\ref{fig:Fmove}).
We can enumerate basis elements by the label assigned to this loop. Let $\set{\ket{a}_\cC}_a$ and $\set{\ket{a}_{\cC'}}_a$ be the elements of the bases~$\cB_{\cC}$ and $\cB_{\cC'}$, respectively. Note that $a$ ranges over all elements consistent with the fusion rules.
 
For the models considered in this article, these are $N_{ij}^{a}=N_{kl}^{a}=1$.
Let $Q=Q(i,j,k,l)$ be the set of such elements. The basis change is given by the $F$-matrix
\begin{equation*}
    \ket{m}_{\cC'}
    =\sum_n
     F^{ijm}_{kln}
     \ket{n}_\cC.
\end{equation*}
Considering a locality-preserving automorphism which preserves the boundary labels (this is reasonable if we think of them as certain boundary conditions of the system), we can apply the procedure explained above to find the action
\begin{equation*}
    U
    \ket{a}_\cC
    =\phaseExp{}{a}
     \ket{\permLoop{C}(a)}_\cC
\end{equation*} 
 on basis states. Here $\permLoop{C}\colon Q\rightarrow Q$ permutes fusion-consistent labels.  To apply the reasoning above, we have to use the $|Q\times Q|$-basis change matrix ${\bf V}$ defined by ${\bf V}_{m,n}=F^{ijm}_{kln}$.

Solving the consistency relation~\eqref{eq:consistencycondition}
(for the permutations $\permLoop{C}$, $\permLoop{C'}$ and phases $\{\phaseArg{}{a}\}_a,\{\phaseArg{'}{a}\}_a$)
shows that for any permutation~$\permLoop{C}$  that is part of a solution, the function~$\varphi$ takes the form
\begin{align}
\phaseArg{}{a}&=\eta+f(a)\ ,\label{eq:etafaphi}
\end{align}
where $\eta$ is a global phase and $f$ belongs to a certain set of functions
which we denote 
\begin{align}
\Fset{\fourdot{i}{j}{k}{l}}{\four{i}{j}{k}{l}{\permLoop{C}(\cdot)}}\ .\label{eq:permutationset}
\end{align}
(The reason for this notation will become clearer when we discuss isomorphisms in Section~\ref{sec:fourpuncturedchange}; here we are concerned with relative phases arising from automorphisms.)  
In summary, we have
\begin{align}
U\ket{a}_\cC&=e^{i\eta}e^{if(a)}\ket{\permLoop{C}(a)}_\cC\qquad\textrm{ where }f\in \Fset{\fourdot{i}{j}{k}{l}}{\four{i}{j}{k}{l}{\permLoop{C}(\cdot)}}\ .\label{eq:Uetafa}
\end{align}
Here the set~\eqref{eq:permutationset} can be computed by
solving the 
consistency relation
\begin{align}
{\bf V}{\bf \Pi}{\bf D}(\{\phaseArg{}{a}\}_a)&={\bf \Pi'}{\bf D}(\{\phaseArg{'}{a}\}_a){\bf V} \label{eq:auto4punctureconsistencyeq}
\end{align}
with ${\bf V}_{m,n}=F^{ijm}_{kln}$. This scenario is a special case of the commutative diagram displayed in Fig.~\ref{fig:Fmovesquaretilde}.

\subsection{Determining phases for the four-punctured sphere in general\label{sec:fourpuncturedchange}}
Consider the four-punctured sphere~$\Sigma$ with fixed labels $i,j,k,l\in \labels$ on the punctures. Let  $\z{\imath},\z{\jmath},\z{k},\z{l}$ be another set of labels such that
the spaces~$\cH_{\Sigma(i,j,k,l)}$  and~$\cH_{\Sigma(\z{\imath},\z{\jmath},\z{k},\z{l})}$ are isomorphic.
In this situation, we can try to characterize locality-preserving isomorphisms
between two systems defined on~$\Sigma(i,j,k,l)$ and~$\Sigma(\z{\imath},\z{\jmath},\z{k},\z{l})$, respectively.
This situation is slightly more general than what we considered before (automorphisms of the same system), but it is easy to see that all arguments applied so far extend to this situation.  Note that we could have phrased our whole discussion in terms of isomorphisms between different spaces. However, we chose not to do so to minimize the amount of notation required; instead, we only consider this situation in this section. This generalization for the $4$-punctured sphere is all we need to treat automorphisms on higher-genus surfaces.

For $\cH_{\Sigma(i,j,k,l)}$, we have two bases $\cB_{\cC}$, $\cB_{\cC'}$, corresponding to two different DAP-decompositions differing by one loop. Similarly,
for  $\cH_{\Sigma(\z{\imath},\z{\jmath},\z{k},\z{l})}$, we have two bases $\z{\cB}_{\cC}$, $\z{\cB}_{\cC'}$, corresponding to two different DAP-decompositions differing by one loop.  We can enumerate the basis elements by the label assigned to this loop. Let $\set{\ket{a}_\cC}_a$ and $\set{\ket{a}_{\cC'}}_a$ be the elements of the basis $\cB_{\cC}$ and $\cB_{\cC'}$, respectively. Here $a$ ranges over the set $Q=Q(i,j,k,l)\subset\labels$ of 
all elements consistent with the fusion rules, i.e., we must have $N_{ij}^{a}=N_{kl}^{a}=1$.
Similarly, let $\set{\ket{\tilde{a}}_\cC}_{\tilde{a}}$ and $\set{\ket{\tilde{a}}_{\cC'}}_{\tilde{a}}$ be the elements of the basis $\z{\cB}_{\cC}$ and $\z{\cB}_{\cC'}$, respectively, where now $\z{a}\in
\z{Q}=Q(\z{\imath},\z{\jmath},\z{k},\z{l})$.

In this situation, we have two basis changes, 
\begin{align*}
\ket{m}_{\cC'} = \sum_n {\bf V}_{m,n} \ket{n}_\cC \textrm{ where } {\bf V}_{m,n}=F^{ijm}_{kln}\qquad\textrm{ and }\qquad \ket{\tilde{m}}_{\cC'} = \sum_{\z{n}} {\bf \z{V}}_{\z{m},\z{n}} \ket{\tilde{n}}_\cC \textrm{ where } \z{\bf V}_{\z{m},\z{n}}=F^{\z{\imath}\z{\jmath}\z{m}}_{\z{k}\z{l}\z{n}} .
\end{align*}
Now consider a locality-preserving isomorphism~$U$ which takes
the boundary labels $(i,j,k,l)$  to $(\z{\imath},\z{\jmath},\z{k},\z{l})$. We can then apply
the framework above to find the action
\begin{align*}
U \ket{a}_\cC =\phaseExp{}{a} \ket{\permLoop{C}(a)}_\cC \qquad\textrm{ or }\qquad U \ket{a}_{\cC'} =\phaseExp{'}{a} \ket{\permLoop{C'}(a)}_{\cC'}
 \end{align*} 
 on basis states. Here $\permLoop{C},\permLoop{C'}\colon Q\rightarrow \z{Q}$ take
fusion-consistent labels on $\Sigma(i,j,k,l)$ to fusion-consistent labels on~$\Sigma(\z{\imath},\z{\jmath},\z{k},\z{l})$.  
Because the spaces are isomorphic, we must have~$|Q|=|\z{Q}|$, hence $\permLoop{C},\permLoop{C'}$ can be represented by permutation matrices~${\bf \Pi}, {\bf \Pi'}$ 
in the basis pairs $(\cB_\cC,\tilde{\cB}_\cC)$ or $(\cB_{\cC'},\tilde{\cB}_{\cC'})$, respectively.  
Proceeding similarly with ${\bf U}$, we get the consistency equation
${\bf\tilde{V}}\bU=\bU'\bV$
or 
\begin{align}
{\bf \z{V}}{\bf \Pi}{\bf D}(\{\phaseArg{}{a}\}_a)&={\bf \Pi'}{\bf D}(\{\phaseArg{'}{a}\}_a){\bf V},\label{eq:consistencyequationpermphases}
\end{align}
which is expressed in the form of a commutative diagram as in Fig.~\ref{fig:Fmovesquaretilde}.
Equation \eqref{eq:consistencyequationpermphases} only differs from equation \eqref{eq:consistencycondition} in allowing boundary labels to change and the basis transformation matrix ${\bf \z{V}}$ must change accordingly.

   \begin{figure}
\begin{center}
\includegraphics[width=0.4\textwidth]{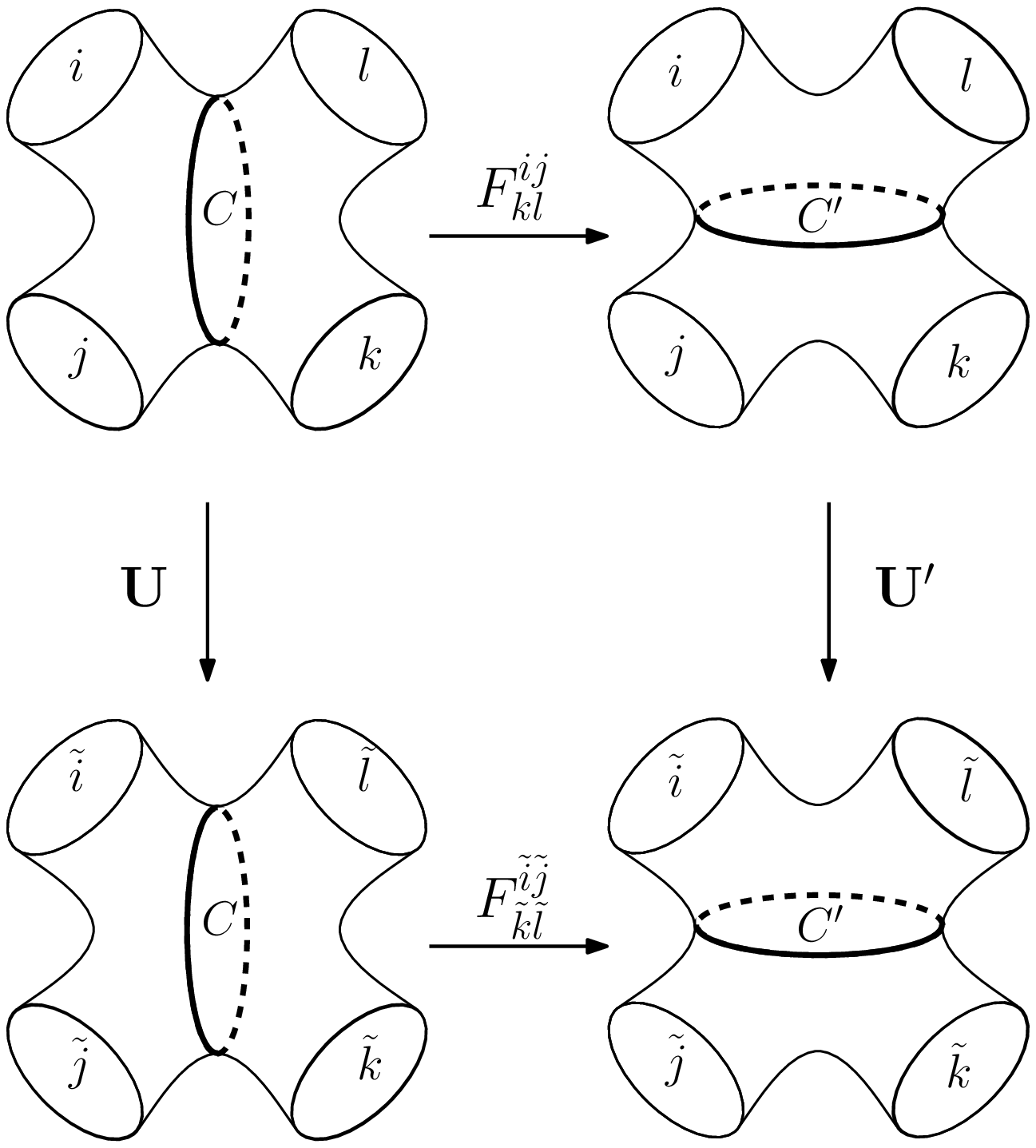}
\end{center}
\caption{ An isomorphism $\cH_{\Sigma(i,j,k,l)}\to \cH_{\Sigma(\z{\imath},\z{\jmath},\z{k},\z{l})}$ of two $4$-punctured spheres can be given as either $\mathbf{U}$, which relates the bases $\cB_C$ of $\cH_{\Sigma(i,j,k,l)}$ to $\tilde{\cB}_{C}$ of $\cH_{\Sigma(\z{\imath},\z{\jmath},\z{k},\z{l})}$, or as $\mathbf{U'}$ relating different bases $\cB_{C'}$ of $\cH_{\Sigma(i,j,k,l)}$ to $\tilde{\cB}_{C'}$ of $\cH_{\Sigma(\z{\imath},\z{\jmath},\z{k},\z{l})}$. The bases of $\cH_{\Sigma(i,j,k,l)}$ and $\cH_{\Sigma(\z{\imath},\z{\jmath},\z{k},\z{l})}$ are related through the $F$-moves $F^{ij}_{kl}$ and $F^{\z{\imath}\z{\jmath}}_{\z{k}\z{l}}$, respectively. The consistency equation \eqref{eq:consistencyequationpermphases} can be expressed as a commutative diagram. In the case where $\Sigma(i,j,k,l)=\Sigma(\z{\imath},\z{\jmath},\z{k},\z{l})$ have identical boundary labels such an isomorphism becomes an automorphism, and this reduces to the consistency equation \eqref{eq:auto4punctureconsistencyeq}.    }
\label{fig:Fmovesquaretilde} 
\end{figure} 

For a given set of boundary labels~$(i,j,k,l)$, $(\z{\imath},\z{\jmath},\z{k},\z{l})$, and a fixed choice of~$\permLoop{C}$ (which fixes~${\bf \Pi}$),  any solution
$({\bf \Pi'}, \{\phaseArg{}{a}\}_a, \{\phaseArg{'}{a}\}_a)$ of~\eqref{eq:consistencyequationpermphases} has phases~$\{\phaseArg{}{a}\}_a$  of the ``universal'' form
\begin{align}
\phaseArg{}{a} &=\eta+f(a)\qquad\textrm{ for all }a\in Q(i,j,k,l)\ ,
\end{align}
where $\eta\in [0,2\pi)$ is an arbitrary global phase independent of~$a$, 
and $f$ belongs to
a set $\Fset{\fourdot{i}{j}{k}{l}}{\four{\z{\imath}}{\z{\jmath}}{\z{k}}{\z{l}}{\permLoop{C}(\cdot)}}$ of functions that can be computed from~\eqref{eq:consistencyequationpermphases} as discussed below.

In summary, we have shown that~$U$ acts as
\begin{align}
U \ket{a}_\cC &=e^{i\eta}e^{if(a)} \ket{\permLoop{C}(a)}_\cC \qquad\textrm{ with }\qquad f\in\Fset{\fourdot{i}{j}{k}{l}}{\four{\z{\imath}}{\z{\jmath}}{\z{k}}{\z{l}}{\permLoop{C}(\cdot)}}\ ,\label{eq:psiactionphaser}
\end{align} 
and where the latter set can be determined by
solving the consistency relation~\eqref{eq:consistencyequationpermphases}.

\subsection{Localization of  phases for higher-genus surfaces\label{sec:localizationoffphases}}
We now argue that the phases appearing in Eq.~\eqref{eq:upsielltransformation} of Proposition~\ref{prop:globalfusionconstraint} also factorize into certain essentially local terms, similar to how the overall permutation~$\perms$ of fusion-consistent labelings decomposes into a collection~$\perms=\{\permLoop{C}\}_{C\in\cC}$  of permutations of labels. 
More precisely, we will argue that conclusion~\eqref{eq:psiactionphaser} 
can be extended to more general surfaces.

\newcommand*{\neighbors}[1]{N(#1)}
Consider a fixed DAP-decomposition~$\cC$ of~$\Sigma$. 
We call a curve~$C\in\cC$ {\em internal} if the intersection of~$\Sigma$ with a ball containing~$C$ has the form of a $4$-punctured sphere with
boundary components~$C_1,C_2,C_3,C_4$
consisting of curves `neighboring'~$C$ in the DAP decomposition. 
We call $\neighbors{C}=\{C_1,C_2,C_3,C_4\}$  the neighbors (or neighborhood) of $C$ as illustrated in Fig.~\ref{fig:Npunctureneighbors}. Key to the following observations is that
a basis vector~$\ket{\labeling{\ell}}$ whose restriction to these neighbors
is given by $\labeling{\ell}\restriction\neighbors{C}=\bigl(\labeling{\ell}(C_1),\dots,\labeling{\ell}(C_4)\bigr)$
gets mapped under $U$ to a vector proportional to~$\ket{\perms(\labeling{\ell})}$, which assigns the labels $\perms(\labeling{\ell})\restriction \neighbors{C}=\bigl(\permLoop{C_1}[\labeling{\ell}(C_1)],\dots,\permLoop{C_4}[\labeling{\ell}(C_4)]\bigr)$
to the same curves. This means that the restriction of $U$ to this subspace
satisfies similar consistency conditions as the  isomorphisms between
Hilbert spaces associated with the $4$-punctured spheres  $\Sigma(\labeling{\ell}\restriction N(C))$ and $\Sigma\bigl(\perms(\labeling{\ell})\restriction N(C)\bigr)$ studied in Section~\ref{sec:fourpuncturedgeneral}. In particular, for a fixed labeling~$\labeling{\ell}$ the dependence of the phase~$\phaseArg{}{\labeling{\ell}}$ on the label~$\labeling{\ell}(C)$ is given by
a function from the set
$\Fset{\fourdot{i}{j}{k}{l}}{\four{\z{\imath}}{\z{\jmath}}{\z{k}}{\z{l}}{\permLoop{C}(\cdot)}}$,
where $(i,j,k,l)=\labeling{\ell}\restriction N(C)$ and
$(\z{\imath},\z{\jmath},\z{k},\z{l})=\perms(\labeling{\ell})\restriction N(C)$.
In the following, we simply write
$ \fset{\labeling{\ell}\restriction N(C)}{\perms(\labeling{\ell})\restriction N(C)}{\permLoop{C}}$ for this set.

\begin{figure}
\begin{center}
\includegraphics[width=0.3\textwidth]{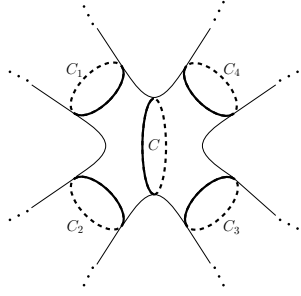}
\end{center}
\caption{For some DAP-decomposition $\cC$ of a surface $\Sigma$, a curve $C\in\cC$ is considered internal if its neighbors $N(C)=\{C_1,C_2,C_3,C_4\}$ define the boundaries of a $4$-punctured sphere.  }
\label{fig:Npunctureneighbors} 
\end{figure}

 \begin{proposition}[Localization of internal phases]\label{prop:internalphases}
 Let $U$ be a locality-preserving automorphism.   Let $\cC$ be a DAP-decomposition of $\Sigma$, and let~$\perms=\{\permLoop{C}\}_{C\in\cC}$ be the family of permutations defined by Proposition~\ref{prop:mainsingleloop}. Let $\phaseArg{}{\labeling{\ell}}$ for $\labeling{\ell}\in\mathsf{L}(\cC)$ be defined by~\eqref{eq:upsielltransformation}.
If $C\in\cC$ is internal, then
\begin{align*}
\phaseArg{}{\labeling{\ell}}&=\eta(\labeling{\ell}{\restriction \cC\backslash\{C\}} )+
f_{\perms\restriction N(C)}(\labeling{\ell}\restriction N(C),\labeling{\ell}(C))
\end{align*}
for some functions~$\eta$ and $f$. Furthermore, we have
\begin{align*}
f_{\perms\restriction N(C)}(\labeling{\ell}\restriction N(C), \cdot )\in \fset
{\labeling{\ell}\restriction N(C)}{\perms(\labeling{\ell})\restriction N(C)}{\permLoop{C}}\ .
\end{align*}
  In particular, the dependence of~$\phaseArg{}{\labeling{\ell}}$ on $\labeling{\ell}(C)$ is ``local'' and ``controlled'' by the 
labeling $\labeling{\ell}\restriction N(C)$ of the neighbors.
 \end{proposition}
In other words, if we fix a family of permutations~$\perms$, 
and the labels on the neighbors~$N(C)$, then the dependence on the label~$\labeling{\ell}(C)$ of the internal edge is essentially fixed.

\begin{proof}
We will focus our attention on the subspace~$\cH_{(i,j,k,l,\star)}\subseteq\cH_\Sigma$ spanned by labelings~$\labeling{\ell}$
 with $(\labeling{\ell}(C_1),\labeling{\ell}(C_2),\labeling{\ell}(C_3),\labeling{\ell}(C_4))=(i,j,k,l)$ and $\labeling{\ell}\restriction \mathcal{C}\backslash\{C,C_1,C_2,C_3,C_4\}=\star$ fixed (arbitrarily).  For the purpose of this proof, it will be convenient to  represent basis vectors~$\ket{\labeling{\ell}}$ associated with such a labeling~$\labeling{\ell}\in\mathsf{L}(C)$ as a vector
\begin{align*}
\ket{\labeling{\ell}}=\ket{\labeling{\ell}(C),\labeling{\ell}(C_1),\labeling{\ell}(C_2),\labeling{\ell}(C_3),\labeling{\ell}(C_4),\star}&=\ket{a,i,j,k,l,\star}\ .
\end{align*}
Defining $\z{\imath}=\permLoop{C_1}(i)$, $\z{\jmath}=\permLoop{C_2}(j)$, $\z{k}=\permLoop{C_3}(k)$, $\z{l}=\permLoop{C_4}(l)$, we can rewrite~\eqref{eq:upsielltransformation} in the form\begin{align*}
U\ket{a,i,j,k,l,\star}&=\phaseExp{}{a,i,j,k,l,\star} \ket{\permLoop{C}(a),\z{\imath},\z{\jmath},\z{k},\z{l},\z{\star}}\ ,
\end{align*}
where $\z{\star}=\perms_\restriction(\star)$ for some map~$\perms_\restriction$ taking labelings of the set $\cC\backslash\{C, C_1,C_2,C_3,C_4\}$ consistent with $(i,j,k,l)$ to those consistent with $(\z{\imath},\z{\jmath},\z{k},\z{l})$. We conclude that
the restriction of $U$ to $\cH_{(i,j,k,l,\star)}$ implements
an isomorphism $\cH_{(i,j,k,l,\star)}\cong \cH_{(\z{\imath},\z{\jmath},\z{k},\z{l},\z{\star})}$. 
Since these spaces are isomorphic to~$\cH_{\Sigma(i,j,k,l)}$ and 
$\cH_{\Sigma(\z{\imath},\z{\jmath},\z{k},\z{l})}$, respectively, we can apply the result of
Section~\ref{sec:fourpuncturedchange}. Indeed, the consistency relation
imposed by the $F$-move is entirely local, not affecting 
labels associated with curves not belonging to~$\{C,C_1,C_2,C_3,C_4\}$. We conclude
from~\eqref{eq:psiactionphaser} that
\begin{align*}
\phaseArg{}{a,i,j,k,l,\star}&=\eta(i,j,k,l,\star)+f(a), \textrm{ where }f\in \Fset{\fourdot{i}{j}{k}{l}}{\four{\z{\imath}}{\z{\jmath}}{\z{k}}{\z{l}}{\permLoop{C}(\cdot)}}\ .
\end{align*}
Since $(a,i,j,k,l,\star)$ were arbitrary, this proves the claim. 
\end{proof}

For example, for $S^2(z^{N+3})$ (as described above), we can apply Proposition~\ref{prop:internalphases} to the $j$-th internal edge~$C_j$ to obtain
\begin{align}
\phaseArg{}{\labelSequenceSphere{x}}&=\eta_j (x_1,\ldots,\widehat{x}_j,\ldots,x_N)+
f_j(x_{j-1},x_j,x_{j+1}),
\end{align}
where
\begin{align} f_j(x_{j-1},\cdot,x_{j+1})\in \Fset{\fourdot{z}{x_{j-1}}{x_{j+1}}{z}}{\four{z}{\z{x}_{j-1}}{\z{x}_{j+1}}{z}{\permLoop{C_j}(\cdot) }},\nonumber
\end{align}
and
\begin{equation*}
   \z{x}_{j-
          1}
   =\permLoop{C_{j-
                 1}}(x_{j-
                        1}),
   \qquad
   \z{x}_{j+
          1}
   =\permLoop{C_{j+
                 1}}(x_{j+
                        1}).
\end{equation*}
Here, we use $\hat{x}_j$ to indicate that this argument is omitted.

\subsection{Characterizing protected gates on the $M$-punctured sphere using $F$-moves\label{sec:summaryprotected}}
The results in this section give the following procedure for characterizing protected gates
associated with $\cH_{S^2(z^M)}$, the Hilbert space of $M=N+3$ anyons of type~$z$. We know from
 Proposition~\ref{prop:globalfusionconstraint} 
that  the action~$U\ket{\labeling{\ell}}=\phaseExp{}{\labeling{\ell}}\ket{\perms(\labeling{\ell})}$ on fusion-consistent labelings is parametrized by certain families $\perms=\{\permLoop{C}\}_{C\in\cC}$ of permutations, as well as a function~$\varphi$ describing the phase-dependence. To characterize the latter, we  
\begin{enumerate}[(i)] 
\item
determine the set of allowed `local' permutations~$\permLoop{C}$ and associated phases~$f$ for any occuring internal curve~$C$. This amounts to solving the consistency equation~\eqref{eq:consistencyequationpermphases} for the four-punctured sphere, with appropriate boundary labels. 
For the standard pants decomposition of the $N+3$-punctured sphere, 
this means finding all pairs 
\begin{align*}
\left(\permLoop{C_j}, f_j\right)\qquad\textrm{ where }f_j\in \Fset{\fourdot{z}{x_{j-1}}{x_{j+1}}{z}}{\four{z}{\z{x}_{j-1}}{\z{x}_{j+1}}{z}{\permLoop{C_j}(\cdot) }}\ .
\end{align*}
These correspond to isomorphisms between the Hilbert spaces associated with the labeled surfaces  $S^2(z,x_{j-1},x_{j+1},z)$
and $S^2(z,\z{x}_{j-1},\z{x}_{j+1},z)$, where $x_{j-1},\z{x}_{j-1}\in Q(j-1),x_{j+1},\z{x}_{j+1}\in Q(j+1)$.
\item
we constrain the family $\perms=\{\permLoop{C}\}_{C\in\cC}$ of allowed permutations by using
the global constraints arising from fusion rules and gluing (Proposition~\ref{prop:gluingconstraint}). In the case of $N+3$ Fibonacci anyons on the sphere with standard pants decomposition~$\cC$, dimensional arguments show that all $\permLoop{C_j}=\id$ are equal to the identity permutation. For Ising anyons, the fusion rules imply that every permutation with even index is equal to the identity permutation,~$\permLoop{C_{2j}}=\id$ (in fact, there is only a single allowed label). 
\item
we determine the phases~$\phaseArg{}{\labeling{\ell}}$ by using
the localization property of Proposition~\ref{prop:internalphases} for internal curves~$C$. For $N+3$ anyons of type $z$ on the sphere, this results in the consistency conditions
\begin{align}\phaseArg{}{\labelSequenceSphere{x}}&=\eta_j (x_1,\ldots,\widehat{x}_j,\ldots,x_N)+
f_j(x_{j-1},x_j,x_{j+1})\qquad\textrm{ where }\nonumber\\
 f_j(x_{j-1},\cdot,x_{j+1})&\in \Fset{\fourdot{z}{x_{j-1}}{x_{j+1}}{z}}{\four{z}{\z{x}_{j-1}}{\z{x}_{j+1}}{z}{\permLoop{C_j}(\cdot) }}\qquad\textrm{ for } j=1, \ldots, N\ .\label{eq:systemfjphases}
\end{align}
\end{enumerate}
In Section~\ref{sec:isingnpunctured}, we apply this procedure to Ising anyons; in this case, the system of equations~\eqref{eq:systemfjphases} can be solved explicitly.

\section{The Fibonacci and Ising models\label{sec:examplesnonabelian}}
In what follows, we apply the results of the previous sections to the Fibonacci and Ising models. These can be considered as representative examples of non-abelian anyon models.  We illustrate the use of the developed constraints in different scenarios:

In Section~\ref{sec:fibtorus}, we show that there is no non-trivial gate for the Fibonacci model on the torus. This derivation uses the characterization of protected gates in terms of matrices intertwining with the mapping class group representation obtained in Section~\ref{sec:basischangesmcg}.
Note that we cannot apply Corollary~\ref{cor:density} because the representation of the mapping class group on the torus is finite for the Fibonacci model. 

In Section~\ref{sec:fibnpunctured}, we then consider  a system with $M$~Fibonacci anyons (where $M\geq 4$ so that the space $\cH_{S^2(\tau^{M})}$ has non-zero dimension). We establish the following statement:
\begin{theorem}[Fibonacci anyon model]\label{thm:fib}
For $M\geq 4$, any  locality-preserving automorphism $U$ on the $M$-punctured sphere~$S^2(\tau^{M})$ is trivial (i.e., proportional to the identity).
\end{theorem}
This proof is a direct consequence of Corollary~\ref{cor:density} and the known  density of braiding~\cite{freedmanlarsenwang02,freedmankitaevlarsenwang03}. We additionally provide an independent proof not relying on this result.

Finally, we consider systems with $M$ Ising anyons; the associated Hilbert space~$\cH_{S^2(\sigma^{M})}$ has non-zero dimension if and only if $M\geq 4$ is even. 
In this case, there is a natural isomorphism~$\cH_{S^2(\sigma^{M})}\cong(\mathbb{C}^2)^{\otimes M/2-1}$ (described below, see Eq.~\eqref{eq:isingisomorphism}). 
Defining the $(M/2-1)$-qubit Pauli group on the latter space in the usual way, we get the following statement:
\begin{theorem}[Ising anyon model]\label{thm:ising}
 Any  locality-preserving automorphism $U$ of $S^2(\sigma^{M})$, where $M\geq 4$ is even,  belongs to the $(M/2-1)$-qubit Pauli group.  
\end{theorem}
Our derivation of this result relies on the use of $F$-moves, as discussed in Section~\ref{sec:globalconstraintsfmove}.

\subsection{The  Fibonacci model}
For the Fibonacci model,  we have $\labels=\{1,\tau\}$ and the only non-trivial fusion rule is $\tau\times\tau=1+\tau$ with $d_\tau=\phi=(1+\sqrt{5})/2$.

\subsubsection{On the torus\label{sec:fibtorus}}
We first consider
the torus~$\Sigma$ and show that every protected gate is trivial. We do so by computing some of the sets~$\Delta_{\vartheta}$, $\vartheta\in\mcg_{\Sigma}$ defined in Section~\ref{sec:basischangesmcg}. Recall (see Section~\ref{sec:mcg}) that the mapping class group of the torus is generated by two elements~$s$, $t$. 

The matrix $\bV(s)=S$ representing $s$ is the usual $S$-matrix
(expressed with respect to the ordering $(1,\tau)$)
\begin{equation*}
    S
    =\frac{1}
          {\sqrt{\phi+
                 2}}
     \left(\begin{matrix}
                 1  & \phi \\
               \phi &  -1
           \end{matrix}\right).
\end{equation*}
In particular, the consistency condition~\eqref{eq:compatibilitycond} becomes
\begin{equation*}
    S
    \bPi
    \bD
    S^{-1}
    \in
     \phiperm,
\end{equation*}
where $\bD=\diag(\phase{}{1},\phase{}{\tau})$ and $\phase{}{a}\in\unitaryGroup{1}$.
We consider the two cases:
\begin{enumerate}
\item
For $\bPi=I$, we get (using $\phi^2=\phi+1$)
\begin{equation*}
    S
    \bPi
    \bD
    S^{-1}
    =\frac{1}
          {\phi+
           2}
     \left(\begin{matrix}
               \phase{}{1}+\phase{}{\tau}(\phi+1) & (\phase{}{1}-\phase{}{\tau})\phi \\
               (\phase{}{1}-\phase{}{\tau})\phi   & \phase{}{1}(\phi+1)+\phase{}{\tau}
           \end{matrix}\right).
\end{equation*}
For this to be a unitary monomial matrix, all entries must have modulus~$0$ or~$1$. 
Since $\phi/(\phi+2)<1/2$, the off-diagonal elements always have modulus less than~$1$, and hence must be zero.
That is, we must have $\phase{}{1}=\phase{}{\tau}\eqqcolon \phaseGlobal$, and it follows that the right hand side is in~$\phiperm$.
This implies that $\bPi\bD=\phaseGlobal I$.

\item
For ${\bf \Pi}=\left(\begin{matrix}
0 & 1\\
1 & 0
\end{matrix}\right)$, we get
\begin{equation*}
    S
    \bPi
    \bD
    S^{-1}
    =\frac{1}
          {\phi+
           2}
     \begin{pmatrix}
         (\phase{}{1}+\phase{}{\tau})\phi   & \phase{}{1}(\phi+1)-\phase{}{\tau} \\
         \phase{}{\tau}(\phi+1)-\phase{}{1} & -(\phase{}{1}+\phase{}{\tau})\phi
     \end{pmatrix}.
\end{equation*}
To have the absolute value of the first entry equal to $0$ (see above), we must have
$\phase{}{\tau}=-\phase{}{1}$ and we get
\begin{equation*}
    S
    \bPi
    \bD
    S^{-1}
    =\phase{}{1}
     \begin{pmatrix}
          0 & 1 \\
         -1 & 0
     \end{pmatrix},
\end{equation*}
which is a unitary monomial matrix.
That is, we have $\bPi\bD=\phaseGlobal\left(\begin{matrix}
0 &1 \\
-1 & 0
\end{matrix}\right)$.
\end{enumerate}
Summarizing, we conclude that
\begin{equation}
    \label{eq:deltasfib}
    \phiperm_s
    =\setcomprehension[\bigg]{\phaseGlobal
                              I,
                              \phaseGlobal
                              \begin{pmatrix}
                                   0 & 1 \\
                                  -1 & 0
                              \end{pmatrix}}
                             {\phaseGlobal
                              \in
                               \unitaryGroup{1}}.
\end{equation}

\bigskip

The element $t\in\mcg_\Sigma$ defined by twisting along one of the homologically non-trivial cycles is represented by the matrix~$\bV(t)=T=\diag(1,e^{4\pi i/5})$. We consider the  consistency condition~\eqref{eq:compatibilitycond} for the composition~$st\in\mcg_\Sigma$:
\begin{equation*}
    (S
     T)
    \bPi
    \bD
    (S
     T)^{-1}
    \in
     \phiperm,
\end{equation*}
where $\bD=\diag(\phase{}{1},\phase{}{\tau})$ and $\phase{}{a}\in\unitaryGroup{1}$.
Again, we consider the following two cases:
\begin{enumerate}
\item
For $\bPi=I$, we get
\begin{equation*}
    (S
     T)
    \bPi
    \bD
    (S
     T)^{-1}
    =\frac{1}
          {\phi+
           2}
     \begin{pmatrix}
         \phase{}{1}+\phase{}{\tau}(\phi+1) & (\phase{}{1}-\phase{}{\tau})\phi\\
         (\phase{}{1}-\phase{}{\tau})\phi   & \phase{}{1}(\phi+1)+\phase{}{\tau}
     \end{pmatrix}.
\end{equation*}
This is identical to the first case above, thus $\bPi\bD=\phaseGlobal I$.
\item
For $\bPi=\left(\begin{matrix}0 & 1\\
1 & 0\end{matrix}\right)$, we get
\begin{equation*}
    (S
     T)
    \bPi
    \bD
    (S
     T)^{-1}
    =\frac{\zeta}
          {\phi+
           2}
     \begin{pmatrix}
         (\zeta^3\phase{}{1}-\phase{}{\tau})\phi    & \zeta^3\phase{}{1}(\phi+1)+\phase{}{\tau} \\
         -\zeta^3\phase{}{1}-\phase{}{\tau}(\phi+1) & -(\zeta^3\phase{}{1}-\phase{}{\tau})\phi
     \end{pmatrix},
\end{equation*}
where $\zeta=\mathrm{e}^{\mathrm{i}\pi/5}$.
Since $\phi/(\phi+2)<1/2$, the diagonal elements must vanish, that is, we have $\phase{}{\tau}=\zeta^3\phase{}{1}$.
This indeed then gives an element of~$\phiperm$, and
$\bPi\bD=
\phaseGlobal
\begin{pmatrix}
0 & \mathrm{e}^{3\pi\mathrm{i}/5}\\
1 & 0
\end{pmatrix}$.
\end{enumerate}
In summary, we have shown that
\begin{equation}
    \phiperm_{s
              t}
    =\setcomprehension[\bigg]{\phaseGlobal
                              I,
                              \phaseGlobal
                              \begin{pmatrix}
                                  0 & \mathrm{e}^{3\pi\mathrm{i}/5} \\
                                  1 &               0
                              \end{pmatrix}}
                             {\phaseGlobal
                              \in
                               \unitaryGroup{1}}.
    \label{eq:deltastfib}
\end{equation}
Combining~\eqref{eq:deltasfib} and~\eqref{eq:deltastfib}, we conclude that
\begin{equation*}
    \phiperm_s\cap
    \phiperm_{st}
    =\setcomprehension{\phaseGlobal
                       I}
                      {\phaseGlobal
                       \in
                        \unitaryGroup{1}},
\end{equation*}
and this means that
$\phiperm_{\mcg_\Sigma}\subset\phiperm_s\cap\phiperm_{st}=\setcomprehension{\phaseGlobal I}{\phaseGlobal\in\unitaryGroup{1}}$. According to Theorem~\ref{thm:mainmappinglcassgroup}, this implies that there is no non-trivial protected gate on the torus.

 Note that this  conclusion is consistent with the form of a Dehn twist, given by the logical unitary~$U=\diag(1,e^{4\pi i/5})$
 (with the `topological' phases or twists on the diagonal): Dehn twists do \emph{not} preserve locality!
 For example, for a Dehn twist along~$C_1$, an operator supported on~$C_2$ may end up with support in the neighborhood of the union $C_1\cup C_2$ under conjugation by the unitary realizing the Dehn twist.

\subsubsection{On the $M$-punctured sphere\label{sec:fibnpunctured}}
We now provide a proof of Theorem~\ref{thm:fib}. As already mentioned,
braiding of $M\geq 4$~Fibonacci anyons is known to be universal~\cite{freedmanlarsenwang02,freedmankitaevlarsenwang03}, hence we could invoke Corollary~\ref{cor:density}. Instead, we give a different proof by exploiting the equivalence relation introduced in Section~\ref{sec:diagonalgates} and analyzing the dimension of the associated spaces (i.e., using the constraints arising from the gluing postulate, see Section~\ref{sec:globalconstraintsgluingaxiom}).

Consider the $M$-punctured sphere $\Sigma=S^2(\tau^{M})$ corresponding to $M$ Fibonacci anyons.  We will use as our `standard' basis the one arising from the standard DAP decomposition~$\cC$ of the $M$-punctured sphere introduced in Section~\ref{sec:gluing} (see Fig.~\ref{fig:standardpants}). We then have the following statement:

\begin{lemma}\label{lem:connectednessfib}
There is only one equivalence class under the relation~$\sim$.
Furthermore, the set of braids~$\{\sigma_{j}\}_{j=1}^{M-1}$ generates the relation~$\sim$. 
\end{lemma}
\begin{proof}
Let $\labelSequenceSphere{x}$ and $\labelSequenceSphere{x}'$ be two fusion-consistent labelings that are related by interchanging $\tau=x_j$ and $1=x_j'$ (or vice versa) in the $j$-th entry (but are otherwise the same). Fusion-consistency implies that
$x_{j-1}=x'_{j-1}=x_{j+1}=x'_{j+1}=\tau$. In particular, the relevant braid matrix describing the action of $V(\sigma_j)$ is $B(\tau,\tau)$ which has non-zero entries everywhere. We conclude that
\begin{equation*}
    \langle
     \labelSequenceSphere{x}'
    \vert
     V(\sigma_j)
    \vert
     \labelSequenceSphere{x}\rangle
    \neq
     0
    \qquad
    \text{and}
    \qquad
    \langle
     \labelSequenceSphere{x}'
    \vert
     V(\sigma_j)
    \vert
     \labelSequenceSphere{x}'\rangle\neq 0\ .
\end{equation*}
This implies that $\labelSequenceSphere{x}\connected{\sigma_j}\labelSequenceSphere{x}'$.
Since any fusion-consistent labeling can be obtained from the sequence~$\tau^N=(\tau,\ldots,\tau)$ by such interchanges, 
we conclude that any two fusion-consistent labelings are equivalent. That is, there is only one equivalence class under~$\sim$.
\end{proof}

We will now argue that the conditions of Lemma~\ref{lem:trivialphasegatecharact}~\eqref{it:simtrivialproperty} apply in this situation: that is, any protected gate~$U$ acts diagonally in any of the bases~$\cB_{\sigma_j(\cC)}$ obtained from the standard DAP-decomposition by applying a braid group generator~$\sigma_j$. In fact, we will argue more generally that~$U$ acts diagonally in any basis defined by a DAP-decomposition.

To do so, consider first the standard DAP-decomposition and the  spaces~$\cH_{\Sigma'_j(a,a)}$ for $j\in\{1, \dots, M-3\}$ and $a\in\{1,\tau\}$ (cf.~\eqref{eq:cutspheres}), where $\Sigma'_j$ is obtained from~$\Sigma$ by cutting along the curve~$C_j$ which leaves a $j+2$-punctured and a $(M-j)$-punctured sphere, respectively.  Note that $\tau$ is its own antiparticle ($\dual{\tau}=\tau$), and hence it suffices to consider~$\Sigma'_j(\tau,\tau)$ and $\Sigma'_j(1,1)$. 
 Our goal is to identify pairs $(a,\z{a})$ such that  $\cH_{\Sigma'_j(a,a)}\cong \cH_{\Sigma'_j(\z{a},\z{a})}$ are isomorphic, this being a necessary condition for a permutation satisfying~$\permLoop{C_j}(a)=\z{a}$ (see Proposition~\eqref{prop:gluingconstraint} and Eq.~\eqref{eq:dimisomorphismcutspaces}).  
To compute $\dim \cH_{\Sigma'_j(a,a)}$ for $a\in\{1,\tau\}$, we make use of the general fact that $\dim\cH_{S^2(\tau^M)}=\Phi_{M-1}$ where $\Phi_M$ denotes the $M$-th Fibonacci number, 
starting with $\Phi_0=0$ and $\Phi_1=1$ and satisfying the recurrence relation $\Phi_{M+1}=\Phi_M+\Phi_{M-1}$. 
From~\eqref{eq:cutspheres}, we obtain 
$\dim \cH_{\Sigma'_j(1,1)}=\Phi_{j}\Phi_{M-j-2}$ and $\dim \cH_{\Sigma'_j(\tau,\tau)}=\Phi_{j+1}\Phi_{M-j-1}$,  excluding the case
$j=1=M-3$ which satisfies $\dim \cH_{\Sigma'_1(1,1)}=\Phi_1\Phi_{M-3} = \dim \cH_{\Sigma'_{M-3}(1,1)}$ and $\dim \cH_{\Sigma'_1(\tau,\tau)}=\Phi_2\Phi_{M-2} = \dim \cH_{\Sigma'_{M-3}(\tau,\tau)}$, it follows from the monotonicity and positivity of $\Phi$ that
\begin{align}
\dim \cH_{\Sigma'_j(1,1)} < \dim \cH_{\Sigma'_j(\tau,\tau)}\qquad\textrm{ for }M>4, \textrm{ and all } j \in \{1,...,M-3 \}.
\label{eq:dimensiondifferencebound}
\end{align}
 Hence, according to the consistency condition~\eqref{eq:dimisomorphismcutspaces}, for $M>4$,  we only get an isomorphism $\cH_{\Sigma'(a,\dual{a})}\cong \cH_{\Sigma'(\permLoop{C}(a),\dual{\permLoop{C}(a)})}$ with $\permLoop{C}=\id$ being trivial for any internal loop $C$ in a standard DAP decomposition.  This shows that a protected gate acts diagonally in the standard basis.

Observe that this argument only involved the dimensions of the fusion spaces obtained by cutting along a curve~$C_j$ in  the pants decomposition.  Since it is generally true that cutting along a curve will decompose the $M$-punctured sphere into an $j+2$-punctured and a $(M-j)$-punctured sphere, respectively (for some $j$), the argument extends to arbitrary DAP-decompositions. In particular, $U$ is diagonal with respect to each of the bases~$\cB_{\sigma_j(\cC)}$, as claimed.

We have shown that
the conditions of Lemma~\ref{lem:trivialphasegatecharact} apply.
With Lemma~\ref{lem:connectednessfib}, Theorem~\ref{thm:fib} is immediate.

\subsection{The Ising model\label{sec:isingnpunctured}}
\newcommand{\SUM}[2]{\displaystyle\sum_{#1}^{#2}}

The Ising anyon model has label set $\mathbb{A} =\{1, \psi, \sigma\}$ and non-trivial fusion rules
\[
 \psi\times\psi=1,   \ \ \psi\times\sigma=\sigma,  \ \ \sigma\times\sigma=1+\psi.
\]

\subsubsection{On the 4-punctured sphere\label{sec:isingfourpuncturedcase}}

Consider the possible spaces $\cH_{S^2(\sigma,j,k,\sigma)}$ for $\{j,k\}\in\labels$, and observe that fusion consistency implies
\begin{displaymath}
\dim \cH_{S^2(\sigma,j,k,\sigma)}=
\left\{
\begin{array}{ll}
0 \ \ \text{if} \ \ j\neq k=\sigma \ \text{or}\ k\neq j=\sigma \\
1 \ \ \text{if} \ \ j,k\in\{1,\psi\}, \\
2 \ \ \text{if} \ \  j=k=\sigma.
\end{array}
\right.
\end{displaymath}
Therefore, the only nontrivial case to consider is $\cH_{S^2(\sigma,\sigma,\sigma,\sigma)}=\cH_{S^2(\sigma^4)}$ with an ordered basis $\set{\ket{1},\ket{\psi}}$. A locality-preserving automorphism of $\cH_{S^2(\sigma^4)}$ will act as
\begin{align*}
U\ket{a}&=e^{i\eta}e^{if(a)}\ket{\permLoop{C}(a)}\qquad\textrm{ where }f\in \Fset{\fourdot{\sigma}{\sigma}{\sigma}{\sigma}}{\four{\sigma}{\sigma}{\sigma}{\sigma}{\permLoop{C}(\cdot)}}\ 
\end{align*}
A valid permutation~$\permLoop{C}$ of $\{1,\psi\}$ that defines the action of $U$, and the set of phases can be determined as follows. Let $\mathcal{B}_\cC=\set{\ket{1}_\cC,\ket{\psi}_\cC}$ and $\mathcal{B}_{\cC'}=\set{\ket{1}_{\cC'},\ket{\psi}_{\cC'}}$ be corresponding ordered bases of $\cH_{S^2(\sigma^4)}$ for the two DAP-decomposition $\mathcal{C}$ and $\mathcal{C'}$, respectively. The $F$-matrix relating these two bases is given in the ordered basis~$\mathcal{B}_\cC$ as
\begin{equation*}
F=\frac{1}{\sqrt{2}}\begin{pmatrix} 1 & 1 \\ 1 & -1 \end{pmatrix}.
\end{equation*}

Now consider some locality-preserving automorphism~$U$ expressed in the bases~$\mathcal{B}_\cC$ and~$\mathcal{B}_{\cC'}$ as $\bU=\bPi\bD$ and $\bU'=\bPi'\bD'$ respectively,
for some $2\times 2$ permutation matrices $\mathbf{\Pi}, \mathbf{\Pi'}$ and diagonal matrices $\bD=\diag(\phase{}{1},\phase{}{\psi})$ and $\bD'=\diag(\phase{'}{1},\phase{'}{\psi})$ with phases $\phase{}{a},\phase{'}{a}\in\unitaryGroup{1}$.
Then the consistency relation takes the form $\mathbf{U'}=F\mathbf{U}F^{-1}$. Next, we find all consistent solutions for a given permutation $\mathbf{\Pi}$.

\begin{enumerate}
    \item For $\bPi=I$, we get
        \begin{equation}
            F
            \bPi
            \bD
    F^{-1}
    =\frac{1}
          {2}
     \begin{pmatrix}
         \phase{}{1}+\phase{}{\psi} & \phase{}{1}-\phase{}{\psi} \\
         \phase{}{1}-\phase{}{\psi} & \phase{}{1}+\phase{}{\psi}
     \end{pmatrix}
    =\bPi'
     \bD'.
    \label{eq:ising4punctureconstraint1}
\end{equation}

Suppose that $\bPi'=I$. Then the consistency relation~\eqref{eq:ising4punctureconstraint1} becomes 
\begin{equation*}
    \frac{1}
         {2}
    \begin{pmatrix}
        \phase{}{1}+\phase{}{\psi} & \phase{}{1}-\phase{}{\psi} \\
        \phase{}{1}-\phase{}{\psi} & \phase{}{1}+\phase{}{\psi}
    \end{pmatrix}
    =\begin{pmatrix}
         \phase{'}{1} &         0       \\
               0      & \phase{'}{\psi}
     \end{pmatrix}, 
\end{equation*}
which implies $\phase{}{1}=\phase{}{\psi}=\phase{'}{1}=\phase{'}{\psi}\eqqcolon\mathrm{e}^{\mathrm{i}\eta}$. Therefore $U$ expressed in the basis~$\mathcal{B}_\cC$ is trivial up to a global phase:
\begin{equation*}
    \bU
    =\mathrm{e}^{\mathrm{i}
                 \eta}
     I. 
\end{equation*}

Suppose instead that $\mathbf{\Pi'}=\begin{pmatrix} 0& 1 \\ 1 & 0 \end{pmatrix}$. The consistency relation \eqref{eq:ising4punctureconstraint1} then becomes 
\begin{equation*}
    \frac{1}
         {2}
    \begin{pmatrix}
        \phase{}{1}+\phase{}{\psi} & \phase{}{1}-\phase{}{\psi} \\
        \phase{}{1}-\phase{}{\psi} & \phase{}{1}+\phase{}{\psi}
    \end{pmatrix}
    =\begin{pmatrix}
               0      & \phase{'}{\psi} \\
         \phase{'}{1} &        0
     \end{pmatrix}, 
\end{equation*}
which implies $\phase{}{1}=-\phase{}{\psi}$ and $\phase{'}{1}=\phase{'}{\psi}=\phase{}{1}$. Setting $\mathrm{e}^{\mathrm{i}\eta}\coloneqq\phase{}{1}$, implies that $U$ expressed in the basis~$\mathcal{B}_\cC$ is given by
\begin{equation*}
    \bU
    =\mathrm{e}^{\mathrm{i}
                 \eta}
     \begin{pmatrix}
         1 &  0 \\
         0 & -1
     \end{pmatrix}.
\end{equation*}

These two solutions of the consistency relation, for the case $\bPi=I$, now determine the only two functions of the set
\begin{equation*}
    \Fset{\fourdot{\sigma}{\sigma}{\sigma}{\sigma}}{\four{\sigma}{\sigma}{\sigma}{\sigma}{\id(\cdot)}}
    =\set{(f(1),
           f(\psi))}
    =\set{(0,
           0),
          (0,
           \pi)}.
\end{equation*}

    \item For $\mathbf{\Pi}=\begin{pmatrix} 0&1\\1&0\end{pmatrix}$, corresponding to the transposition $(\psi,1)$, we get
        \begin{equation}
            F
            \bPi
            \bD
            F^{-1}
            =\frac{1}
                  {2}
             \begin{pmatrix}
                 \phase{}{1}+\phase{}{\psi}  & \phase{}{1}-\phase{}{\psi} \\
                 -\phase{}{1}+\phase{}{\psi} & -\phase{}{1}-\phase{}{\psi}
             \end{pmatrix}
            =\bPi'
             \bD'.
            \label{eq:ising4punctureconstraint2}
        \end{equation}
        By taking $\bPi'=I$, this becomes
        \begin{equation*}
            \frac{1}
                 {2}
            \begin{pmatrix}
                \phase{}{1}+\phase{}{\psi}  & \phase{}{1}-\phase{}{\psi} \\
                -\phase{}{1}+\phase{}{\psi} & -\phase{}{1}-\phase{}{\psi}
            \end{pmatrix}
            =\begin{pmatrix}
                 \phase{'}{1} &        0        \\
                       0      & \phase{'}{\psi}
             \end{pmatrix},
        \end{equation*}
        which implies $\phase{}{1}=\phase{}{\psi}=\phase{'}{1}=-\phase{'}{\psi}$.
        Letting $\mathrm{e}^{\mathrm{i}\eta}\coloneqq\phase{}{1}$ allows $U$ to be expressed in the basis $\mathcal{B}_\mathcal{C}$ by
        \begin{equation*}
            \bU
            =\mathrm{e}^{\mathrm{i}
                         \eta}
             \begin{pmatrix}
                 0 & 1 \\
                 1 & 0
             \end{pmatrix}.
        \end{equation*}
        Instead, suppose now that $\mathbf{\Pi'}=\begin{pmatrix} 0& 1 \\ 1 & 0 \end{pmatrix}$. Then the consistency relation~\eqref{eq:ising4punctureconstraint2} is of the form
        \begin{equation*}
            \frac{1}
                 {2}
            \begin{pmatrix}
                \phase{}{1}+\phase{}{\psi}  & \phase{}{1}-\phase{}{\psi} \\
                -\phase{}{1}+\phase{}{\psi} & -\phase{}{1}-\phase{}{\psi}
            \end{pmatrix}
            =\begin{pmatrix}
                       0      & \phase{'}{\psi} \\
                 \phase{'}{1} &        0
             \end{pmatrix},
        \end{equation*}
        implying that $\phase{}{1}=-\phase{}{\psi}=-\phase{'}{1}=\phase{'}{\psi}$. Let $\mathrm{e}^{\mathrm{i}\eta}\coloneqq\phase{}{1}$, then this shows that $U$ expressed in the basis $\mathcal{B}_\mathcal{C}$ is given by
        \begin{align*}
            \bU
            =\mathrm{e}^{\mathrm{i}
                         \eta}
             \begin{pmatrix}
                 0 & -1 \\
                 1 &  0
             \end{pmatrix}. 
        \end{align*}
        Furthermore, these two solutions completely determine the relevant set of functions (which happens to be the same as the previous case for $\bPi=I$):
        \begin{equation*}
            \Fset{\fourdot{\sigma}{\sigma}{\sigma}{\sigma}}{\four{\sigma}{\sigma}{\sigma}{\sigma}{(\psi,1)(\cdot)}}
            =\set{(f(1),
                   f(\psi))}
            =\set{(0,
                   0),
                  (0,
                   \pi)}.
        \end{equation*}
\end{enumerate}
By denoting the single qubit (logical) Pauli group as
\begin{equation*}
    \mathcal{P}
    \coloneqq
     \setcomprehension[\Big]{\phaseGlobal
                             \begin{pmatrix}
                                 1 & 0 \\
                                 0 & 1
                             \end{pmatrix},
                             \phaseGlobal
                             \begin{pmatrix}
                                 1 &  0 \\
                                 0 & -1
                             \end{pmatrix},
                             \phaseGlobal
                             \begin{pmatrix}
                                 0 & 1 \\
                                 1 & 0
                             \end{pmatrix},
                             \phaseGlobal
                             \begin{pmatrix}
                                      0     & -\mathrm{i} \\
                                 \mathrm{i} &      0
                             \end{pmatrix}}
                            {\phaseGlobal
                             \in
                              \unitaryGroup{1}},
\end{equation*}
these results can be summarized as follows: If $U$ is a locality-preserving automorphism of the fusion space $\cH_{S^2(\sigma^4)}$ of the $4$-punctured sphere, then $U$ expressed in the basis $\mathcal{B}_\mathcal{C}$ is in $\mathcal{P}$.

\subsubsection{On the $M$-punctured sphere}
Let $M\geq 4$ and consider the $M=N+3$-punctured sphere $S^2(\sigma^M)$ and corresponding space $\cH_{S^2(\sigma^M)}$. For the `standard' DAP-decomposition $\cC$ of $S^2(\sigma^M)$, a consistent labeling $\mathsf{L}(\cC)$ corresponds to a sequence $\bigl(\labeling{\ell}(C_1),\dots,\labeling{\ell}(C_N)\bigr)\eqqcolon(x_1,\dots,x_N)\eqqcolon\labelSequenceSphere{x}$. It is readily observed that $\dim \cH_{S^2(\sigma^M)}=0$ if $M$ is odd, as there are no consistent labelings in this case.

 Therefore, in what follows we will restrict our discussion to the $M=N+3$-punctured sphere where $N$ is any odd positive integer. In this case, any consistent labeling $\labeling{\ell}\in\mathsf{L}(\cC)$ yields a sequence $(x_1,\dots, x_N)$ where $x_i\in\{1,\psi\}$ for odd $i$ and $x_i=\sigma$ is fixed for even $i$. Actually any such labeling of this form is consistent, giving an isomorphism
 defined in terms of orthonormal basis elements by
\begin{align}
\begin{matrix}
W:\cH_{S^2(\sigma^{N+3})} &\rightarrow & (\mathbb{C}^2)^{(N+1)/2}\\
\ket{\labelSequenceSphere{x}}&\mapsto & \ket{x_1}\otimes\ket{x_3}\otimes\dots\otimes\ket{x_N}\ . 
\end{matrix}\label{eq:isingisomorphism}
\end{align}

\begin{lemma}\label{lem:connectednessising}
Consider the `standard' basis of the $M$-punctured sphere $S^2(\sigma^M)$, where $M\geq 4$ is even. Then there is only one equivalence class under the relation~$\sim$.
Furthermore, the set of braids~$\{\sigma_{j}\}_{j=1}^{M-1}$ generates the relation~$\sim$. 
\end{lemma}
\begin{proof}
If two fusion-consistent labelings~$\labelSequenceSphere{x}$, $\labelSequenceSphere{x}'$ differ only in location~$2j+1$, they can be connected by $\sigma_{2j+1}$: the relevant braid matrix is
\begin{equation*}
    B(\sigma,
      \sigma)
    =\frac{\mathrm{e}^{-
                       3
                       \pi
                       \mathrm{i}/
                       8}}
          {\sqrt{2}}
     \begin{pmatrix}
         \mathrm{i} &      1     \\
              1     & \mathrm{i}
     \end{pmatrix}.
\end{equation*}
We have $\labelSequenceSphere{x}\connected{\sigma_{2j+1}}\labelSequenceSphere{x}'$, and it follows that there is only one equivalence class under~$\sim$.
\end{proof}

Now consider a locality-preserving automorphism~$U$ of $\cH_{S^2(\sigma^{N+3})}$ and its associated family~$\perms=\{\permLoop{C_j}\}$ of permutations. Because only sequences~$\labelSequenceSphere{x}$ with $x_{2j}=\sigma$ for all $j$ are fusion-consistent, and~$\perms$ is a permutation on~$\mathsf{L}(\cC)$,  we conclude that $\permLoop{C_{2j}}(\sigma)=\sigma$ for all $j$.  In other words, we can essentially ignore labels carrying even indices. For odd indices, only labels $x_{2j+1}\in\{1,\psi\}$ are allowed, which means that $\permLoop{C_{2j+1}}\in \{\id, (\psi,1)\}$ either 
leaves the label invariant or interchanges $\psi$ and $1$. 
In conclusion, $\perms=\{\permLoop{C_j}\}_{j=1}^N$
are of the form $\permLoop{C_j}\in\{\id,(\psi,1)\}$ for odd $j$, and $\permLoop{C_j}=\id$ for even $j$.

For odd $j=2k+1$, we obtain the constraint
\begin{align*}
\phaseArg{}{\labelSequenceSphere{x}}=\eta_{2k+1}(x_1,\dots,\widehat{x_{2k+1}},\dots,x_N)+f_{2k+1}(x_{2k+1})\qquad\textrm{ for }k=0,\ldots, (N-1)/2
\end{align*}
where $f_{2k+1}\in \Fset{\fourdot{\sigma}{\sigma}{\sigma}{\sigma}}{\four{\sigma}{\sigma}{\sigma}{\sigma}{\permLoop{C_{2k+1}}(\cdot)}}$ given that for even labels $\permLoop{C_{2m}}(x_{2m})=x_{2m}=\sigma$.
Let us write
\begin{align}
\phaseArg{}{\labelSequenceSphere{x}}&=\eta(\labelSequenceSphere{x}) +\sum_{m=0}^{(N+1)/2}f_{2m+1}(x_{2m+1})\ \label{eq:isingphases}
\end{align}
and show that $\eta(\labelSequenceSphere{x})=\eta$ is actually independent of the labeling~$\labelSequenceSphere{x}$. Indeed, we can write
\begin{align*}
\eta(\labelSequenceSphere{x})&=\bigl(\phaseArg{}{\labelSequenceSphere{x}}-f_{2k+1}(x_{2k+1})\bigr)-\sum_{m, m\neq k}^{(N+1)/2}f_{2m+1}(x_{2m+1})\\
&=\eta_{2k+1}(x_1,\ldots,\widehat{x_{2k+1}},\ldots,x_N)-\sum_{m, m\neq k}^{(N+1)/2}f_{2m+1}(x_{2m+1})\ 
\end{align*}
Since this holds for all~$k$, we conclude that
$\eta(\labelSequenceSphere{x})=\eta(\widehat{x_1},x_2,\widehat{x_3},x_4,\ldots)$ is a function of the even entries only. 
But the latter are all fixed as $x_{2m}=\sigma$, hence $\eta(\labelSequenceSphere{x})=\eta$ is simply a global phase.

We can now combine these results into a general statement concerning locality-preserving automorphisms of the $M$-punctured sphere $S^2(\sigma^M)$. 
Again, since $\dim\cH_{S^2(\sigma^M)}=0$ for odd $M$ and $\dim\cH_{S^2(\sigma^2)}=1$, we are only concerned with the cases where $M=N+3\geq4$ is even. 
Let $\{\ket{\labelSequenceSphere{x}}\}_{\labelSequenceSphere{x}\in\mathsf{L}(\cC)}$ be a basis of $\cH_{S^2(\sigma^M)}$. Then such an automorphism must act on $\cH_{S^2(\sigma^M)}$ as
\begin{align*}
U\ket{\labelSequenceSphere{x}}=\phaseExp{}{\labelSequenceSphere{x}} \ket{\perms(\labelSequenceSphere{x})}, \qquad \text{where} \qquad \phaseArg{}{\labelSequenceSphere{x}}&=\eta+\sum_{m=0}^{(N+1)/2}f_{2m+1}(x_{2m+1})
\end{align*}
and
\begin{align*}
f_{2k+1}\in \Fset{\fourdot{\sigma}{\sigma}{\sigma}{\sigma}}{\four{\sigma}{\sigma}{\sigma}{\sigma}{\permLoop{C_{2k+1}}(\cdot)}}
=\set[\big]{\bigl(f(1),f(\psi)\bigr)}
=\set{(0,
       0),
      (0,
       \pi)}.
\end{align*}
More explicitly, we have
\begin{equation*}
    U
    \ket{\labelSequenceSphere{x}}
    =e^{i
        \eta}
     \left(\prod_{m
                  =1}^{(N+
                        1)/
                       2}
           e^{i
              f_{2
                 m+
                 1}(x_{2
                       m+
                       1})}\right)
     \ket{\permLoop{C_1}(x_1),
          x_2,
          \permLoop{C_3}(x_3),
          x_4,
          \dots,
          \permLoop{C_N}(x_N)}.
\end{equation*}
In particular, under the isomorphism~\eqref{eq:isingisomorphism}, we get
\begin{align*}
WU W^{-1}&=e^{i\eta} \bigotimes_{m=1}^{(N+1)/2} U_m\qquad\textrm{ where }\qquad U_m \ket{a}=e^{if_{2m-1}(a)}\ket{\permLoop{C_{2m-1}}(a)}\ .
\end{align*}
From Section~\ref{sec:isingfourpuncturedcase}, we know that~$U_m$ is a single-qubit Pauli for each~$m$ up to a global phase. This
concludes the proof of Theorem~\ref{thm:ising}. 

\section{Abelian anyon models\label{sec:abelianmodels}}
Our goal in this section is to characterize  topologically protected gates
in general abelian anyon models. For simplicity, we will restrict our attention to closed $2$-manifolds~$\Sigma$ (see Fig.~\ref{fig:TopologicalLoops}). We have seen in Lemma~\ref{thm:permutingprojectors} that in an arbitrary anyon model, protected gates permute the idempotents along closed loops. In this section we show that for the case of abelian anyon models, the protected gates can only permute the labels of string operators along closed loops (up to phases), which refines Lemma~\ref{thm:permutingprojectors} for abelian models. To formalize this notion, we introduce the generalized Pauli and Clifford groups in Section~\ref{sec:generalizedpauligroup}. The main result of this Section,  can then be stated as follows:

\begin{theorem}\label{thm:homologypreserving}
For an abelian anyon model,
any locality-preserving unitary
automorphism~$U$ acting on~$\cH_{\Sigma}$ has logical action
$\logical{U}\in\mathsf{Clifford}^{\star}_{\Sigma}$. 
\end{theorem}

For abelian anyon models, the set~$\labels$ of particles is an abelian group and the fusion rules (i.e., the Verlinde algebra~\eqref{eq:verlindealgebrarelation}) are given by the group product, $N^{c}_{ab}=1$ if and only if~$c=ab$ and~$N^{c}_{ab}=0$ otherwise.
In other words,  any two particles~$a$ and~$b$ fuse to a unique particle~$c=a b$, and the identity element~$1\in\labels$ is the only particle satisfying $1 a=a$ for all $a\in\labels$. Another requirement is that the $S$ matrix is composed entirely of phases (divided by the quantum dimension $\mathcal{D}$), and $S_{1a}=S_{a1}= 1/\mathcal{D}$ for all $a \in \labels$. Furthermore, the involution~$a\mapsto\dual{a}$ defining the antiparticle associated to~$a\in \labels$ is simple the inverse $\dual{a}=a^{-1}$ with respect to the group multiplication. Note that,
by the fundamental theorem of finitely generated abelian groups, the group~$\labels$ is isomorphic to $\bbZ_{N_1}\times\bbZ_{N_2}\times\dots\times\bbZ_{N_r}$ for some prime powers~$N_j$.  The number $N=\mathsf{lcm}(N_1,\ldots,N_r)$ will play an important role in the following, determining e.g., the order of a protected gate.

It is well known that for abelian anyons $a$ and $b$, and two inequivalent loops $C$ and $C'$ whose intersection number is~1 in the manifold $\Sigma$ the relation 
\begin{align}
\logical{F_{\dual{b}}(C')}\logical{F_{\dual{a}}(C)} \logical{F_b(C')} \logical{F_a(C)} = \mathcal{D} S_{ab} \logical{\id}
\label{eq:SmatrixFromStringOperators}
\end{align}
holds.
As we will see, this provides an additional constraint on the logical action of a protected gate~$U$. The following consistency condition must hold:
\begin{lemma}\label{lem:conjugateloopconsistency}
Let $C$ and $C'$ be two loops on~$\Sigma$ which intersect once.
Consider the action of a locality preserving unitary automorphism of the code on the string operators on $C$ and $C'$, that is
\begin{equation}
    \rho_U(\logical{F_b(C)})
    =\sum_{d}\Lambda_{b,d}\logical{F_{d}(C)},
    \qquad
    \rho_U(\logical{F_b(C')})
    =\sum_{d}\Lambda'_{b,d}\logical{F_{d}(C')}. 
    \label{eq:rhoufbmappingprop3}
\end{equation}
Then the matrices $\Lambda$ and $\Lambda'$ must satisfy the following consistency condition
\begin{align}
 \Lambda_{a,c}~\Lambda'_{b,d} \left( S_{cd} - S_{ab} \right) =0 ~~~\forall a, b, c, d \in \labels. \label{eq:AbelianConditionOnProjectors}
\end{align}
\end{lemma}
\begin{proof}
Since in an abelian anyon model every string operator~$[F_a(C)]$ is unitary the relation~\eqref{eq:SmatrixFromStringOperators} is equivalent to the commutation relation
\begin{equation*}
    \logical{F_b(C')}
    \logical{F_a(C)}
    =\mathcal{D}
     S_{a
        b}
     \logical{F_a(C)}
     \logical{F_b(C')}.
\end{equation*}
Conjugating this by~$U$ and rearranging terms yields
\begin{eqnarray}
0= \sum_{c,d} \Lambda_{a,c}~\Lambda'_{b,d} \left( \mathcal{D} S_{cd} - \mathcal{D} S_{ab}\right) \logical{F_{c}(C)}\logical{F_{d}(C')}.
\end{eqnarray}
The claim follows from linear independence of the logical operators $\logical{F_{c}(C)}\logical{F_{d}(C')} $.
\end{proof}

Invoking our previous result of Lemma~\ref{thm:permutingprojectors}, the following lemma is implied:
\begin{lemma}\label{lem:abelianstringpermutation}
The anyon labels of \emph{string} operators along the loop are permuted by $U$
\begin{align}
\Lambda_{b,d} = e^{i \phi_{b}} \delta_{d,\tilde{\pi}(b)},
\end{align}
for some phase $\phi_{b}$, and where $\tilde{\pi}$ is a permutation of anyon labels. 
\end{lemma}
\begin{proof}
Recall from~(\ref{eq:rhoufbmappingprop}) that
\begin{equation}
\Lambda_{b,d} = \sum_a\frac{S_{b,a}}{S_{1,a}} S_{1,\permLoop{C}(a)}\conjugate{S_{d,\permLoop{C}(a)}} = \sum_a S_{b,a} \conjugate{S_{d,\permLoop{C}(a)}}, 
\label{eq:rhoufbmappingprop2}
\end{equation}
where $\permLoop{C}$ is the permutation of the central idempotents associated with loop~$C$, where the second equality holds for abelian anyons. An analogous equation holds for loop $C'$. Now sum over all $a \in \labels$ in~\eqref{eq:AbelianConditionOnProjectors}. To evaluate the sum, we require $ \sum_a\Lambda_{a,c}$ and $ \sum_a\Lambda_{a,c} S_{ab}$. Firstly,
\begin{equation*}
    \sum_a
    \Lambda_{a,
             c}
    =\sum_{a,
           g}
     S_{a,
        g}
     \conjugate{S_{c,
                   \permLoop{C}(g)}}
    =\mathcal{D}
     \sum_g
     \delta_{g,
             1}
     \conjugate{S_{c,
                   \permLoop{C}(g)}}
    =\mathcal{D}
     \conjugate{S_{c,
                   \permLoop{C}(1)}},
\end{equation*}
where we used unitarity of the $S$-matrix, $\delta_{1z} =\sum_x \conjugate{S_{x1}} S_{xz} = \sum_x S_{xz}/\mathcal{D}$. Secondly,
\begin{equation*}
    \sum_a
    \Lambda_{a,
             c}
    S_{a
       b}
    =\sum_{a,
           g}
     S_{a,
        g}
     \conjugate{S_{c,
                   \permLoop{C}(g)}}
     S_{ab}
    =\sum_{a,
           g}
     S_{a,
        g}
     \conjugate{S_{c,
                   \permLoop{C}(g)}}
     \conjugate{S_{a
                   \dual{b}}}
    =\sum_g
     \delta_{g,
             \dual{b}}
     \conjugate{S_{c,
                   \permLoop{C}(g)}}
    =\conjugate{S_{c,
                   \permLoop{C}(\dual{b})}}.
\end{equation*}
Therefore~\eqref{eq:AbelianConditionOnProjectors} implies
\begin{equation}
    \bigl(\mathcal{D}
          S_{c
             d}
          \conjugate{S_{c,
                        \permLoop{C}(1)}}-
          \conjugate{S_{c,
                        \permLoop{C}(\dual{b})}}\bigr)
    \Lambda_{b,
             d}'
    =0
    \qquad
    \forall
     b,
     c,
     d
    \in
     \labels.
    \label{eq:AbelianConditionOnProjectors2}
\end{equation}
For any $B \in \labels$, there must exist at least one anyon $D \in \labels$ such that $\Lambda'_{B,D} \neq 0$. Then
\begin{equation}
    \mathcal{D}
    S_{c
       D}
    \conjugate{S_{c,
                  \permLoop{C}(1)}}-
    \conjugate{S_{c,
                  \permLoop{C}(\dual{B})}}
    =0
    \qquad
    \forall
     c
    \in
     \labels.
\end{equation}
For each $D'\neq D$, there must be some~$C \in \labels$ such that $S_{CD} \neq S_{CD'}$.
Therefore substituting into~\eqref{eq:AbelianConditionOnProjectors2} the values $b=B,c=C$ and $d=D'$, the term in brackets must be non-zero, implying  $\Lambda'_{B,D'} = 0$ for all $D'\neq D$. Unitarity of $U$ yields the claim for loop $C'$. 
\end{proof}

\subsection{The generalized Pauli and Clifford groups\label{sec:generalizedpauligroup}}
Consider the case where 
 $\labels=\mathbb{Z}_{N_1}\times\dots\times\mathbb{Z}_{N_r}$ and set~$N=\mathsf{lcm}(N_1,\ldots,N_r)$. We define the following group associated with the surface~$\Sigma$. 

  \begin{definition}[Pauli group]\label{def:pauliclifford}
Consider a genus-$g$ surface~$\Sigma$ and let $\cG=\{C_j\}_{j=1}^{3g-1}$ be the loops
associated with generators of the mapping class group as in Fig.~\ref{fig:TopologicalLoops}.  The {\em Pauli group~$\pauli_{\Sigma}$ associated with~$\Sigma$} is
 \begin{align*}
 \pauli_{\Sigma}\coloneqq\big\langle \left\{\phaseGlobal\logical{F_a(C)}\ \big|\ \phaseGlobal\in \langle e^{2\pi i/N}\rangle, a\in\labels, C\in\cG \right\}\big\rangle\ ,
 \end{align*} 
i.e., the set of logical operators 
generated by taking products of string-operators associated with~$\cG$,  where
$\langle e^{2\pi i/N}\rangle $ is the subgroup of~$\unitaryGroup{1}$ consisting of $N$-th roots of unity. 
\end{definition}

According to Eq.~\eqref{eq:SmatrixFromStringOperators}, we can always reorder and write each element $P\in\pauli_{\Sigma}$ in the standard form
\begin{align*}
P= \phaseGlobal\logical{F_{a_1}(C_1)}\cdots \logical{F_{a_{3g-1}}(C_{3g-1})}\qquad\textrm{ for some }\phaseGlobal\in \langle e^{2\pi i/N}\rangle,\  a_j\in\labels\  .
\end{align*}
This  shows that the group~$\pauli_{\Sigma}$ is finite. Furthermore, since
$a^N=1$  for every $a\in\labels$, we conclude that $P^N=\phaseGlobal\logical{\id}$ is proportional to the identity up to a phase $\phaseGlobal\in\langle e^{2\pi i/N}\rangle$. That is, every element of the Pauli group~$\pauli_\Sigma$ has order dividing~$N$.

Given this definition, we can proceed to give the definition of the Clifford group.
\begin{definition}[Clifford group]
 The 
  {\em Clifford group associated with~$\Sigma$}
  is the group of logical unitaries
  \begin{align*}
  \cliff_\Sigma\coloneqq\{\phaseGlobal\logical{U}\ |\ \logical{U}\pauli_\Sigma\logical{U}^{-1}\subset\pauli_\Sigma,\phaseGlobal\in\langle e^{2\pi i/N}\rangle  \}\ .
  \end{align*}
In this definition,~$\logical{U}$ is any logical unitary on the code space. 
  \end{definition}
We can define a `homology-preserving subgroup' of $\cliff_\Sigma$. To do so, we first introduce
the following subgroup of~$\pauli_{\Sigma}$ associated with a loop on~$\Sigma$.
\begin{definition}[Restricted Pauli group]
Let $C\in\cG$ be a single closed loop. We set
\begin{align*}
 \pauli_{\Sigma}(C)\coloneqq\big\langle \left\{\phaseGlobal\logical{F_a(C)}\ \big|\ \phaseGlobal\in \langle e^{2\pi i/N}\rangle, a\in\labels\right\}\big\rangle\ ,
\end{align*}
i.e., the subgroup generated by  string-operators associated with the loop~$C$. 
\end{definition}
It is straightforward to check that for any $C\in\cG$, the subgroup $\pauli_{\Sigma_g}(C)\subset \pauli_{\Sigma_g}$  is normal; furthermore, any $P\in\pauli_{\Sigma_g}(C)$ has the
simple form of a product $P=\phaseGlobal[F_{a_1}(C)]\cdots [F_{a_r}(C)]$. 

Given this definition, we can define a subgroup of Clifford group elements as follows:
\begin{definition}[Homology-preserving Clifford group]
 The 
  {\em homology-preserving Clifford group associated with~$\Sigma$}
  is the subgroup
  \begin{align*}
   \cliff^\star_\Sigma\coloneqq\{\phaseGlobal\logical{U}\ |\ \logical{U}\pauli_\Sigma(C)\logical{U}^{-1}\subset\pauli_\Sigma(C)\textrm{ for all }C\in\cG,\ \phaseGlobal\in\langle e^{2\pi i/N}\rangle  \}\ .
  \end{align*} 
\end{definition}
Note that
this is a proper subgroup, i.e., $\cliff^\star_\Sigma \subsetneq    \cliff_\Sigma$, as can be seen from the following example.
\begin{example}
Consider for example Kitaev's $D(\mathbb{Z}_2)$-code on a torus~$\Sigma_2$ (cf.~Example~\ref{ex:kitaevtoric}). In this case, there are two inequivalent homologically non-trivial cycles $C_1$ and $C_2$. In the language of stabilizer codes, the logical operators $(\bar{X}_1,\bar{Z}_1)=(F_e(C_1),F_m(C_2))$  and $(\bar{X}_2,\bar{Z}_2)=(F_e(C_2),F_m(C_1))$ are often referred to as the logical Pauli operators associated with the first and second logical qubit, respectively. Consider the logical Hadamard $\bar{H}_1$ on the first qubit, which acts as
\begin{align*}
\bar{H}_1 \bar{X}_1\bar{H}_1 ^\dagger=\bar{Z}_1\qquad\textrm{ and }\qquad  \bar{H}_1 \bar{Z}_1\bar{H}_1 ^\dagger=\bar{X}_1
\end{align*}
but leaves $\bar{X}_2$ and $\bar{Z}_2$ invariant.  Then $\bar{H}_1$ belongs to the Clifford group,    $\bar{H}_1\in \cliff_\Sigma$. However, $\bar{H}_1\not\in\cliff^\star_\Sigma$ because $\bar{X}_1$ and $\bar{Z}_1$ belong to different homology classes (specified by $C_1$ and $C_2$, respectively). 

\end{example}

In the following, we make use of the existence of a loop $C'$ which intersects with a given loop~$C$ exactly once. Note that this is not necessarily given, but works  in the special case where~$C$ is one of the~$3g-1$ curves~$\{C_j\}_{j=1}^{3g-1}$ associated with the generators of the mapping class group of the genus-$g$ surface~$\Sigma_g$  (cf.~Fig.~1). We are now ready to prove Theorem~\ref{thm:homologypreserving}, i.e., that a protected gate~$U$ has logical action~$\logical{U}\in\mathsf{Clifford}^{\star}_{\Sigma}$.

\begin{proof}
 By Lemma~\ref{lem:abelianstringpermutation}, we 
have that
$\sum_{c}\Lambda_{a,c} \logical{F_c(C)}=\phaseGlobal \logical{F_b(C)}$
for some $\phaseGlobal\in\unitaryGroup{1}$ and $b\in\mathbb{A}$. It remains to show that~$\phaseGlobal$ is an $N$-th root of unity. 
We have
\begin{equation*}
    \phaseGlobal^N
    \logical{\id}
    =\phaseGlobal^N
     \logical{F_b(C)^N}
    =\logical{\phaseGlobal
              F_b(C)}^N
    =\logical{U}
     \logical{F_a(C)}^N
     \logical{U^\dagger}
    =\logical{\id}
\end{equation*}
because the string operators~$F_a(C)$ have order 
 dividing~$N$, thus we must have $\phaseGlobal^N=1$.
 Because $a$ and $C$ were arbitrary, this concludes  the proof that $\logical{U}\in\cliff^\star_\Sigma$. 
\end{proof}

\section*{Acknowledgements}
RK and SS gratefully acknowledge support by NSERC, and MB, FP, and JP gratefully acknowledge support by NSF grants PHY-0803371 and PHY-1125565, NSA/ARO grant W911NF-09-1-0442, and AFOSR/DARPA grant FA8750-12-2-0308. RK is supported by the Technische Universit\"at M\"unchen -- Institute for Advanced Study, funded by the German Excellence Initiative and the European Union Seventh Framework Programme under grant agreement no.~291763. OB gratefully acknowledges support by the ERC (TAQ).
The Institute for Quantum Information and Matter (IQIM) is an NSF Physics Frontiers Center with support by the Gordon and Betty Moore Foundation. RK and SS thank the IQIM for their hospitality. We thank Jeongwan Haah, Olivier Landon-Cardinal and Beni Yoshida for helpful discussions, and the referees and editors for their comments.

\appendix

\section{Density on a subspace and protected gates~\label{sec:appendixsuperselection}} 

\begin{lemma}\label{lem:generalizeddensity}
Let $\cH_0$ be an invariant subspace under the mapping class group representation, and suppose the action of $\mcg_\Sigma$ is dense in the projective unitary group $\mathsf{PU}(\cH_0)$.
Let $\cH_1$ be the orthogonal complement of $\cH_0$ in $\cH_\Sigma$. Assume that
the decomposition $\cH_0\oplus\cH_1$ stems from the gluing postulate in the sense that
$\cH_j=\bigoplus_{\vec{a}\in\Lambda_j}\cH_{\Sigma'(\vec{a})}$ for $j=0,1$, 
where $\Lambda_0,\Lambda_1$ are disjoint 
set of labelings 
of the boundary components of the surface~$\Sigma'$ obtained by cutting~$\Sigma$ along  a family~$\vec{C}$  of pairwise non-intersecting curves. 
If $\dim \cH_1<\dim\cH_0$ (or a similar assumption), then any protected gate $U$ leaves $\cH_0$ invariant and acts as a global phase on it.
\end{lemma}
\begin{proof}
Extending~$\vec{C}$ to a DAP-decomposition~$\cC$, the unitary~$U$ expressed  in the (suitably ordered) basis~$\cB_\cC$ takes the form
\begin{align*}
\bU=\left(
\begin{matrix}
\bU_{00} & \bU_{01}\\
\bU_{10} & \bU_{11}
\end{matrix}
\right)\ ,
\end{align*}
where ${\bf U}_{jk}$ describes the operator $P_{\cH_j}UP_{\cH_k}$ obtained by projecting the domain and image of $U$ to $\cH_k$ and $\cH_j$, respectively. 

Consider the Schur decomposition
 $\bU_{00}=\bW_{00}\Gamma\bW_{00}^\dagger$
 of $\bU_{00}$, i.e., $\bW_{00}$ is a unitary matrix and $\Gamma$ is upper triangular.  There are different cases to consider:
 \begin{enumerate}[(i)]
 \item
 If $\Gamma$ is diagonal with a single eigenvalue~$\lambda$, then
 \begin{align*}
\bU=\left(
\begin{matrix}
\lambda I & \bU_{01}\\
\bU_{10} & \bU_{11}
\end{matrix}
\right)\ . 
  \end{align*} 
  Assume for the sake of contradiction that $\lambda=0$. Writing $d_j=\dim \cH_{j}$, the $d_1\times d_0$-matrix $\bU_{10}$, must have exactly $d_0$ non-zero values, each in a different row because $\bU\in\Delta$. This is only possible if $d_1>d_0$, contradicting our assumption. 

  We conclude that $\lambda\neq 0$. But then the condition $\bU\in\Delta$ requires that $\lambda\in\unitaryGroup{1}$ and $\bU_{01}=\bU_{10}=0$ (since we cannot have more than one non-zero entry in each column or row).
 
  \item
  $\Gamma$ has a non-zero off-diagonal element~$\Gamma_{j,k}$, $j<k$. We will show that this is not consistent with the fact that $U$ is a protected gate (i.e., leads to a contradiction). By reordering basis elements of~$\cB_{\cC}$, we can assume  without loss of generality that $\Gamma_{1,2}\neq 0$. 
By using, e.g., Solovay-Kitaev on~$\cH_0$, we find a 
product $\tilde{\bV}=\bV(\vartheta_1)\cdots\bV(\vartheta_m)$ of   images of mapping class group elements approximating
  $\bV=\bW_{00}^\dagger\oplus \bW_{11}$, where $\bW_{11}$ is an arbitrary unitary on~$\cH_1$. 

Consider the matrix~$\bV\bU\bV^\dagger$. 
We have 
$(\bV\bU\bV^\dagger)_{j,k}=\Gamma_{j,k}$ for $j,k=1,\ldots,\dim\cH_{0}$. In particular, 
$(\bV\bU\bV^\dagger)_{1,2}\neq 0$
and $(\bV\bU\bV^\dagger)_{2,1}=0$.

We claim that we must have
$(\bV\bU\bV^\dagger)_{1,1}=(\bV\bU\bV^\dagger)_{2,2}=0$. To show this, assume for the sake of contradiction that one of these diagonal entries is non-zero. Then~$\bV\bU\bV^\dagger\not\in \Delta$ since it has two non-zero entries in the same row or column.  But this implies  $\tilde{\bV}\bU\tilde{\bV^\dagger}\not\in\Delta$ since~$\tilde{\bV}\bU\tilde{\bV}^\dagger\approx\bV\bU\bV^\dagger$, a contradiction to the fact that $\bU\in\Delta_{\vartheta_1\cdots\vartheta_m}$. 

Now let $X_{j,k}=(\bV\bU\bV^\dagger)_{j,k}$ for $j,k\in\{1,2\}$ be the principal minor $2\times 2$~submatrix. We have established that its only non-zero entry is~$X_{1,2}$.
Using the Hadamard matrix~$H$, we then have 
$(HXH^\dagger)_{1,1}=X_{1,2}/2\neq 0$ and $(HXH^\dagger)_{1,2}=-X_{1,2}/2\neq 0$. Let ${\bf H}=H\oplus I_{(\dim \cH_0-2)} $. By Solovay-Kitaev, we can find a product 
$\tilde{\bV}'=\bV(\vartheta'_1)\cdots\bV(\vartheta'_\ell)$ of   images of mapping class group elements approximating
  $\bV'={\bf H}\oplus \bW'_{11}$, where $\bW'_{11}$ is an arbitrary unitary on~$\cH_1$. 
Then we have
\begin{align*}
(\bV'\bV \bU\bV^\dagger (\bV')^\dagger)_{1,1}&=X_{1,2}/2\neq 0\\
(\bV'\bV \bU\bV^\dagger (\bV')^\dagger)_{1,2}&=-X_{1,2}/2\neq 0\ ,
\end{align*}
which shows that 
$\bV'\bV \bU\bV^\dagger (\bV')^\dagger\not\in\Delta$. 
By continuity, this shows that
$\tilde{\bV}'\tilde{\bV}\bU\tilde{\bV}^\dagger(\tilde{\bV}')^\dagger\not\in\Delta$, contradicting the fact that $\bU\in\Delta_{\vartheta_1'\cdots\vartheta_\ell'\vartheta_1\cdots\vartheta_m}$.

  \item $\Gamma$ is diagonal with distinct eigenvalues: in this case we can apply the same kind of argument as in the proof of Corollary~\ref{cor:density}.
 \end{enumerate}
  \end{proof}


\begin{thebibliography}{10}

\bibitem{Atiyah89}
M.~Atiyah.
\newblock Topological quantum field theories.
\newblock {\em Inst. Hautes {\'E}tudes Sci. Publ. Math.}, 68:175--186, 1989.

\bibitem{BeigiShorWhalen11}
S.~Beigi, P.~W. Shor, and D.~Whalen.
\newblock The quantum double model with boundary: Condensations and symmetries.
\newblock {\em Communications in Mathematical Physics}, 306(3):663--694, 2011.

\bibitem{BombinTwists}
H.~Bombin.
\newblock Topological order with a twist: {I}sing anyons from an abelian model.
\newblock {\em Phys. Rev. Lett.}, 105:030403, 2010.

\bibitem{BMD:topo}
H.~Bombin and M.~A. Martin-Delgado.
\newblock Topological quantum distillation.
\newblock {\em Phys. Rev. Lett.}, 97:180501, 2006.

\bibitem{BombinDelgado08}
H.~Bombin and M.~A. Martin-Delgado.
\newblock Family of non-abelian kitaev models on a lattice: Topological
  condensation and confinement.
\newblock {\em Phys. Rev. B}, 78:115421, Sep 2008.

\bibitem{Bombin06}
H.~Bombin and M.A. Martin-Delgado.
\newblock Topological computation without braiding.
\newblock {\em Phys.Rev.Lett.}, 98:160502, 2007.

\bibitem{Bravyi2006}
S.~Bravyi, M.~Hastings, and F.~Verstraete.
\newblock {Lieb-Robinson Bounds and the Generation of Correlations and
  Topological Quantum Order}.
\newblock {\em Physical Review Letters}, 97(5):050401, July 2006.

\bibitem{BK:surface}
S.~{Bravyi} and A.~Y. {Kitaev}.
\newblock {Quantum codes on a lattice with boundary}.
\newblock 1998.
\newblock {\sf arXiv:quant-ph/9811052}.

\bibitem{BravyiKitaev2005}
S.~Bravyi and A.~Y. Kitaev.
\newblock Universal quantum computation with ideal {C}lifford gates and noisy
  ancillas.
\newblock {\em Phys. Rev. A}, 71:022316, Feb 2005.

\bibitem{BravyiKoenig2013}
S.~Bravyi and R.~K\"onig.
\newblock Classification of topologically protected gates for local stabilizer
  codes.
\newblock {\em Phys. Rev. Lett.}, 110:170503, Apr 2013.

\bibitem{Brell2014}
Courtney~G. Brell, Simon Burton, Guillaume Dauphinais, Steven~T. Flammia, and
  David Poulin.
\newblock {Thermalization, Error Correction, and Memory Lifetime for Ising
  Anyon Systems}.
\newblock {\em Physical Review X}, 4(3):031058, sep 2014.

\bibitem{Burton2015}
S.~Burton, Courtney~G. Brell, and S.~T. Flammia.
\newblock {Classical Simulation of Quantum Error Correction in a Fibonacci
  Anyon Code}.
\newblock page~5, Jun 2015.
\newblock arXiv:1506.03815.

\bibitem{Chen2010}
Xie Chen, Zheng-Cheng Gu, and Xiao-Gang Wen.
\newblock {Local unitary transformation, long-range quantum entanglement, wave
  function renormalization, and topological order}.
\newblock {\em Physical Review B}, 82(15):155138, October 2010.

\bibitem{EastinKnill2009}
B.~Eastin and E.~Knill.
\newblock Restrictions on transversal encoded quantum gate sets.
\newblock {\em Phys. Rev. Lett.}, 102:110502, Mar 2009.

\bibitem{elseetal12}
D.~V. Else, I.~Schwarz, S.~D. Bartlett, and A.~C. Doherty.
\newblock Symmetry-protected phases for measurement-based quantum computation.
\newblock {\em Phys. Rev. Lett.}, 108:240505, Jun 2012.

\bibitem{Fowler2012}
A.~G. Fowler, M.~Mariantoni, J.~M. Martinis, and A.~N. Cleland.
\newblock Surface codes: Towards practical large-scale quantum computation.
\newblock {\em Phys. Rev. A}, 86:032324, Sep 2012.

\bibitem{Fowler08}
A.~G. Fowler, A.~M. Stephens, and P.~Groszkowski.
\newblock High threshold universal quantum computation on the surface code.
\newblock {\em Phys. Rev. A}, 80:052312, 2009.

\bibitem{Freedmanetal08}
M.~Freedman, C.~Nayak, K.~Walker, and Z.~Wang.
\newblock {\em On Picture (2+1)-TQFTs}, chapter~2, pages 19--106.
\newblock August 2008.

\bibitem{freedmankitaevlarsenwang03}
M.~H. Freedman, A.~Kitaev, M.~J. Larsen, and Z.~Wang.
\newblock Topological quantum computation.
\newblock {\em Bull. Amer. Math. Soc.}, 40:31--38, 2003.

\bibitem{FreedmanKitaevWang02}
M.~H. Freedman, A.~Y. Kitaev, and Z.~Wang.
\newblock Simulation of topological field theories by quantum computers.
\newblock {\em Commun. Math. Phys.}, 227:587--603, 2002.

\bibitem{freedmanlarsenwang02}
M.~H. Freedman, M.~Larsen, and Z.~Wang.
\newblock A modular functor which is universal for quantum computation.
\newblock {\em Communications in Mathematical Physics}, 227(3):605--622, 2002.

\bibitem{Haah11}
J.~Haah.
\newblock Local stabilizer codes in three dimensions without string logical
  operators.
\newblock {\em Phys. Rev. A}, 83:042330, 2011.

\bibitem{Haah14}
J.~Haah.
\newblock An invariant of topologically ordered states under local unitary
  transformations, July 2014.
\newblock {\sf arXiv:1407.2926}.

\bibitem{Hutter2015}
A.~Hutter and J.~R. Wootton.
\newblock {Continuous error correction for Ising anyons}, Aug 2015.
\newblock arXiv:1508.04033.

\bibitem{Walker91}
Walker. K.
\newblock On {W}itten's 3-manifold invariants, 1991.
\newblock Lecture notes, {\sf http://canyon23.net/math/1991TQFTNotes.pdf}.

\bibitem{KitaevPreskill06}
A.~Kitaev and J.~Preskill.
\newblock Topological entanglement entropy.
\newblock {\em Phys. Rev. Lett.}, 96:110404, Mar 2006.

\bibitem{Kitaev03}
A.~Y. Kitaev.
\newblock Fault-tolerant quantum computation by anyons.
\newblock {\em Annals of Physics}, 303(1):2, 2003.

\bibitem{Kitaev05}
A.~Y. {Kitaev}.
\newblock {Anyons in an exactly solved model and beyond}.
\newblock {\em Ann. Phys.}, 321(1):2, 2006.

\bibitem{KitaevKong12}
A.~Y. Kitaev and L.~Kong.
\newblock Models for gapped boundaries and domain walls.
\newblock {\em Communications in Mathematical Physics}, 313(2):351--373, 2012.

\bibitem{KongWen14}
L.~Kong and X.-G. Wen.
\newblock Braided fusion categories, gravitational anomalies, and the
  mathematical framework for topological orders in any dimensions, May 2014.
\newblock {\sf arXiv:1405.5858}.

\bibitem{MacLane}
S.~M. Lane.
\newblock {\em Categories for the Working Mathematician}.
\newblock Graduate Texts in Mathematics. Springer New York, 1998.

\bibitem{LevinWen}
M.~A. Levin and X.-G. Wen.
\newblock String-net condensation: A physical mechanism for topological phases.
\newblock {\em Phys. Rev. B}, 71:045110, Jan 2005.

\bibitem{LevWenTopOrd}
M.~A. Levin and X.-G. Wen.
\newblock Detecting topological order in a ground state wave function.
\newblock {\em Phys. Rev. Lett.}, 96:110405, Mar 2006.

\bibitem{Lieb1972}
Elliott~H. Lieb and Derek~W. Robinson.
\newblock {The finite group velocity of quantum spin systems}.
\newblock {\em Communications in Mathematical Physics}, 28(3):251--257, 1972.

\bibitem{Michnicki12}
K.~P. Michnicki.
\newblock 3{D} topological quantum memory with a power-law energy barrier.
\newblock {\em Phys. Rev. Lett.}, 113:130501, Sep 2014.

\bibitem{MooreSeiberg98}
G.~Moore and N.~Seiberg.
\newblock Polynomial equations for rational conformal field theories.
\newblock {\em Physics Letters B}, 212(4):451--460, October 1998.

\bibitem{PastawskiYoshida14}
F.~Pastawski and B.~Yoshida.
\newblock Fault-tolerant logical gates in quantum error-correcting codes.
\newblock {\em Phys. Rev. A}, 91:012305, Jan 2015.

\bibitem{PedrocchiDiVincenzo15}
Fabio~L. Pedrocchi and David~P. DiVincenzo.
\newblock Majorana braiding with thermal noise.
\newblock {\em Phys. Rev. Lett.}, 115:120402, Sep 2015.

\bibitem{preskillnotes}
J.~Preskill.
\newblock Lecture notes on quantum computation, 2004.
\newblock available at {\sf
  http://www.theory.caltech.edu/people/preskill/ph229/~lecture}.

\bibitem{Segal}
G.~Segal.
\newblock {\em The definition of conformal field theory}, volume 308.
\newblock London Math. Soc. Lecture Note Ser., Cambridge University Press,
  2004.
\newblock preprint.

\bibitem{Verlinde88}
E.~Verlinde.
\newblock Fusion rules and modular transformations in 2d conformal field
  theory.
\newblock {\em Nucl. Phys. B}, 300:360--376, 1988.

\bibitem{wangbook}
Z.~Wang.
\newblock {\em Topological Quantum Computation}.
\newblock Number 112 in Regional Conference Series in Mathematics. Conference
  Board of the Mathematical Sciences, 2010.

\bibitem{Witten89}
E.~Witten.
\newblock Quantum field theory and the {J}ones polynomial.
\newblock {\em Comm. Math. Phys.}, 121(3):351--399, 1989.

\end{thebibliography}
\end{document}